\documentclass[a4paper,twoside,cleardoublepage=plain,bibliography=totoc]{scrbook}

\usepackage[a4paper]{geometry}                    

\usepackage{xurl}                                 
\usepackage[backref=page]{hyperref}               
\hypersetup{
linktocpage,                                      
colorlinks  = true,                               
urlcolor    = cyan,
linkcolor   = red,
citecolor   = blue
}
\renewcommand*{\backref}[1]{}
\renewcommand*{\backrefalt}[4]{[%
\ifcase #1 Not cited.%
  \or Cited on page~#2.%
  \else Cited on pages #2.%
\fi]}

\usepackage{datetime}                             
  \newdateformat{monthonly}{\monthname[\THEMONTH]}
\usepackage{amssymb,amsthm,amsmath}               
\usepackage{graphicx}                             
\usepackage{natbib}                               
\usepackage{floatrow}                             
\usepackage{booktabs}                             
\usepackage{paralist}                             
\usepackage{titlesec}                             
\usepackage[standardsections]{scrhack}            
\usepackage{parskip}                              

\usepackage{caption}
\usepackage{enumitem}
\usepackage{mathtools}
\usepackage{cleveref}
\usepackage{marvosym}
\usepackage{xspace}
\usepackage{multirow}
\usepackage{tikz}
\usetikzlibrary{positioning}
\usetikzlibrary{arrows.meta}
\usetikzlibrary{decorations.pathmorphing}
\tikzset{>={Latex[length=0.6em]}} 

\titleformat{\chapter}[display]
{\bfseries\huge}
{\filleft\Large\chaptertitlename~\thechapter}
{3ex}
{\titlerule\vspace{1.5ex}\filright}
[\vspace{1ex}\titlerule]

\newcommand{\AuthorName} {Timothy Horscroft}

\newcommand{\ProjectTitle} {Convergence Properties of Dynamic Processes on Graphs}

\newif\ifStandardTitle 
\StandardTitlefalse     

\newcommand{\School} {School of Computing}

\newcommand{\College} {College of Engineering, Computing and Cybernetics (CECC)}

\newcommand{\ProjectPoints} {24}

\newif\ifHonoursThesis 
\HonoursThesistrue     

\newcommand{\Semester} {S1/S2}

\newcommand{\Year} {2023}

\newcommand{\Degree} {Bachelor of Philosophy (Honours) - Science}

\newcommand{\CourseCode} {COMP1234}
\newcommand{\CourseName} {Course Name}

\newcommand{\FirstSupervisor} {Dr.\ Ahad N. Zehmakan}

\newif\ifTwoSupervisors 
\TwoSupervisorsfalse     

\newcommand{\SecondSupervisor} {%
Dr.\ Second Supervisor
\\Prof.\ Dr.\ Third Supervisor (if there is any)
\\Prof.\ Dr.\ Dr.\ Fourth Supervisor (if we need four, five, etc.)
}

\theoremstyle{definition}
\newtheorem{definition}{Definition}[section]   
\theoremstyle{plain}
\newtheorem{prop}[definition]{Proposition} 
\newtheorem{lem}[definition]{Lemma}        
\newtheorem{thm}[definition]{Theorem}      
\newtheorem{cor}[definition]{Corollary}    
\newtheorem{ass}[definition]{Assumption}
\newtheorem{myEquation}[definition]{Equation}
\newenvironment{proofsketch}{\proof}{\endproof}

\newenvironment{defn}[1]{%
    \begin{definition}#1}{%
    \end{definition}%
}
\newenvironment{eqn}[1]{
    \begin{myEquation}#1}{
    \end{myEquation}
}

\newcommand{\Erdos}{Erdős\xspace}
\newcommand{\Barabasi}{Barabási\xspace}
\newcommand{\ER}{Erdős–Rényi\xspace}
\newcommand{\BA}{Barabási–Albert\xspace}
\newcommand{\bigO}{\mathcal{O}}
\newcommand{\R}{\mathbb{R}}
\newcommand{\Z}{\mathbb{Z}}
\newcommand{\E}{\mathbb{E}}
\crefname{enumi}{type}{types} 
\newenvironment{mytable}{\par\addvspace{6pt plus2pt minus1pt}\captionsetup{type=table}}{\par}
\newenvironment{myfigure}{\par\addvspace{6pt plus2pt minus1pt}\captionsetup{type=figure}}{\par}
\def\picgap{2em} 
\tikzset{
    main/.style = {draw, thick, circle, minimum size = 1em},
    redd/.style = {ultra thick, red},
    bluu/.style = {ultra thick, blue},
    grnn/.style = {ultra thick, green!70!black},
    yelw/.style = {ultra thick, yellow!70!black},
    unredd/.style = {thin, black},
    scc/.style = {teal, very thick, fill=lightgray, fill opacity=0.3},
    sccred/.style = {red, very thick, fill=lightgray!60!red, fill opacity=0.3},
    bfs/.style = {green!70!black, very thick, fill=lightgray, fill opacity=0.3},
    scclabel/.style = {teal},
    bfslabel/.style = {green!70!black}
}                                    

\begin{document}

\pagenumbering{roman}

%
%

\newgeometry{left=2.5cm,right=2.5cm,top=2.5cm}
\thispagestyle{empty}

\newdateformat{monthyeardate}{%
  \monthname[\THEMONTH] \THEYEAR}

\ifStandardTitle 

\noindent
\begin{minipage}[t]{6cm}%
{\footnotesize%
\raisebox{-\height}{{\bfseries The Australian National University}} \\
~2600 ACT~\textbar~Canberra~\textbar~Australia}
\end{minipage}%
\hfill%
\begin{minipage}[b]{10cm}%
\hfill\raisebox{-\height}{\includegraphics[height=2 cm]{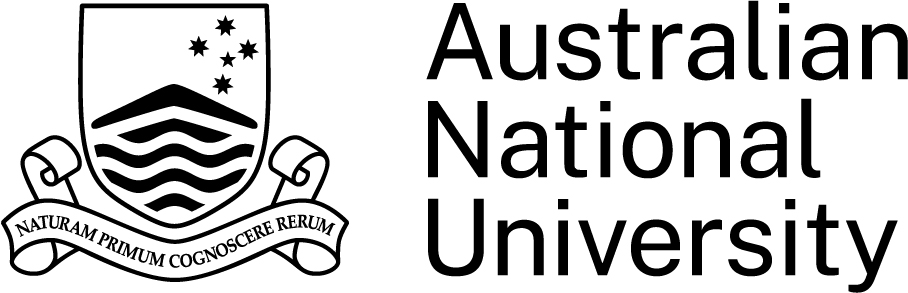}}
\end{minipage}

\ \\[2em]
\phantom{x} \hfill
\begin{minipage}{58.75 mm}
\raggedright
\bfseries \School\\[.5em]
\mdseries%
\noindent\College
\end{minipage}\\[6 em]
\hfill

\noindent
\parbox{140mm}{\sffamily \bfseries \Huge %
\ProjectTitle%
}\\[.75 em]
{--- \ProjectPoints{} pt \ifHonoursThesis Honours \else research \fi project (\Semester{} \Year)}\\[3 em]

\ifHonoursThesis%
A thesis submitted for the degree\\
\emph{\Degree}\\[3 em]
\else%
A report submitted for the course\\
\emph{\CourseCode, \CourseName}\\[3 em]
\fi

\noindent
{\footnotesize \textbf{By:}}\\
\AuthorName\\[2em]

\noindent
{\footnotesize \bfseries Supervisor\ifTwoSupervisors{}s\fi:}\\
{\footnotesize \FirstSupervisor%
\ifTwoSupervisors\\\SecondSupervisor\fi}\\[2 em]
\vfill
{\footnotesize \monthyeardate\today}

\else 

\begin{center}
\ \\[1em]
{\bfseries \Huge \ProjectTitle}\\[4em]
\ifHonoursThesis%
\Large{A thesis submitted for the degree}\\
\Large{\emph{\Degree}}\\[.5em]
{\ProjectPoints{} pt Honours project, \Semester{} \Year}
\else%
\Large{A report submitted for the course}\\
\Large{\emph{\CourseCode, \CourseName}}\\[.5em]
{\ProjectPoints{} pt research project, \Semester{} \Year}
\fi
\ \\[4em]
{\footnotesize \textbf By:}\\
\textbf{\AuthorName}\\[3em]
{\bfseries Supervisor\ifTwoSupervisors{}s\fi:}\\
{\FirstSupervisor%
\ifTwoSupervisors\\\SecondSupervisor\fi}\\[6em]
\includegraphics[height=2.5cm]{figures/ANU-logos/ANU_Primary_Horizontal_Black.jpg}\ \\[3em]
{\bfseries \School}\\
{\mdseries \College}\\
The Australian National University
\vfill
\normalsize{\monthyeardate\today}
\end{center}

\fi

\restoregeometry
{\sffamily\bfseries\Large Declaration:}\\

I declare that this work:\\

\begin{itemize}
  \item upholds the principles of academic integrity, as defined in the \href{https://www.anu.edu.au/about/governance/legislation}{University Academic Misconduct Rules};
  \item is original, except where collaboration (for example group work) has been authorised in writing by the course convener in the class summary and/or Wattle site;
  \item is produced for the purposes of this assessment task and has not been submitted for assessment in any other context, except where authorised in writing by the course convener;
  \item gives appropriate acknowledgement of the ideas, scholarship and intellectual property of others insofar as these have been used;
  \item in no part involves copying, cheating, collusion, fabrication, plagiarism or recycling.
\end{itemize}

\vspace{1 cm}
\hfill \monthname, \AuthorName
\newpage

\chapter*{Acknowledgements}

I would first of all like to thank my honours supervisor Ahad Zehmakan, who was also my ICPC coach and employer for tutoring. His consistent dedication to all three of these roles and many others helped make this thesis possible. I would like to express my gratitude for Ahad's guidance in not only the direction and writing of this thesis, but in chats about my future and career; he has been a great role model and has made pursuing academia all the more enticing an option.

I would also like to thank my previous supervisor (and employer for tutoring) Pascal Bercher, for preparing me for an honours project, and for always being fun to be around.

Thank you to Zaran for competing in ICPC regionals alongside me, accompanying me to Korean study events and always being available to talk to.

Finally, a big thanks to my family and friends for encouraging me to pursue my interest in academic studies, and last but not least Alyssia, for the last two years of wonderful memories, road trips, beautiful cards and consistent company through all the ups and downs of life.                        
\chapter*{Abstract}

Theoretical computer science plays an important role in the understanding of social networks and their properties. We can model information rippling throughout social networks, or the opinions of social media users for example, using graph theory and Markov chains. In this thesis, we model social networks as graphs, and consider two such processes:
\begin{itemize}
    \item Nodes talk to other nodes and find middle ground, causing their opinions to come closer to consensus (the load balancing model)
    \item All nodes take the maximum value of their neighbours in lockstep (the synchronous maximum model)
\end{itemize}
We study the convergence behaviours of each process, such as the eventual state of the graph, the convergence time and the period. We provide proofs of the eventual states and periods for each of the above models, and theoretical bounds for the worst case convergence times. We verify these with experiments, and explore further questions such as the average case convergence time of various special classes of graphs, or the convergence times when the model is altered slightly.                                

\renewcommand{\contentsname}{Table of Contents}   
\cleardoublepage\tableofcontents\cleardoublepage  
\pagenumbering{arabic}

\chapter{Introduction}



The \citeauthor{disinfoDozen2021} published a report in 2021 exposing that the majority of anti-vaccine fake news posts on various social media websites such as Facebook and Twitter were spread by just twelve people. These posts included claims that the pandemic was fake, orchestrated by Bill Gates, or designed to oppress and wipe out non-white people. An analysis of over 812\,000 of such posts found that 65\% of were attributable to just twelve individuals, nicknamed the \emph{disinformation dozen}. This disinformation was spread around the globe; almost 80\% of consumers in the United States reported having seen fake news on the coronavirus outbreak, and 38.2\% admitted to accidentally having shared fake news with others \citep{statistaFakeNews}. Widespread misinformation pertaining to vaccines caused hesitancy in many places around the globe, and estimates of tens of thousands of lives could have been saved if the World Health Organisation had met its vaccination goals \citep{wells2022global}.

Heading into the 21st century, the world is becoming increasingly globalised and interconnected. The internet, which started as a way for government researchers to share information by having computers communicate with each other \citep{leiner2009brief} has become a facet of daily life for over half of the world's population \citep{dataReportalOverview}. With 51\% of teenagers reporting that they use social media at least daily \citep{vogels2022teens}, and 54\% saying it would be hard to give it up \citep{lenhart2015teens}, the importance of social networks and social media in this unprecedented internet age has become clear. Many make their friends online now \citep{socialMediaAndTeens}, which has allowed a select few sites (e.g. Facebook, TikTok, Twitter, etc.) to become hubs where social interaction can take place.

As the world becomes more connected in this way, companies and institutions gain the power to spread ideas and exert their influence. Advertisers have targeted consumers with free samples in the hopes that their product will spread; if Alice mentions the product to Bob and Charlie, who each talk to two of their friends and so on, then the information ripples exponentially throughout the global social network. Because of this, many have found a career in being an `online influencer' for various products. More troubling is that political institutions or individuals such as those in the disinformation dozen, who could use this power to spread facts to help the population address the coronavirus outbreak, or call awareness to genuine political and societal issues, have instead used this power to increase their voter share by targeting and even misinforming the population.

In the wake of global events that have very quickly inflicted hardship or even death upon millions of people, a stronger understanding of the dynamics of social networks is therefore imperative in ensuring that we as a society are prepared to deal with such adversities. A solid foundation of formal models for these phenomena, and an understanding of the properties and consequences of these models, is crucial for our understanding of the real world. For this, and many other reasons, researchers across a large variety of disciplines have recently put attention into studying fake news and the ways in which information, particularly rumours, spread through social networks, cf.~\cite{kempe2003maximizing, liu2023fast, n2020rumor}. \cite{pennycook2021psychology} looked at the psychology of fake news and found a substantial disconnect between what users believe and what they share on social media, driven primarily by inattention rather than a purposeful desire to mislead. \cite{vosoughi2018spread} conducted data analysis and found that fake news spreads more than the truth due to human users, not bots. \cite{nekovee2007theory} constructed a model for rumour spreading on social networks and developed equations which show that scale-free networks (i.e. those that resemble the real world) are particularly prone to rumour spreading.

Researchers in computer science have used Natural Language Processing and Data Mining techniques to tackle fake news \citep{bondielli2019survey}. Furthermore, the processes by which the opinions of those in a social network update over time have been modelled in a number of ways and analysed from a more theoretical perspective, cf.~\cite{avin2019majority,soltani2023minimum,gartner2018majority}; this is the primary focus of this thesis. The idea is that the real world is highly complex, and therefore impenetrable by the abstract simplicity of pure mathematics and theoretical computer science. However, by constructing models which strip back the complexity and hone in on one or two key aspects of the process, the theoretical analysis becomes clearer, and insights about the dynamics of the model carry over into the real world. In other words, rather than trying to understand the complex real world all at once, each field does its part to focus on a smaller, simplified problem, and insights are pooled afterwards. The tools computer science brings to the table include Markov chain models, experimental analysis, and a comprehensive literature of graph theory.

Social networks are most often viewed as \emph{undirected graphs} where people are nodes and two nodes are connected to each other if the people are friends. Sometimes we consider users as being influenced by others or `following' others on social media, without necessitating that the converse is true, in which case we use \emph{directed graphs}. Each node has a value corresponding to its \emph{opinion}, which updates based on the opinions of adjacent nodes. We can customise which nodes update on each timestep and how they do so to get the `physics' of our \emph{dynamic process}. The consequences of the physics of this process, including its convergence behaviours and the expected time taken until its convergence, are what the literature aims to study.

In \emph{asynchronous} models, one or more nodes are chosen to update their opinions; typically the selection is random, and the \emph{expected} behaviour of the model is analysed. (Also, the setup where the choices are made by an adversary with a various level of power have been studied.) For example, we might pick an edge randomly and update both nodes relative to each other; this is called a \emph{symmetric} model, and it illustrates two people having a conversation and influencing each other. Over time, nodes with higher degree, corresponding to people in the social network with more connections, have a higher number of conversations and update more often; if we think this is unrealistic and wish to prevent this, we can consider \emph{synchronous} models, in which all nodes update their opinions simultaneously, i.e.\ in lockstep fashion.

The first of two dynamic processes considered in this thesis simulates people reaching a consensus through conversation. Each iteration, an edge is chosen, and the two friends on its endpoints have a conversation, which brings their opinions closer together in value. We assume that conversations between friends always promote agreement, which is not necessarily the case in real life; other dynamic processes could be constructed to model polarisation or radicalisation, which are also of interest to researchers \citep{bermingham2009combining} \citep{lara2017measuring}. We show that over time, each connected cluster of friends has two similar (adjacent in value) opinions that they converge to, and all extreme opinions are pulled towards the center of their cluster. We analyse a possible worst case social network in detail by reducing it to the gambler's ruin problem and using Markov chains, then use it to provide a polynomial bound on the expected number of iterations until convergence. We then use experimental analysis both to verify these results, and to investigate further questions of interest.

Then, we turn our attention to a second dynamic process where each node takes the maximum of its neighbours in lockstep, causing values to cycle around the graph. Whilst this may not have a clear analogy to a real social phenomenon, the techniques used in analysing this graph (i.e.\ decomposing into strongly connected components and considering equivalence classes) are a worthwhile addition to the literature and may prove to be useful when similar models are considered in the future. In the undirected case, we show that the maximum spreads from one end of the graph to the other and back again, meaning the convergence time is bounded by the diameter. In the directed case, we divide the graph into equivalence classes and show what value each class converges to. We show that the process eventually converges to a cycle, as the graph's maximum values are passed around in a predictable manner forever. This gives us useful results for the undirected case. Then, we analyse the convergence time, again both theoretically and experimentally.                            
\chapter{Background}\label{chap:background}

This chapter provides the prerequisite concepts for the chapters that follow. We start by providing our definitions of graphs and valuations, then dynamic processes and the models to be analysed. We cover some preliminary theory on Markov chains and strongly connected components as they are important for discussing convergence time and other convergence properties. Then, we move on to the \emph{experimental background}, which covers in detail the process of running the experiments, and justifies the methodology used. This is important, as it allows us to elaborate on the results of our experiments in subsequent chapters.

\section{Theoretical Background}

\subsection{Graphs}
Let $G=(V,E)$ be a graph where $V=\{v_1,\cdots,v_n\}$ is the set of vertices, and $E\subseteq\{(v_i,v_j) : v_i,v_j\in V, i\neq j\}$ is the set of edges. Note that this definition enforces that $G$ is simple, since there are no multi-edges or loops. Let $n\coloneqq|V|$ and $m\coloneqq|E|$. For instance, $n=7$, $m=8$, $V=\{v_1,\cdots,v_7\}$ and
\begin{align*}
    E = \{  &(v_1,v_2), \\
            &(v_2,v_3), \\
            &(v_3,v_4), \\
            &(v_4,v_5), \\
            &(v_5,v_6), \\
            &(v_6,v_1), \\
            &(v_6,v_7), \\
            &(v_7,v_5)\}
\end{align*}
corresponds to the following graph:
\begin{center}
    \begin{tikzpicture}
        \def\dist{5em};
        \foreach \angle in {1, 2, ..., 6}
            \node[main] (\angle) at ({120-60*(\angle-1)}:\dist) {$v_\angle$};
        \path (5) -- node[pos=1,main] (7) {$v_7$} +(-\dist, 0);
        \foreach \from/\to in {1/2, 2/3, 3/4, 4/5, 5/6, 6/1}
            \path[->] (\from) edge[bend left = 20] (\to);
        \foreach \from/\to in {6/7, 7/5}
            \path[->] (\from) edge[bend right = 20] (\to);
    \end{tikzpicture}
\end{center}
An \emph{undirected graph} $G$ satisfies $(u,v)\in E\iff(v,u)\in E$, meaning each edge can be traversed in both directions.

\begin{defn}[Neighbourhoods]\label{defn:Gamma}
    Let $\Gamma^k(v)$ be the set of vertices reachable in $k$ steps from $v$. That is, $\Gamma^k : V \to 2^V$, and:
    \begin{align*}
        \Gamma^0(v) &= \{v\} \\
        \Gamma^k(v) &= \{x\in V : (w,x)\in E\text{ for some }w\in\Gamma^{k-1}(v)\}
    \end{align*}
\end{defn}
In particular, let $\Gamma=\Gamma^1$, so $\Gamma(v)$ is the set of vertices $v$ pointed at by $v$. Note that:
\begin{eqn}[Immediate Neighbourhood]\label{eqn:Gamma1}
    \[ \Gamma(v) = \{x\in V : (v,x)\in E\} \]
\end{eqn}
For instance, in our example graph $\Gamma^4(v_7)=\{v_2,v_5\}$.

Let $d(u,v)$ denote the \emph{shortest distance} between nodes $u,v\in V$, that is, the number of edges in the shortest path from $u$ to $v$. Formally:
\[ d(u,v) \coloneqq \min\{k\geq 0 : v\in\Gamma^k(u)\} \]
We say $d(u,v)=\infty$ if $u$ and $v$ are not in the same connected component.
\begin{defn}[Diameter]\label{defn:diameter}
    The \emph{diameter} of a graph $G$ is denoted $D$, and is the longest of the shortest distances (excluding infinity). Formally:
    \[ D \coloneqq \max_{\shortstack{$u,v\in V$ \\ $d(u,v)\neq\infty$}}d(u,v) \]
\end{defn}

\subsection{Valuations}
Let $f_t$ be a \emph{valuation function} where $f_t(v)$ describes the value (or opinion) of $v$ at time $t\geq 0$. Note that $f_0$ describes the initial valuation (or opinions) of the network. Two valuations are \emph{isomorphic} (i.e.\ equivalent up to reordering) if each connected component is assigned the same multiset of values by both. Formally:
\begin{defn}\label{defn:isomorphic}
    Two valuations $\alpha$ and $\beta$ of an undirected graph $G=(V,E)$ are \emph{isomorphic} if and only if there is a bijection $I:V\to V$ such that $\beta(v)=\alpha(I(v))$ and $I(v)$ is in the same connected component of $G$ as $v$ for all $v\in V$.
\end{defn}
For example, these two valuations are isomorphic:
\begin{center}
    \begin{tikzpicture}[node distance = 3em]
        \def\dist{3em};
        \foreach \angle in {1, 2, 3}
            \node[main] (\angle) at ({90-(\angle-1)*120}:\dist) {\angle};
        \foreach \from/\to in {1/2, 2/3, 3/1}
            \path (\from) edge (\to);
        \node[main] (4) at (8em, 0) {4};
        \node[main] (5) [right = of 4] {5};
        \path (4) edge (5);
    \end{tikzpicture} \\
    \begin{tikzpicture}
        \def\dist{3em};
        \node[main] (1) at (90:\dist) {2};
        \node[main] (2) at (-30:\dist) {1};
        \node[main] (3) at (-150:\dist) {3};
        \foreach \from/\to in {1/2, 2/3, 3/1}
            \path (\from) edge (\to);
        \node[main] (4) at (8em, 0) {5};
        \node[main] (5) [right = of 4] {4};
        \path   (4) edge (5)
                (1) -- node[pos=1] {} +(0, \picgap);
    \end{tikzpicture}
\end{center}
On the other hand, the following two valuations are not isomorphic, as any isomorphism satisfying the first property sends $5$ to a different connected component.
\begin{center}
    \begin{tikzpicture}[node distance = 3em]
        \def\dist{3em};
        \foreach \angle in {1, 2, 3}
            \node[main] (\angle) at ({90-(\angle-1)*120}:\dist) {\angle};
        \foreach \from/\to in {1/2, 2/3, 3/1}
            \path (\from) edge (\to);
        \node[main] (4) at (8em, 0) {4};
        \node[main] (5) [right = of 4] {5};
        \path (4) edge (5);
    \end{tikzpicture} \\
    \begin{tikzpicture}
        \def\dist{3em};
        \node[main] (1) at (90:\dist) {5};
        \node[main] (2) at (-30:\dist) {1};
        \node[main] (3) at (-150:\dist) {3};
        \foreach \from/\to in {1/2, 2/3, 3/1}
            \path (\from) edge (\to);
        \node[main] (4) at (8em, 0) {2};
        \node[main] (5) [right = of 4] {4};
        \path   (4) edge (5)
                (1) -- node[pos=1] {} +(0, \picgap);
    \end{tikzpicture}
\end{center}

\subsection{Dynamic Processes}
Two types of update rules are considered in this thesis:
\begin{enumerate}
    \item $v$ is adjacent to a neighbour $w_v$. Then the update rule is $g:\R\times\R\to\R$, and for $t\geq 1$, we update $f_t(v)=g(f_{t-1}(v), f_{t-1}(w_v))$. Then we could:
    \begin{itemize}
        \item Update all $v$ simultaneously. Each $v$ picks a neighbour $w_v$ uniformly randomly \textbf{(synchronous model)}.
        \item Pick a vertex $v$ uniformly randomly and have it pick a neighbour $w_v$ uniformly randomly. Then update $v$ \textbf{(asynchronous model)}.
        \item If the graph is undirected, pick an edge $(u,v)\in E$ randomly and update $u$ and $v$ based on each other \textbf{(symmetric model)}. This means we set:
        \begin{align*}
            f_t(u) &= g(f_{t-1}(u), f_{t-1}(v)) \\
            f_t(v) &= g(f_{t-1}(v), f_{t-1}(u))
        \end{align*}
    \end{itemize}
    
    \item \label{typeTwo} $v$ updates based on all of its neighbours at once. For instance, it could update based on the min, max, sum or mode. We could:
    \begin{itemize}
        \item Update all $v$ simultaneously \textbf{(synchronous model)}
        \item Pick a vertex $v$ to update uniformly randomly \textbf{(asynchronous model)}
    \end{itemize}
\end{enumerate}

\vspace{1.5em}
For the first type, if a vertex $v$ has no neighbours, then it will not pick a neighbour, and will not update. For the second type however, there is a problem, since the minimum, maximum and mode of an empty set are undefined. Rather than introduce infinities, we assume all vertices are pointing to at least one neighbour that they can update from.
\begin{ass}\label{assNeighbour}
    When considering an update rule of \cref{typeTwo} on a graph $G=(V,E)$, we assume all vertices in $V$ have an outdegree of at least $1$.
\end{ass}

The \emph{load balancing} update rule is as follows:
\[ g(a,b) \coloneqq \begin{cases}
    a-1, & b < a \\
    a, & b = a \\
    a+1, & b > a
\end{cases} \]
That is, $a$ is brought no more than one step closer to $b$. An example of an application for the load balancing model is in politics, where opinions are often modelled as a spectrum (e.g.\ left-wing vs. right-wing, authoritarian vs. libertarian, democrat vs. republican, etc.). We can assume, to simplify the model, that conversations occur only between two people; sometimes, this assumption is lifted and people have conversations in \emph{clusters}, to prevent a \emph{celebrity} (i.e.\ high-degree node) from spreading their opinions (or a virus) to too many people at once \citep{galam2002minority}. The load balancing model illustrates when those with differing political opinions find some kind of `common ground' in their conversation, and being to understand the other side, bringing their opinions closer to consensus. Note that the load-balancing model is relative, not absolute; there is no special behaviour surrounding the opinion $0$, nor is there any difference between conversing with someone on your side vs.\ the other side. The \emph{symmetric load balancing} model is the topic of \cref{chap:load-balancing}.

The \emph{synchronous maximum model} is the topic of \cref{chap:max-model}, where $v$ updates based on the maximum of its neighbours. That is, for $t\geq 1$:
\[ f_t(v) = \max_{w\in\Gamma(v)}f_{t-1}(v) \]
Note that if $v$ is larger than all of its neighbours, it will throw away its value and update to something lower. If this were not the case, this model could mimic the spread of information throughout a social network, or the accuracy of peoples' understanding of a factual issue. Conversations between two vertices would correspond to one friend informing the other in more detail about an issue, and we could study the time taken for the entire graph to be fully informed. However, such a model would be far too simplistic mathematically, as the maximum would simply propagate throughout the network. Thus, we sacrifice real world applicability to analyse a more mathematically interesting model, with the hopes that analysing the convergence properties of such a model will introduce more insightful techniques and methods into the literature, and give results that are still worth knowing.

\subsection{Markov Chains}
Let $G=(V,E)$ be a graph and $f_0 : 2^V\to\R$ be an initial valuation function. When trying to predict $f_1$, in general, there may be many possibilities. For example, in the symmetric load balancing model, any of the following updates are possible:
\begin{center}
    \begin{tikzpicture}
        \def\dist{3em};
        \node[main,redd] (1) at (90:\dist) {4};
        \node[main] (2) at (-30:\dist) {6};
        \node[main,redd] (3) at (-150:\dist) {5};
        \path   (1) edge (2) edge[redd] (3)
                (2) edge (3);
        \draw[-{Latex[length=1em]}] (6em,0) -- +(4em,0);
        \path   (16em,0)    -- node[pos=1,main] (1) {5} +(90:\dist)
                            -- node[pos=1,main] (2) {6} +(-30:\dist)
                            -- node[pos=1,main] (3) {4} +(-150:\dist)
                (1) edge (2) edge (3)
                (2) edge (3);
    \end{tikzpicture} \\
    \begin{tikzpicture}
        \def\dist{3em};
        \node[main,redd] (1) at (90:\dist) {4};
        \node[main,redd] (2) at (-30:\dist) {6};
        \node[main] (3) at (-150:\dist) {5};
        \path   (1) edge[redd] (2) edge (3)
                (2) edge (3);
        \draw[-{Latex[length=1em]}] (6em,0) -- +(4em,0);
        \path   (16em,0)    -- node[pos=1,main] (1) {5} +(90:\dist)
                            -- node[pos=1,main] (2) {5} +(-30:\dist)
                            -- node[pos=1,main] (3) {5} +(-150:\dist)
                (1) edge (2) edge (3)
                (2) edge (3)
                (1) -- node[pos=1] {} +(0, \picgap);
    \end{tikzpicture} \\\begin{tikzpicture}
        \def\dist{3em};
        \node[main] (1) at (90:\dist) {4};
        \node[main,redd] (2) at (-30:\dist) {6};
        \node[main,redd] (3) at (-150:\dist) {5};
        \path   (1) edge (2) edge (3)
                (2) edge[redd] (3);
        \draw[-{Latex[length=1em]}] (6em,0) -- +(4em,0);
        \path   (16em,0)    -- node[pos=1,main] (1) {4} +(90:\dist)
                            -- node[pos=1,main] (2) {5} +(-30:\dist)
                            -- node[pos=1,main] (3) {6} +(-150:\dist)
                (1) edge (2) edge (3)
                (2) edge (3)
                (1) -- node[pos=1] {} +(0, \picgap);
    \end{tikzpicture}
\end{center}
If $g$, $h$ and $k$ are the possibilities for $f_1$, each of which with a $1/3$ chance, then we can describe $f_0$ updating to $f_1$ as follows:
\begin{center}
    \begin{tikzpicture}
        \node (0) {$f_0$};
        \node (0g) [above right = of 0] {};
        \node (0h) [right = of 0] {};
        \node (0k) [below right = of 0] {};
        \node (g) [right = of 0g] {$g$};
        \node (h) [right = of 0h] {$h$};
        \node (k) [right = of 0k] {$k$};
        \foreach \vert in {g, h, k}
            \path[->] (0) edge node[sloped, above] {0.33} (\vert);
    \end{tikzpicture}
\end{center}
We could continue branching to get a tree diagram, but note that it is possible to reach the same valuation multiple times. Thus, a better visualisation is a meta-graph built upon $G$, where a valuation $\alpha$ points to a valuation $\beta$ with weight $p$ if the probability is $p$ that $G$ with valuation $\alpha$ updates to become graph $G$ with valuation $\beta$. Below is an example, with one possible update sequence highlighted in red:
\begin{center}
    \begin{tikzpicture}[node distance = 5em]
        \node (0) {$f_0$};
        \node (a) [right = of 0] {$\alpha$};
        \node (b) [right = of a] {$\beta$};
        \node (c) [below = of a] {$\gamma$};
        \path (a) -- node (offa) {} (c);
        \node (offb) [right = of offa] {};
        \node (e) [right = of offb] {$\epsilon$};
        \node[redd] (f) [right = of e] {$\cdots$};
        \node (d) [above = of e] {$\delta$};
        \node (g) [right = of d] {$\cdots$};
        \path[->]   (0) edge[redd] node[above] {0.50} (a)
                        edge node[sloped, above] {0.50} (c)
                    (a) edge[redd] node[below] {0.67} (b)
                        edge[redd] node[sloped, above] {0.33} (c)
                    (b) edge[redd, out=150, in=30] node[above] {0.50} (a)
                        edge node[sloped, above] {0.25} (d)
                        edge node[sloped, below] {0.25} (e)
                    (c) edge node[sloped, below] {0.67} (b)
                        edge[redd] node[sloped, below] {0.33} (e)
                    (d) edge node[above] {0.50} (g)
                        edge[out=-60, in=60] node[sloped, below] {0.50} (e)
                    (e) edge node[sloped, above] {0.25} (d)
                        edge[redd] node[above] {0.75} (f);
    \end{tikzpicture}
\end{center}
In this example, $f_1=f_3=\alpha$, $f_2=\beta$, $f_4=\gamma$ and $f_5=\epsilon$. This directed meta-graph is called the \emph{Markov chain} of possibilities for $G,f_0$.

\subsection{Strongly Connected Components}\label{sec:SCCs}
A directed graph $G=(V,E)$ is \emph{strongly connected} if and only if for all $u,v\in V$ there is a path from $u$ to $v$, i.e.\ $u$ can reach $v$. For example, a cycle is strongly connected since any node can reach any other by going around the cycle.
\begin{center}
    \begin{tikzpicture}
        \foreach \angle in {1, 2, ..., 5}
            \node[main] (\angle) at ({90-(\angle-1)*72}:5em) {$v_\angle$};
        \foreach \from/\to in {1/2, 2/3, 3/4, 4/5, 5/1}
            \path[->] (\from) edge[bend left = 20] (\to);
    \end{tikzpicture}
\end{center}
Any directed graph can be decomposed into its \emph{strongly connected components}. To do this, let $uRv$ be a relation that is true when $u$ can reach $v$, i.e. there is a path from $u$ to $v$. This relation is reflexive (each node can reach itself using the empty path) and transitive (if $uRv$ and $vRw$, then $uRw$), but not necessarily symmetric; in the following graph $uRv$ but we do not have $vRu$:
\begin{center}
    \begin{tikzpicture}[node distance = 4em]
        \node[main] (u) {$u$};
        \node[main] (v) [right = of u] {$v$};
        \node[main] (w) [below = of u] {$w$};
        \node[main] (x) [below = of v] {$x$};
        \path (v) -- node (offv) {} (x);
        \node[main] (y) [right = of offv] {$y$};
        \node[main] (t) [below = of y] {$t$};
        \node[main] (z) [right = of y] {$z$};
        \path[->]   (u) edge (v) edge (w)
                    (v) edge (x) edge (y)
                    (x) edge (w) edge (y) edge (t)
                    (w) edge (v)
                    (y) edge[bend left] (z)
                    (t) edge (y)
                    (z) edge[bend left] (y);
    \end{tikzpicture}
\end{center}
Consider a new relation $uSv$ that is true if and only if $uRv$ and $vRu$, that is, $u$ and $v$ can reach each other. Since $R$ is reflexive and transitive, so is $S$, but we can see from the definition that $S$ is also symmetric. Thus, $S$ is an equivalence relation, so it partitions the graph into equivalence classes.
\begin{center}
    \begin{tikzpicture}[node distance = 4em]
        \node[main] (u) {$u$};
        \node[main] (v) [right = of u] {$v$};
        \node[main] (w) [below = of u] {$w$};
        \node[main] (x) [below = of v] {$x$};
        \path (v) -- node (offv) {} (x);
        \node[main] (y) [right = of offv] {$y$};
        \node[main] (t) [below = of y] {$t$};
        \node[main] (z) [right = of y] {$z$};
        \path[->]   (u) edge (v) edge (w)
                    (v) edge (x) edge (y)
                    (x) edge (w) edge (y) edge (t)
                    (w) edge (v)
                    (y) edge[bend left] node (yz) {} (z)
                    (t) edge (y)
                    (z) edge[bend left] (y);
        \draw[scc] (u) circle (1.7em);
        \draw[scc, rounded corners=2.5em] (2.8,-2.8) -- (-1.3,-2.8) -- (2.8,1.3)[rounded corners=0.8em] -- cycle;
        \draw[scc] (t) circle (1.7em);
        \path (y) -- node (yz) {} (z);
        \draw[scc] (yz) ellipse (5em and 2em);

        \path   (u) -- node[pos=1, scclabel] {\huge\textbf{A}} ++(0, 3.2em)
                (v) -- node[pos=1, scclabel] {\huge\textbf{B}} ++(0, 3.5em)
                (t) -- node[pos=1, scclabel] {\huge\textbf{C}} ++(3.2em, 0)
                (yz) -- node[pos=1, scclabel] {\huge\textbf{D}} ++(0, 3.2em);
    \end{tikzpicture}
\end{center}
These are called the \emph{strongly connected components (SCCs)}. Within each SCC, the induced graph is strongly connected, as all vertices can reach each other. We can then consider a meta-graph, where each node represents an SCC, and we have an edge from one to the other if there is a node in one SCC which has an edge to a node from the other.
\begin{center}
    \begin{tikzpicture}
        \node[main, scclabel] (A) {\huge\textbf{A}};
        \node[main, scclabel] (B) [right = of A] {\huge\textbf{B}};
        \node[main, scclabel] (C) [right = of B] {\huge\textbf{C}};
        \node[main, scclabel] (D) [right = of C] {\huge\textbf{D}};

        \path[->, scclabel] (A) edge (B)
                            (B) edge (C) edge[bend left = 35] (D)
                            (C) edge (D);
    \end{tikzpicture}
\end{center}
Such meta-graphs of SCCs are guaranteed to be \emph{acyclic}, since if there were a cycle from one SCC to another, then they would have to be the same SCC, since nodes in one can reach nodes in the other, and vice versa. The acyclicity means there is a \emph{topological sorting} of the graph; visually, this means the SCCs can be placed in an order from left to right such that edges only point towards the right.

\subsection{Convergence Time and Period}
Consider the SCCs of the example Markov chain from earlier:
\begin{center}
    \begin{tikzpicture}[node distance = 5em]
        \node (0) {$f_0$};
        \node (a) [right = of 0] {$\alpha$};
        \node (b) [right = of a] {$\beta$};
        \node (c) [below = of a] {$\gamma$};
        \path (a) -- node (offa) {} (c);
        \node (offb) [right = of offa] {};
        \node (e) [right = of offb] {$\epsilon$};
        \node[redd] (f) [right = of e] {$\cdots$};
        \node (d) [above = of e] {$\delta$};
        \node (g) [right = of d] {$\cdots$};
        \path[->]   (0) edge[redd] node[above] {0.50} (a)
                    (0) edge node[sloped, above] {0.50} (c)
                    (a) edge[redd] node[below] {0.67} (b)
                    (a) edge[redd] node[sloped, above] {0.33} (c)
                    (b) edge[redd, out=150, in=30] node[above] {0.50} (a)
                    (b) edge node[sloped, above] {0.25} (d)
                    (b) edge node[sloped, below] {0.25} (e)
                    (c) edge node[sloped, below] {0.67} (b)
                    (c) edge[redd] node[sloped, below] {0.33} (e)
                    (d) edge node[above] {0.50} (g)
                    (d) edge[out=-60, in=60] node[sloped, below] {0.50} (e)
                    (e) edge node[sloped, above] {0.25} (d)
                    (e) edge[redd] node[above] {0.75} (f);
        \draw[scc] (0) circle (1.7em);
        \draw[rounded corners=2.5em, scc] (2,1) -- (6.5,1) -- (2,-3.6)[rounded corners=0.8em] -- cycle;
        \path (d) -- node (de) {} (e);
        \draw[scc] (de) ellipse (2em and 5em);
        \path   (a) -- node (ab) {} (b)
                (ab) -- node[pos=1, scclabel] {\huge\textbf{A}} ++(0, 4em)
                (d) -- node[pos=1, scclabel] {\huge\textbf{B}} ++(0, 3em);
    \end{tikzpicture}    
\end{center}
Notice that the process bumps around a little in SCC \textbf{A} before moving on to SCC \textbf{B}. It can never return to \textbf{A} once it enters \textbf{B} because \textbf{A} can reach \textbf{B}, so if \textbf{B} could reach \textbf{A}, it would contradict that they are different SCCs. Another important observation is that because there is an edge from \textbf{A} to \textbf{B}, we cannot stay in \textbf{A} forever. In other words, the probability that we are still in \textbf{A} after $n$ steps decreases exponentially to zero. This means with high probability (w.h.p.) we will eventually end up in a SCC with no edges exiting to other SCCs; we call these \emph{absorbing states}.
\begin{thm}\label{absorb}
    Let $M$ be a finite Markov chain. Then the probability of being absorbed is $1$, that is, the probability of not being in an absorbing state after $n$ steps decreases exponentially to $0$ as $n$ goes to infinity.
\end{thm}
This theorem proves to be important in \cref{chap:load-balancing} for understanding the load balancing model, and also for defining convergence time and period.

Consider the Markov chain corresponding to a synchronous model. Since synchronous models are deterministic, they update in a predictable way each time. If the Markov chain is finite, by the pigeonhole principle the same state will eventually be reached twice. When this happens, the process enters a loop, and will continue running through the same possibilities forever.
\begin{center}
    \begin{tikzpicture}[node distance = 3.5em]
        \node (0) {$f_0$};
        \node (a) [right = of 0] {$\alpha$};
        \node (b) [right = of a] {$\beta$};
        \node (c) [right = of b] {$\gamma$};
        \node (d) [right = of c] {$\delta$};
        \node (dots) [right = of d] {$\cdots$};
        \node (z) [right = of dots] {$\omega$};
        \path[->]   (0) edge node[below] {1.00} (a)
                    (a) edge node[below] {1.00} (b)
                    (b) edge node[below] {1.00} (c)
                    (c) edge node[below] {1.00} (d)
                    (d) edge node[below] {1.00} (dots)
                    (dots) edge node[below] {1.00} (z)
                    (z) edge[bend right = 25] node[below] {1.00} (c);
        \draw[scc]  (0) circle (1.7em)
                    (a) circle (1.7em)
                    (b) circle (1.7em);
        \path (d) -- node (centr) {} (dots);
        \draw[scc] (centr) ellipse (9.7em and 3.2em);
    \end{tikzpicture}
\end{center}
Consider the SCCs of this Markov chain. Each node forms its own SCC, aside from the absorbing state at the end. The period $p$ is the size of this absorbing state, and the process eventually repeats the same $p$ valuations over and over again.
\begin{thm}\label{direcPeriod}
    Let $M$ is the finite Markov chain for a deterministic process. Then there exists $T\geq 0$ such that $f_t=f_{t+p}$ for all $t\geq T$.
\end{thm}
This is consistent with the definition of a period in other fields (e.g.\ group theory in mathematics). For non-deterministic processes, we still define the period $p$ to be the size of the absorbing state, even though the process may not necessarily cycle through the same $p$ states on repeat anymore.

Now that we know what happens in finite Markov chains, we can define \emph{convergence time} to be the earliest $t$ such that $f_t$ is in an absorbing state. Since it is impossible to move from one absorbing state to another, the first time we enter an absorbing state is also the last. For deterministic processes, the convergence time is also the minimum $T$ such that $f_t=f_{t+p}$ for all $t\geq T$.

\section{Experimental Background}
The results presented in the majority of this thesis were discovered by theoretical analysis, and the proofs are simple enough to be understood with the above background. However, many interesting results concerning these models are much more challenging to analyse from a purely theoretical perspective, and so only proof sketches, or sometimes no proofs at all can be easily provided. Experimental analysis opens up opportunities to deal with the shortcomings of theoretical analysis, which include the following:
\begin{itemize}
    \item For the load-balancing and strongly connected maximum models, theoretical analysis directs us to particular bounds for the convergence time, but does not immediately present us with clean and elegant proofs. Experimental analysis can be used to reinforce the validity of such bounds, while they remain theoretically unproven.

    \item It is easy to conjecture that graphs that are more sparsely connected should take longer to spread their values throughout the graph, but such a statement is difficult to make rigorous from a theoretical perspective. Even if we define and consider \emph{expander graphs} for example, the properties of such graphs do not immediately translate into elegant proofs of their longer convergence time. We can instead verify these conjectures by simulating the process on special cases and comparing it with the theoretical bound, circumventing the need for challenging theoretical analysis.

    \item Theoretical bounds simply tell us the worst case for any valuation and any graph; analysing the average case, or the `real world' case, often takes much more effort as random graphs and random valuations need to be properly defined. It is also challenging to rigorously define what constitutes a `random graph that could feasibly show up in the real world', and such constructions do not lend themselves easily to elegant proofs. Experimental analysis can give us a much clearer idea of what tends to happen in practice when these processes play out.
\end{itemize}
The shortest path problem on graphs which allow negative edge weights illustrates an example of the benefits of experimental analysis in algorithmic computer science literature. The Shortest-Path Faster Algorithm, first published by \cite{moore1959shortest}, has the same worst case time complexity $O(|V|\cdot|E|)$ as the Bellman-Ford algorithm, but experiments seem to suggest that it is much faster $O(|E|)$ on average.

Experimental analysis has its faults however; it can be vulnerable to biased experiments (and hence biased data), and just because a pattern seems to emerge from data does not give any guarantee that it can be generalised to cases outside of the scope of testing (consider the infamous Borwein integrals). Hence, we combine theoretical and experimental analysis to achieve the greatest overall understanding of the processes. The theoretical analysis gives us guarantees we know will always hold in all cases, and the experimental analysis lets us see how far off the theoretical bounds are in practice, for real social networks.

\subsection{Methodology}
The experimental analysis in this thesis was conducted in Python 3 on a laptop with 16G RAM and 2.8GHz 11th Gen Intel Core i7-1165G7 CPU, using a single thread. The number of trials varied based on the experiment and will thus be listed with each experiment. In general, it was set to the highest power of 10 such that the largest test case still ran in under five or so minutes, but for some test cases it was unavoidable that the experiment would run for upwards of an hour. 

\subsection{Analysed Graphs}\label{analysedGraphs}
Experiments were conducted on three types of graphs:
\begin{enumerate}
    \item Graphs of users in social networks (e.g.\ Facebook, Twitter)
    \item Random graphs (\ER and \BA models)
    \item Graphs motivated by theoretical analysis (e.g.\ constructions for the worst case)
\end{enumerate}

The experiments on social networks are conducted on four datasets from the SNAP database: Facebook and Twitter \citep{leskovec2012learning}, Twitch \citep{rozemberczki2019multiscale} and Wikipedia \citep{leskovec2010signed}. The Facebook and Twitch graphs are undirected as the friendships are mutual, whereas the Twitter and Wikipedia graphs are directed, as for Twitter, a user can follow another without being followed back, and for Wikipedia, a user can vote for another without being voted for in return.
\begin{mytable}
    \centering
    \begin{tabular}{ccccc}
        \toprule
        Dataset & Directed & Nodes & Edges & Diameter \\
        \midrule
        Facebook & No & 4\,039 & 88\,234 & 8 \\
        Twitch & No & 7\,126 & 35\,324 & 10 \\
        Twitter & Yes & 81\,306 & 1\,768\,149 & 7 \\
        Twitter* & Yes & 69\,572 & 1\,696\,476 & 7 \\
        Wikipedia & Yes & 7\,115 & 103\,689 & 7 \\
        Wikipedia* & Yes & 5\,158 & 70\,922 & 6 \\
        \bottomrule
    \end{tabular}
    \caption{Basic information for the analysed social networks}
\end{mytable}
Facebook and Twitch are connected, and therefore satisfy \Cref{assNeighbour}. Twitter is weakly connected (i.e.\ connected when viewed as an undirected graph), but is not strongly connected, and contains many users who are not following anyone, violating \Cref{assNeighbour}. Wikipedia is not even weakly connected, and thus also violates this assumption. Thus, modified graphs Twitter* and Wikipedia* were constructed from the Twitter and Wikipedia graphs by recursively removing the vertices that violated the assumption, removing about 10\% of the nodes. The resulting graphs contain some anomalies such as users who only have an outgoing edge to themselves and will therefore never update their opinion, but these do not violate the assumption and are therefore permitted.
\begin{mytable}
    \centering
    \begin{tabular}{cccccccc}
        \toprule
        \multirow{2}{*}{Dataset} & \multirow{2}{*}{Directed} & \multicolumn{3}{c}{Indegree} & \multicolumn{3}{c}{Outdegree} \\
        \cmidrule(lr){3-5} \cmidrule(lr){6-8}
        && Min & Max & Avg & Min & Max & Avg \\
        \midrule
        Facebook & No & 1 & 1045 & 43.69 & 1 & 1045 & 43.69 \\
        Twitch & No & 1 & 720 & 9.91 & 1 & 720 & 9.91 \\
        Twitter & Yes & 0 & 3\,383 & 21.75 & 0 & 1\,205 & 21.75 \\
        Twitter* & Yes & 0 & 3\,383 & 24.38 & 1 & 1\,165 & 24.38 \\
        Wikipedia & Yes & 0 & 457 & 14.57 & 0 & 893 & 14.57 \\
        Wikipedia* & Yes & 0 & 457 & 13.75 & 1 & 596 & 13.75 \\
        \bottomrule
    \end{tabular}
    \caption{Degree information for the analysed social networks}
\end{mytable}
Twitter and Twitter* are \emph{almost strongly connected}; their five largest SCCs have sizes 68\,413, 73, 53, 20 and 18. Wikipedia has a unique non-trivial SCC of size 1\,300; all other SCCs consist of a single node. Thus, we expect them to behave like strongly connected graphs, and have similar convergence times.

There are two related variants of the \ER model for undirected graphs; in this thesis, an \ER random graph $G(n,p)$ is generated from two values $n\geq 0$ and $0\leq p\leq 1$, where $n$ is the number of vertices and each of the $\binom{n}{2}$ edges independently have a probability $p$ of being included. For example, to generate an \Erdos-0.5 random graph, we iterate through each edge and add it with $p=0.5$ probability. In fact, for $p=0.5$ specifically, this process can be shown to uniformly randomly select a graph from the set of $2^{\binom{n}{2}}$ possible labelled graphs with $n$ vertices. Thus, we will use $p=0.5$ for the experiments, and denote this as \Erdos-0.5.

For directed graphs, we now iterate through the $n(n-1)$ ordered pairs of vertices rather than the $\binom{n}{2}$ pairs of unordered ones, and include that edge with a probability of 0.5. This means there is a 0.25 probability that both edges between $u$ and $v$ exist, a 0.25 probability that neither does, and 0.25 probabilities that one connects to the other but not vice versa.

Similarly, a \BA random graph $G(n,m)$ is an undirected graph with $n$ nodes, and $m$ is a parameter which correlates with the graphs \emph{connectedness}. Each iteration, a new node is added and connected to a random sample of $m$ nodes from the existing graph. Instead of choosing these $m$ nodes uniformly randomly, they are weighted in exact accordance with their degree. This creates \emph{scale-free networks}, which model social networks (as people with many connections are often easier to connect with for newcomers), social media (new users are recommended celebrity accounts by the platform itself), and even the World Wide Web (search engines like Google are more likely to recommend pages with many links to other pages).

For the experiments, \BA graphs are randomly generated by the NetworkX package in Python. This generates a \BA graph with $m(n-m)$ edges, meaning if we want $n$ nodes and $k$ edges, we can solve $m(n-m)=k$ to find the appropriate value for the parameter $m$. Doing this for the social networks above, we get the following values for $m$.
\begin{mytable}
    \centering
    \begin{tabular}{ccccc}
        \toprule
        Graph & Directed & Nodes & Edges & $m$ \\
        \midrule
        Facebook & No & 4\,039 & 88\,234 & 22.0 \\
        Twitch & No & 7\,126 & 35\,324 & 5.0 \\
        Twitter* & Yes & 69\,572 & 1\,696\,476 & 24.4 \\
        Wikipedia* & Yes & 5\,158 & 70\,922 & 13.8 \\
        \bottomrule
    \end{tabular}
    \caption{Appropriate values of $m$ for the various social networks}
\end{mytable}
From this table, we can derive that a value somewhere around $m=10$ is appropriate. Since we require that $m<n$ for all graph sizes $n$ we test with, and the lowest of these is $n=10$, we set $m=9$ for all experiments.

Note that \ER and \BA random graphs are not guaranteed to satisfy \Cref{assNeighbour}, but they are extremely likely to. This means we can just repeat the generation process until we obtain graphs that do satisfy it.

The remaining graphs considered are motivated by theoretical analysis and will be discussed later in the thesis. The primary value of interest in the experimental analysis is the convergence time, but for some cases we are also interested in other questions, such as which values are converged to, or the probability that a particular value is converged to.                              
\chapter{Related Work}\label{chap:relatedWork}


As early as 1974, \citeauthor{degroot1974reaching} proposed a model where the opinions of a committee are numbers which update based on a weighted average of the opinions of the other members. \cite{goles1980periodic} showed that if the members update according to DeGroot's model, but map their opinions to either 0 or 1 based on whether they are greater than a threshold value, then they will either reach a stable state (not necessarily in agreement) and never update again, or they will flicker back and forth between two states. Later, \cite{fogelman1983transient} showed that the convergence time is $\bigO(n^2)$, and \cite{frischknecht2013convergence} showed that this bound is almost tight (up to a sub-logarithmic factor).

If this idea is combined with graph theory, we get the \emph{majority model}, where nodes are black or white, and update each round to the colour shared by the majority of their neighbours, failing to update in the case of a tie. This could be used to model things like democrats and republicans in a social network, where the assumption is that the more friends one makes who support a given party, the more likely one is to align and support that party as well. For the synchronous case, the result by \citeauthor{goles1980periodic} shows that the period is one or two, and a result by \cite{poljak1986pre} shows that the convergence time is $\Theta(m)$, where $m$ is the number of edges in the graph. If instead of failing to update in the case of a tie, it chooses black or white randomly, with probability 0.5 each, then it is possible to construct examples with exponential convergence times and the ability to reach exponentially many colourings \citep{zehmakan2023random}.

For the asynchronous case, \cite{bredereck2017manipulating} showed that the graph converges to a stable state with high probability, and presented an algorithm to efficiently calculate the minimum and maximum number of black nodes possible. Their paper also considers how an adversary might attempt to manipulate the election result. \cite{brill2016pairwise} proposed a similar model, but where each node has a ranking of $k$ candidates. If $k=2$, this simplifies to the above model, but if $k\geq 3$, a node may switch candidates adjacent in its ranking if it disagrees with its neighbours on the order of those candidates.

\cite{feige1990randomized} proposed a model where a \emph{rumour} is known initially by a single node, and on each iteration, all nodes aware of the rumour choose a neighbour uniformly at random, and inform them of the rumour; if they already know the rumour, then nothing happens. Eventually, the entire graph is made aware of the rumour, and the time taken for this to occur is the primary question. As well as rumour spreading, this can be used to model information being broadcast through a computer network, or viruses gradually infecting more and more people in a global network. \citeauthor{feige1990randomized} went on to show that the convergence time of the process is $\bigO(\Delta(D+\log n))$, where $\Delta$ is the maximum degree of any node, $D$ is the diameter and $n$ is the number of nodes. The Push-Pull model, a similar model proposed by \cite{demers1987epidemic}, includes in each round a `pull' step, where all uninfected nodes pick a neighbour at random and become infected if this neighbour is infected. Much later, \cite{giakkoupis2012rumor} showed that the convergence time for this model is $\bigO(\Phi^{-1}\log n)$, where $\Phi$ is the \emph{conductance} of the graph (which is outside the scope of this thesis, but is related to the graph's \emph{connectedness}).

\cite{kempe2003maximizing} published a seminal paper in the research surrounding rumour-spreading models. It analyses the \emph{Linear Threshold} and \emph{Independent Cascade} models, both of which model a scenario where some nodes are initially active, and activate their neighbours, cascading throughout the network and eventually stabilising, similar to the aforementioned Alice, Bob and Charlie story. In the \emph{Linear Threshold} model, a node becomes active if a weighted proportion of its neighbours become active, and in the \emph{Independent Cascade} model, when a node becomes active, it activates each of its neighbours with certain probabilities. The paper by \citeauthor{kempe2003maximizing} reformulated these models to be more mathematically elegant, then used the theory of submodular functions to find efficient approximation algorithms. This influenced a variety of other papers; \cite{fazli2014non} studied when two or more competing `rumours' spread throughout the network, and posed the problem in the context of companies putting out technologies and competing for market share. \cite{he2018stability} introduced uncertainty into the model to better reflect the uncertainties of data scraped from real world social networks.

\cite{wilder2018controlling} combined this with the previous work on election control to consider the problem of controlling an election by influencing a set of nodes, using the same technique of submodular functions. Further works include \cite{abouei2022election} who introduced uncertainty, \cite{castiglioni2021election} who allowed adversaries to add or remove edges and \cite{zehmakan2023rumors} who experimented with rumour spreading in graphs with strong and weak expansion properties and with various countermeasures, then the effectiveness of said countermeasures.

The litertature on opinion formation model and information spreading process is rather vast. The interested reader is recommended to check~\cite{musco2018minimizing} for averaging based model such as FJ model,~\cite{zehmakan2021majority} for other variants of majority model, and~\cite{zehmakan2019spread} for bootstrap percolation models.                             
\chapter{Load Balancing Model}\label{chap:load-balancing}

The load balancing model illustrates how interactions can cause the opinions of those in a social network to converge over time. This chapter deals exclusively with \textbf{undirected graphs}. In the symmetric updating model, when an edge is updated, it brings the values of its endpoints closer together. We begin with a motivating example (\Cref{sec:motivating}), then in \Cref{sec:convBehaviour}, we will see that, given enough time, the probability that the entire network is brought to two adjacent values eventually becomes certain, and these values are known (\Cref{divisionConnected}, \Cref{divisionAny}). Then, in \Cref{sec:convTime}, we study the convergence time of this process, which we can show has two \emph{dimensions}. We provide a bound for one dimension and prove its tightness with a cleverly-constructed example (\Cref{binValuationsThm}), then provide proof sketches that bound a candidate worst case valuation by reducing it to the gambler's ruin problem (\Cref{ohNcubed}, \Cref{omegaNcubed}).

\section{Motivating Example}\label{sec:motivating}
Before diving into technical analysis, it is important to orient ourselves and get an intuitive idea of the most important behaviours and characteristics that result from the rules we have defined. Typically, we want these examples to be large enough to allow the more complex and important behaviours to present themselves, whilst being small enough to calculate reasonably quickly, and sparse enough to not distract from their most important features. For this chapter and subsequent chapters, we present examples engineered to pique the curiosity of the reader and cause them to make observations that help in the study of the process. The goal is to make the reader feel that they could have discovered the proof strategies and techniques themselves.

The following graph serves well as the motivating example for this chapter.
\begin{center}
    \begin{tikzpicture}
        \node[main] (1) {4};
        \node[main] (2) [below = of 1] {2};
        \node[main] (3) [right = of 1] {6};
        \node[main] (4) [right = of 2] {6};
        \node[main] (5) [right = of 3] {5};
        \path   (1) edge (2) edge (3) edge (4)
                (2) edge (4)
                (3) edge (4) edge (5)
                (4) edge (5);
        \path (1) -- node (mid12) {} (2);
        \node [left = of mid12] {$t=0$};
    \end{tikzpicture}
\end{center}
Here, each vertex $v\in V$ is labelled with its initial value $f_0(v)$. The first edge to update might be:
\begin{center}
    \begin{tikzpicture}
        \node[main] (1) {4};
        \node[main, redd] (2) [below = of 1] {2};
        \node[main] (3) [right = of 1] {6};
        \node[main, redd] (4) [right = of 2] {6};
        \node[main] (5) [right = of 3] {5};
        \path   (1) edge (2) edge (3) edge (4)
                (2) edge[redd] (4)
                (3) edge (4) edge (5)
                (4) edge (5);
        \path (1) -- node (mid12) {} (2);
        \node [left = of mid12] {$t=0$};
    \end{tikzpicture}
\end{center}
If so, $2$ is brought one step closer to $6$, so its new value is $3$. Similarly, $6$ is brought one step closer to $2$, so its new value is $5$:
\begin{center}
    \begin{tikzpicture}
        \node[main] (1) {4};
        \node[main] (2) [below = of 1] {3};
        \node[main] (3) [right = of 1] {6};
        \node[main] (4) [right = of 2] {5};
        \node[main] (5) [right = of 3] {5};
        \path   (1) edge (2) edge (3) edge (4)
                (2) edge (4)
                (3) edge (4) edge (5)
                (4) edge (5);
        \path (1) -- node (mid12) {} (2);
        \node [left = of mid12] {$t=1$};
    \end{tikzpicture}
\end{center}
Suppose the following edge is chosen next, at $t=1$:
\begin{center}
    \begin{tikzpicture}
        \node[main] (1) {4};
        \node[main] (2) [below = of 1] {3};
        \node[main] (3) [right = of 1] {6};
        \node[main, redd] (4) [right = of 2] {5};
        \node[main, redd] (5) [right = of 3] {5};
       \path   (1) edge (2) edge (3) edge (4)
                (2) edge (4)
                (3) edge (4) edge (5)
                (4) edge[redd] (5);
        \path (1) -- node (mid12) {} (2);
        \node [left = of mid12] {$t=1$};
    \end{tikzpicture}
\end{center}
Then neither value changes as they both already have the same opinion:
\begin{center}
    \begin{tikzpicture}
        \node[main] (1) {4};
        \node[main] (2) [below = of 1] {3};
        \node[main] (3) [right = of 1] {6};
        \node[main] (4) [right = of 2] {5};
        \node[main] (5) [right = of 3] {5};
        \path   (1) edge (2) edge (3) edge (4)
                (2) edge (4)
                (3) edge (4) edge (5)
                (4) edge (5);
        \path (1) -- node (mid12) {} (2);
        \node [left = of mid12] {$t=2$};
    \end{tikzpicture}
\end{center}
Next, suppose the following edge is chosen at $t=2$:
\begin{center}
    \begin{tikzpicture}
        \node[main, redd] (1) {4};
        \node[main, redd] (2) [below = of 1] {3};
        \node[main] (3) [right = of 1] {6};
        \node[main] (4) [right = of 2] {5};
        \node[main] (5) [right = of 3] {5};
        \path   (1) edge[redd] (2) edge (3) edge (4)
                (2) edge (4)
                (3) edge (4) edge (5)
                (4) edge (5);
        \path (1) -- node (mid12) {} (2);
        \node [left = of mid12] {$t=2$};
    \end{tikzpicture}
\end{center}
Then, $3$ is brought one step closer to $4$, so it becomes $4$. Similarly, $4$ is brought one step closer to $3$, so it becomes $3$, and hence the values swap:
\begin{center}
    \begin{tikzpicture}
        \node[main] (1) {3};
        \node[main] (2) [below = of 1] {4};
        \node[main] (3) [right = of 1] {6};
        \node[main] (4) [right = of 2] {5};
        \node[main] (5) [right = of 3] {5};
        \path   (1) edge (2) edge (3) edge (4)
                (2) edge (4)
                (3) edge (4) edge (5)
                (4) edge (5);
        \path (1) -- node (mid12) {} (2);
        \node [left = of mid12] {$t=3$};
    \end{tikzpicture}
\end{center}
Suppose the process continues as follows:
\begin{center}
    \begin{tikzpicture}
        \node[main] (1) at (0,0) {3};
        \node[main] (2) [below = of 1] {4};
        \node[main] (3) [right = of 1] {6};
        \node[main] (4) [right = of 2] {5};
        \node[main] (5) [right = of 3] {5};
        \path   (1) edge (2) edge (3) edge (4)
                (2) edge (4)
                (3) edge (4) edge (5)
                (4) edge (5);
        \path (1) -- node (mid12) {} (2);
        \node [left = of mid12] {$t=3$};
    
        \node[main, redd] (1) at (14em,0) {3};
        \node[main] (2) [below = of 1] {4};
        \node[main] (3) [right = of 1] {6};
        \node[main, redd] (4) [right = of 2] {5};
        \node[main] (5) [right = of 3] {5};
        \path   (1) edge (2) edge (3) edge[redd] (4)
                (2) edge (4)
                (3) edge (4) edge (5)
                (4) edge (5);
    \end{tikzpicture} \\
    \begin{tikzpicture}
        \node[main] (1) at (0,0) {4};
        \node[main] (2) [below = of 1] {4};
        \node[main] (3) [right = of 1] {6};
        \node[main] (4) [right = of 2] {4};
        \node[main] (5) [right = of 3] {5};
        \path   (1) edge (2) edge (3) edge (4)
                (2) edge (4)
                (3) edge (4) edge (5)
                (4) edge (5)
                (1) -- node[pos=1] {} +(0, \picgap);
        \path (1) -- node (mid12) {} (2);
        \node [left = of mid12] {$t=4$};
        
        \node[main] (1) at (14em,0) {4};
        \node[main] (2) [below = of 1] {4};
        \node[main, redd] (3) [right = of 1] {6};
        \node[main, redd] (4) [right = of 2] {4};
        \node[main] (5) [right = of 3] {5};
        \path   (1) edge (2) edge (3) edge (4)
                (2) edge (4)
                (3) edge[redd] (4) edge (5)
                (4) edge (5);
        
        \node[main] (1) at (0,-7.8em) {4};
        \node[main] (2) [below = of 1] {4};
        \node[main] (3) [right = of 1] {5};
        \node[main] (4) [right = of 2] {5};
        \node[main] (5) [right = of 3] {5};
        \path   (1) edge (2) edge (3) edge (4)
                (2) edge (4)
                (3) edge (4) edge (5)
                (4) edge (5);
        \path (1) -- node (mid12) {} (2);
        \node [left = of mid12] {$t=5$};
    \end{tikzpicture}
\end{center}
Notice that some of the updates either did nothing or swapped the values, while other updates brought values closer together. Motivated by this observation, we are now ready to categorise the update types:
\begin{defn}[Types of updates]
    Let $G=(V,E)$ be an undirected graph, and suppose we are looking to update $f_t$ to get $f_{t+1}$. Then an update on $(u,v)\in E$, where $a\coloneqq f_t(u)$ and $b\coloneqq f_t(v)$, is
    \begin{itemize}
        \item a swap update if $|a-b|\leq 1$, and
        \item a shrink update if $|a-b|>1$.
    \end{itemize}
\end{defn}
Note that when $|a-b|=0$, i.e.\ $a=b$, we can imagine the values as having swapped even though there was no change.

Notice that for the example graph above at $t=5$, it is impossible for a shrink update to occur ever again, so it is no longer possible for the values of the nodes to change, up to reordering; we have entered an isomorphism class which cannot be left. In other words, there will be two $4$s and three $5$s for all $t\geq 5$ regardless of the update sequence. Recalling \Cref{absorb}, which said that finite Markov chains eventually reach absorbing states, we might come up with the idea that absorbing states are always like this, i.e. classes of isomorphic valuations from which further shrink updates are impossible. This turns out to be correct, but as of now, we are unsure if the Markov chain is even finite; there could be some update sequence that changes the graph in infinitely many ways, never reaching a stable state. We would also like to know what the possible final states look like, and how many there are. The next section introduces techniques that allow us to resolve these questions.

\vspace{3.5em}
\section{Convergence Behaviour}\label{sec:convBehaviour}
To develop an understanding of the structure of the Markov chain, and whether it is finite, we should consider what parts of the process stay the same, and what parts change. This brings us to \emph{invariants} and \emph{potential functions}. Invariants do not change in each individual update, meaning that they never change overall. Potential functions either always increase or always decrease in each update, meaning some property of the graph changes in a predictable way over time. We shall see that there are natural examples of both invariants and potential functions for this process, and that applying them gives us a detailed understanding of the SCCs and absorbing states of the Markov chain. From this, we can prove important results for understanding this model.

The next two lemmas are simple to prove, but taken together give us a powerful understanding of the convergent properties of this process.
\begin{lem}\label{sumInv}
    Let $G=(V,E)$ be a (not necessarily connected) graph. Then the sum $\sum_{v\in V}f_t(v)$ of its values is invariant of $t$.
\end{lem}
\begin{proof}
    Suppose we are at time $t$ and edge $(u,v)\in E$ is selected to be updated. Let $a\coloneqq f_t(u)$ and $b\coloneqq f_t(v)$.
    \begin{itemize}
        \item If we perform a swap update, then $a$ and $b$ switch places, and the values of all other vertices remain the same. Then $f_{t+1}(u)=b$ and $f_{t+1}(v)=a$, so:
        \begin{align*}
            \sum_{w\in V}f_{t+1}(w) &= b + a + \sum_{w\in V\setminus\{u,v\}}f_{t+1}(w) \\
            &= a + b + \sum_{w\in V\setminus\{u,v\}}f_t(w) \\
            &= \sum_{w\in V}f_t(w)
        \end{align*}

        \item If we perform a shrink update, then $|a-b|>1$. Without loss of generality (WLOG), assume $a+1<b$. Then $f_{t+1}(u)=a+1$ and $f_{t+1}(v)=b-1$, so:
        \begin{align*}
            \sum_{w\in V}f_{t+1}(w) &= (a+1) + (b-1) + \sum_{w\in V\setminus\{u,v\}}f_{t+1}(w) \\
            &= a + b + \sum_{w\in V\setminus\{u,v\}}f_t(w) \\
            &= \sum_{w\in V}f_t(w)
        \end{align*}
    \end{itemize}
    Thus, no matter the update, the sum remains the same, and hence by induction, is invariant with respect to $t$.
\end{proof}
Now, we want a potential function that describes how close the values are together. This will remain the same after swap updates, since the values don't get any closer, but should either strictly increase or strictly decrease after shrink updates, such as the following:
\[ (1,9) \to (2,8) \to (3,7) \to (4,6) \to (5,5) \]
A natural choice here is the sum of each value squared. This leads us to the following:
\begin{lem}
    Let $G=(V,E)$ be a graph at time $t$, and suppose a shrink update is performed on $(u,v)\in E$. Then if $a\coloneqq f_t(u)$ and $b\coloneqq f_t(v)$, then the square sum $\sum_{v\in V}f_t^2(v)$ decreases by $2(|a-b|-1)\geq 2$.
\end{lem}
\begin{proof}
    Suppose WLOG that $b+1<a$, so $f_{t+1}(u)=a-1$ and $f_{t+1}(v)=b+1$. All other values remain the same, so:
    \begin{align*}
        \sum_{w\in V}f_{t+1}^2(w) &= (a-1)^2 + (b+1)^2 + \sum_{w\in V\setminus\{u,v\}}f_{t+1}^2(w) \\
        &= a^2-2a+1+b^2+2b+1 + \sum_{w\in V\setminus\{u,v\}}f_{t+1}^2(w) \\
        &= 2(b-a+1) + a^2 + b^2 + \sum_{w\in V\setminus\{u,v\}}f_t^2(w) \\
        &= 2((b+1)-a) + \sum_{w\in V}f_t^2(w)
    \end{align*}
    Since $(b+1)-a<0$, the square sum decreases by $2(a-b-1)=2(|a-b|-1)\geq 2$. Assuming $a+1<b$ instead yields the same result.
\end{proof}
Since the square sum is invariant after swap updates, we also get the following:
\begin{cor}
    If $G=(V,E)$ is a graph, then the square sum $\sum_{v\in V}f_t^2(v)$ is monotonically decreasing with respect to $t$. \qed
\end{cor}
The monotonically decreasing square sum potential function is the key to proving that we will reach an absorbing state with high probability. It means shrink updates are irreversible, as once the square sum is lowered, it cannot be raised again. That is, shrink updates must move us to a lower SCC in the hierarchy. If $Q$ is the initial square sum, then:
\[ Q=\sum_{v\in V}f_0^2(v)\geq\sum_{v\in V}f_1^2(v)\geq\cdots\geq 0 \]
Since each shrink update decreases the square sum by at least $2$, this means:
\begin{cor}\label{finShrink}
    Let $G=(V,E)$ be a graph and $f_0$ be an initial valuation. If the initial square sum is $Q=\sum_{v\in V}f_0^2(v)$ then at most $\lfloor Q/2\rfloor$ shrink updates are possible. \qed
\end{cor}
The final aggregations to consider are the minimum and maximum of the graph. For simplicity, we can just consider the most extreme opinion, which is the one with the greatest \emph{absolute value}. A key observation is it is impossible for the most extreme opinion in a network to become more extreme, since there is no one more extreme to move towards.
\begin{lem}
    If $G=(V,E)$ is a graph, then $\max_{v\in V}|f_t(v)|$ is monotonically decreasing with respect to $t$.
\end{lem}
\begin{proof}
    Suppose the maximum opinion is $M$ at time $t$. The only way for the maximum opinion to increase is for an $M$ to increase to an $M+1$. This is only possible if the $M$ combines with an opinion greater than it, which is impossible since $M$ is the maximum opinion. The same logic applies for the minimum.
\end{proof}
Since we have bounded the range of possible valuations on both sides, we can finally show that the Markov chain must be finite.
\begin{prop}\label{markovFiniteLB}
    Let $G=(V,E)$ be a graph and $f_0$ be an initial valuation, then let $M=\max_{v\in V}|f_0(v)|$. Then only finitely many valuations are reachable from $f_0$.
\end{prop}
\begin{proof}
    Since $\max_{v\in V}|f_t(v)|$ is monotonically decreasing, any reachable valuation has $-M\leq f_t(v)\leq M$ for all $v\in V$. This range covers $2M+1$ values, and there are $(2M+1)^n$ ways to assign each node a value in this range.
\end{proof}
Alternatively, this follows as a corollary from the fact that only finitely many shrink updates are possible, and swap updates can get you to at most $n!$ isomorphic valuations.

This means if we only include reachable valuations, the Markov chain is finite, so by \Cref{absorb}, w.h.p.\ we will end up in an absorbing state. Thus, we turn our attention towards analysing the properties of these absorbing states. First, let's show that shrink updates are impossible within SCCs.
\begin{prop}\label{swapSCC}
    Let $G=(V,E)$ be a graph. Valuations $\alpha$ and $\beta$ are in the same SCC (i.e.\ $\alpha S\beta$) if and only if $\alpha$ and $\beta$ can reach each other using only swap updates and all paths from $\alpha$ to $\beta$ (and vice versa) consist of only swap updates.
\end{prop}
\begin{proof}
    If $\alpha$ and $\beta$ can reach each other using only swap updates then they can reach each other, so $\alpha S\beta$. Conversely, suppose $\alpha S\beta$. Then there is a path of updates from $\alpha$ to $\beta$. If there were a shrink update along this path, then the square sum would decrease:
    \[ \sum_{v\in V}\beta^2(v) < \sum_{v\in V}\alpha^2(v) \]
    But this would contradict that there is a path of updates from $\beta$ to $\alpha$, and thus all paths from $\alpha$ to $\beta$ consist only of swap updates. Similarly, all paths from $\beta$ to $\alpha$ consist only of swap updates.
\end{proof}
Now, as a quick lemma, we show that applying a swap update does not change the isomorphism class.
\begin{lem}\label{swapIso}
    Let $G=(V,E)$ be a graph and suppose a swap update is applied at time $t$. Then $f_t$ and $f_{t+1}$ are isomorphic.
\end{lem}
\begin{proof}
    A swap update swaps the values of $u,v\in E$ and keeps all other vertices the same. Then the following is an isomorphism between $f_t$ and $f_{t+1}$:
    \begin{align*}
        I &: V\to V \\
        I(w) &= \begin{cases}
            v, & w = u \\
            u, & w = v \\
            w, & \text{otherwise}
        \end{cases}
    \end{align*}
\end{proof}
The previous two results give us a key fact about absorbing states.
\begin{prop}\label{absorbSwap}
    $f_t$ is in an absorbing state if and only if shrink updates are no longer possible. \qed
\end{prop}
\begin{proof}
    If $f_t$ is in an absorbing state and a shrink update were possible, the path through this shrink update would contradict \Cref{swapSCC}. Conversely, if shrink updates are no longer possible then all future updates are swap updates. Any swap update can be reversed by simply swapping that edge again, so it is impossible to leave the SCC.
\end{proof}
Note that in general, it may not be possible to reach everything in an isomorphism class. For example, both of the following valuations are isomorphic, but it is impossible to reach one from the other without performing a shrink update.
\begin{center}
    \begin{tikzpicture}[node distance = 4em]
        \node[main] (1) at (0, 0) {1};
        \node[main] (3) [right = of 1] {3};
        \path (1) edge (3);

        \node[main] (3) at (16em, 0) {3};
        \node[main] (1) [right = of 3] {1};
        \path (3) edge (1);
    \end{tikzpicture}
\end{center}
However it turns out shrink updates are the only thing that can ruin this; in absorbing states where shrink updates are impossible, we can prove a useful equivalence:
\begin{prop}\label{absorbEquiv}
    Let $\alpha, \beta$ be valuations on $G=(V,E)$ in absorbing states, i.e.\ no further shrink updates are possible from either. Then the following are equivalent:
    \begin{enumerate}[label=(\roman*)]
        \item $\alpha S\beta$, i.e.\ $\alpha$ and $\beta$ are in the same SCC,
        \item $\alpha$ and $\beta$ can reach each other using only swap updates and all paths from $\alpha$ to $\beta$ (and vice versa) consist of only swap updates, and
        \item $\alpha$ and $\beta$ are isomorphic.
    \end{enumerate}
\end{prop}
\begin{proof}
    By \Cref{swapSCC} we know \textit{(i)} and \textit{(ii)} are equivalent.

    \textit{(ii)$\implies$(iii)} \\
    Consider a path of swap updates from $\alpha$ to $\beta$. By \Cref{swapIso}, for each swap update there is an isomorphism. Then the composition of these isomorphisms is an overall isomorphism between $\alpha$ and $\beta$.
    
    \textit{(iii)$\implies$(i)} \\
    We first prove this fact by induction for the special case where $G$ is a tree. It is clearly true for trees of size $1$. Now, suppose the claim were true for trees of size $n-1\geq 1$ and consider any tree with $n$ nodes. Since $n\geq 2$ there exists a leaf $\ell\in V$ (i.e.\ a node of degree $1$). Let $I:V\to V$ be an isomorphism from $\alpha$ to $\beta$, so $\beta(\ell)=\alpha(I(\ell))$. Since trees are connected, let $I(\ell),w_1,\cdots,w_k,\ell$ be a path from $I(\ell)$ to $\ell$ in $G$. Then from $\alpha$, we can perform the following sequence of updates to get $\alpha'$:
    \begin{gather*}
        (I(\ell), w_1) \\
        (w_1, w_2) \\
        \vdots \\
        (w_{k-1}, w_k) \\
        (w_k, \ell)
    \end{gather*}
    Since we are in an absorbing state, by \Cref{absorbSwap} these are all swap updates. Then the value $\alpha(I(\ell))$ makes its way along the path until it reaches $\ell$, so $\alpha'(\ell)=\alpha(I(\ell))$. In other words, we have swapped the correct value (according to $\beta$) to $\ell$. Now, removing $\ell$ from the tree $G$ induces a new tree. Since $\alpha'$ was reached from $\alpha$ by swap updates, it is isomorphic to $\alpha$ and thus isomorphic to $\beta$. By the inductive hypothesis, we can reach the induced $\beta$ from the induced $\alpha'$ in this new tree using only swap updates. This doesn't change the value in $\ell$, so we get that for all $v\in V$:
    \[ \beta(v) = \alpha(I(v)) \]
    Thus, we have shown the fact for all trees. But this also shows the fact for connected graphs, since we can simply take a spanning tree and ignore the other edges. Finally, for general (not necessarily connected graphs), we can apply this argument to each connected component, and note that a path from $I(v)$ to $v$ still exists because $I(v)$ is guaranteed to be in the same connected component as $v$ for any isomorphism $I$.
\end{proof}
This means not only are all valuations in an absorbing state isomorphic, but also in an absorbing state, one can reach all isomorphic valuations eventually by applying enough swap updates.

Recall in our motivating example that we landed in an absorbing state, and the reason why further shrink updates were impossible was because all values were $4$ and $5$. It turns out that this is a general fact which applies to all absorbing states.
\begin{lem}\label{within1}
    Let $G=(V,E)$ and $f_t$ be a valuation. Then $f_t$ is in an absorbing state if and only if $|f_t(u)-f_t(v)|\leq 1$ for all $u,v\in V$ in the same connected component.
\end{lem}
\begin{proof}
    \phantom{a}

    $(\impliedby)$ \\
    Suppose $|f_t(u)-f_t(v)|\leq 1$ for all $u,v\in V$ in the same connected component. Then any update is between two vertices with difference less than or equal to $1$, making it a swap update, so $f_t$ and $f_{t+1}$ are isomorphic by \Cref{swapIso}, i.e.\ there exists $I:V\to V$ such that $f_{t+1}(v)=f_t(I(v))$. Then for all $u,v\in V$ in the same connected component, we also know $I(u)$ and $I(v)$ are in the same connected component, so $|f_t(I(u))-f_t(I(v))| \leq 1$, and thus $|f_{t+1}(u)-f_{t+1}(v)| \leq 1$. Continuing by induction, this fact will hold forever, so all future updates must be swap updates and thus we are in an absorbing state by \Cref{absorbSwap}.

    $(\implies)$ \\
    Conversely, suppose there exists $u,v\in V$ in the same connected component such that $|f_t(u)-f_t(v)|>1$. Then there is a path $u,w_1,\cdots,w_k,v$. We want to show we can perform a shrink update, and this is achieved by performing the following sequence of updates:
    \begin{gather*}
        (u, w_1) \\
        (w_1, w_2) \\
        \vdots \\
        (w_{k-1}, w_k)
    \end{gather*}
    If any of these are shrink updates then we are done, so assume they are all swap updates. Then we swapped the value of $f_t(u)$ to $w_k$ (i.e.\ $f_{t+k}(w_k)=f_t(u)$). But then $(w_k, v)$ is a shrink update, since $f_{t+k}(v)=f_t(v)$ and $|f_t(u)-f_t(v)|>1$. Thus, we could perform a shrink update, so $f_t$ was not in an absorbing state by \Cref{absorbSwap}.
\end{proof}

With this, we know that any node in any connected component for an absorbing state must have one of two adjacent opinions, but we are still unsure what those two opinions could be. It turns out there is only one possibility, and we can use the sum invariant from \Cref{sumInv} to show this. We are now ready to bring all the pieces together to prove the key result which gives us a complete understanding of the convergence behaviour. For simplicity, we present the result for connected graphs first.
\begin{thm}\label{divisionConnected}
    Let $G=(V,E)$ be a connected graph and $f_0$ be an initial valuation. Let $S=\sum_{v\in V}f_0(v)$ be the initial sum of the nodes' values. Using the division algorithm, we can write $S=kn+r$ for some $0\leq r<n$. Then w.h.p.\ $G$ ends up in the unique absorbing state consisting of the isomorphism class of valuations with $r$ copies of $k+1$ and $n-r$ copies of $k$.
\end{thm}
\begin{proof}
    Start at $f_0$. By \Cref{markovFiniteLB}, the Markov chain of reachable valuations is finite, so by \Cref{absorb}, we eventually reach an absorbing state where no further shrink updates are possible. Let $t$ be a time that this happens. By \Cref{within1}, and noting that $G$ is connected, $f_t$ satisfies $|f_t(u)-f_t(v)|\leq 1$ for all $u,v\in V$, meaning there exists $k$ such that $f_t(u)\in\{k,k+1\}$ for all $u\in V$. Let $r$ be the number of copies of $k+1$, so there are $n-r$ copies of $k$. We can insist that $r\neq n$ since that corresponds to the case where all values in the graph are $k+1$, in which case we can increase $k$ and set $r=0$. Finally:
    \begin{align*}
        S &= \sum_{v\in V}f_0(v) \\
        &= \sum_{v\in V}f_t(v) \\
        &= r(k+1) + (n-r)k \\
        &= nk + r
    \end{align*}
    But the only values of $k,r$ such that $0\leq r<n$ and $S=kn+r$ are those obtained by the division algorithm. These are all isomorphic, and so by \Cref{absorbEquiv}, they must all be in one unique SCC.
\end{proof}
Now, for disconnected graphs, the same proof allows us to consider each connected component separately:
\begin{thm}\label{divisionAny}
    Let $G=(V,E)$ be any graph and $f_0$ be an initial valuation. By possibly reordering vertex labels, let $v_1,\cdots,v_\ell$ be representatives of each connected component. For $1\leq i\leq\ell$ let $[v_i]$ be the set of vertices connected to $v_i$, let $n_i=|[v_i]|$ and let $S_i=\sum_{v\in[v_i]}f_0(v)$. Using the division algorithm, write $S_i=k_in_i+r_i$ where $0\leq r_i\leq n_i$. Then w.h.p.\ $G$ ends up in the unique absorbing state consisting of the isomorphism class of valuations where for $1\leq i\leq\ell$, the connected component $[v_i]$ has $r_i$ copies of $k_i+1$ and $n_i-r_i$ copies of $k_i$. \qed
\end{thm}
Therefore, regardless of the path taken through the Markov chain, there is only one possible absorbing state, giving us a complete understanding of the result of the process. For both results, note that w.h.p.\ all states in this unique final isomorphism class are reached an infinite number of times. But how many states are in this isomorphism class? Considering our motivating example, there are five nodes; two of them have to be $4$ and the other three have to be $5$, so there are $\binom{5}{3}=10$ ways to achieve this.
\begin{cor}
    The period of a connected graph $G=(V,E)$ with initial valuation $f_0$ is $p=\binom{n}{r}$ where $r$ is the value in \Cref{divisionConnected}. \qed
\end{cor}
For disconnected graphs, we can consider each connected component separately. Each connected component will have $\binom{n}{r}$ reachable valuations for its own values of $n$ and $r$. Then for each of these valuations, any of the possible valuations for the second connected component are possible, so we multiply the results. Extending to all connected components, the result is the product.
\begin{cor}
    The period of any graph $G=(V,E)$ with initial valuation $f_0$ is $p=\prod_{i=1}^\ell\binom{n_i}{r_i}$, where $\ell,n_i,r_i$ are the values in \Cref{divisionAny}. \qed
\end{cor}

Now that we know the convergence state, we can consider the convergence time, which is simply the first moment that this absorbing state is reached, that is, the first time that the minimum and maximum opinions in the valuation are at most one apart. In the next section, we turn our attention to bounding this convergence time in the average case, and presenting a particular `worst-case' example for deeper analysis.

\section{Convergence Time}\label{sec:convTime}
The last section showed that the convergence process consists of two things: bouncing around in an SCC by making swap updates, and decreasing the square sum by making shrink updates. We can imagine this as the \emph{two dimensions} of convergence, and visualise the process as moving to the right via swap updates, then moving downwards via shrink updates. Viewed like this, the process spirals downwards like a ball rolling around in a funnel, giving the convergence a `width' and `height' which can be studied separately, and later multiplied to give the overall bound.
\begin{center}
    \begin{tikzpicture}[node distance = 4em]
        \node (1) {$f_1$};
        \node (0) [left = of 1] {$f_0$};
        \node (2) [right = of 1] {$\cdots$};
        \node (3) [right = of 2] {$f_{k_1-1}$};
        \node (10) [below = of 0] {$f_{k_1}$};
        \node (11) [right = of 10] {$f_{k_1+1}$};
        \node (12) [right = of 11] {$\cdots$};
        \node (13) [right = of 12] {$f_{k_2-1}$};
        \node (20) [below = of 10] {$f_{k_2}$};
        \path (20) -- node[pos=1] {$\vdots$} +(0, -2em);
        \node (21) [below = of 11] {$\vdots$};
        \node (22) [below = of 12] {$\vdots$};
        \node (23) [below = of 13] {$\vdots$};
        \node (30) [below = of 20] {$f_{k_{\ell-1}}$};
        \node (31) [right = of 30] {$f_{k_{\ell-1}+1}$};
        \node (32) [right = of 31] {$\cdots$};
        \node (33) [right = of 32] {$f_{k_\ell-1}$};
        \node (40) [below = of 30] {$f_{k_\ell}$};
        \node (41) [right = of 40] {$f_{k_\ell+1}$};
        \node (42) [right = of 41] {$f_{k_\ell+2}$};
        \node (43) [right = of 42] {$f_{k_\ell+3}$};
        \node (44) [right = of 43] {$\cdots$};

        \path[->]   (0) edge node[above, sloped] {swap} (1)
                    (1) edge node[above, sloped] {swap} (2)
                    (2) edge node[above, sloped] {swap} (3)
                    (3) edge node[above, sloped] {shrink} (10)
                    (10) edge node[below, sloped] {swap} (11)
                    (11) edge node[above, sloped] {swap} (12)
                    (12) edge node[above, sloped] {swap} (13)
                    (13) edge node[above, sloped] {shrink} (20)
                    (23) edge node[above, sloped] {shrink} (30)
                    (30) edge node[below, sloped] {swap} (31)
                    (31) edge node[above, sloped] {swap} (32)
                    (32) edge node[above, sloped] {swap} (33)
                    (33) edge node[above, sloped] {shrink} (40)
                    (40) edge node[below, sloped] {swap} (41)
                    (41) edge node[above, sloped] {swap} (42)
                    (42) edge node[above, sloped] {swap} (43)
                    (43) edge node[above, sloped] {swap} (44);

        \def\hdist{4em};
        \def\vdist{2.5em};
        \path       (0) -- coordinate (L0) +(-\hdist, -\vdist)
                    (3) -- coordinate (R0) +(\hdist, \vdist)
                    (10) -- coordinate (L1) +(-\hdist, -\vdist)
                    (13) -- coordinate (R1) +(\hdist, \vdist)
                    (20) -- coordinate (L2) +(-\hdist, -\vdist)
                    (23) -- coordinate (R2) +(\hdist, \vdist)
                    (30) -- coordinate (L3) +(-\hdist, -\vdist)
                    (33) -- coordinate (R3) +(\hdist, \vdist)
                    (40) -- coordinate (L4) +(-\hdist, -\vdist)
                    (44) -- coordinate (R4) +(\hdist, \vdist);

        \draw[scc, rounded corners = 1.5em] (L0) rectangle (R0)
                                            (L1) rectangle (R1)
                                            (L3) rectangle (R3);
        \draw[sccred, rounded corners = 1.5em] (L4) rectangle (R4);
    \end{tikzpicture}
\end{center}
Note that the height is guaranteed to be finite by \Cref{finShrink}, and the final red state is the absorbing state. The convergence time in this example is $k_\ell$.

Bounding the width, the expected time taken to drop from one SCC to the next, turns out to be the bigger challenge, requiring a detour into the gambler's ruin problem and induction proofs with Markov chains. Instead, we begin our analysis with the height, the expected number of SCCs visited, i.e.\ the expected number of shrink updates.

\subsection{The Height}\label{sec:height}
We showed in \Cref{finShrink} that at most $\lfloor Q/2\rfloor$ shrink updates are possible, where $Q$ is the initial square sum. However, consider the following example:
\begin{center}
    \begin{tikzpicture}[bigg/.style={minimum size = 2.5em}]
        \node[main,bigg] (1) at (0,0) {4};
        \node[main,bigg] (2) [below = of 1] {2};
        \node[main,bigg] (3) [right = of 1] {6};
        \node[main,bigg] (4) [right = of 2] {6};
        \node[main,bigg] (5) [right = of 3] {5};
        \path   (1) edge (2) edge (3) edge (4)
                (2) edge (4)
                (3) edge (4) edge (5)
                (4) edge (5);
        \path (4) -- node[pos=1] {Initial Square Sum: 117} +(0, -2.5em);

        \node[main] (1) at (16em,0) {104};
        \node[main] (2) [below = of 1] {102};
        \node[main] (3) [right = of 1] {106};
        \node[main] (4) [right = of 2] {106};
        \node[main] (5) [right = of 3] {105};
        \path   (1) edge (2) edge (3) edge (4)
                (2) edge (4)
                (3) edge (4) edge (5)
                (4) edge (5);
        \path (4) -- node[pos=1] {Initial Square Sum: 54717} +(0, -2.5em);
    \end{tikzpicture}
\end{center}
Here, we have taken our motivating example and added 100 to all values. The process is clearly equivalent as the relative differences between the nodes remain unchanged, but the initial square sum is much higher. The problem here is that our $\lfloor Q/2\rfloor$ bound assumes that the square sum can be lowered to zero, which is not always the case. Instead, we should be considering the difference between the initial and final square sums. Fortunately, since we know the absorbing state, we know exactly what the final square sum will be.
\begin{defn}[Square sum gap]
    Let $G=(V,E)$ be a graph and $f_0$ its initial valuation. Let
    \begin{align*}
        Q_{init} &\coloneqq \sum_{v\in V}f_0^2(v) \\
        Q_{final} &\coloneqq \sum_{i=1}^\ell \left(r_i(k_i+1)^2+(n_i-r_i)k_i^2\right)
    \end{align*}
    where $\ell$, $n_i$, $k_i$ and $r_i$ are as in \Cref{divisionAny}. Then the \emph{square sum gap} of $G$ is:
    \[ q=\frac{Q_{init}-Q_{final}}{2} \]
\end{defn}
We can apply this to our example to calculate the final square sums:
\begin{center}
    \begin{tikzpicture}[bigg/.style={minimum size = 2.5em}]
        \node[main,bigg] (1) at (0,0) {4};
        \node[main,bigg] (2) [below = of 1] {4};
        \node[main,bigg] (3) [right = of 1] {5};
        \node[main,bigg] (4) [right = of 2] {5};
        \node[main,bigg] (5) [right = of 3] {5};
        \path   (1) edge (2) edge (3) edge (4)
                (2) edge (4)
                (3) edge (4) edge (5)
                (4) edge (5);
        \path (4) -- node[pos=1] {Final Square Sum: 107} +(0, -2.5em);

        \node[main] (1) at (16em,0) {104};
        \node[main] (2) [below = of 1] {104};
        \node[main] (3) [right = of 1] {105};
        \node[main] (4) [right = of 2] {105};
        \node[main] (5) [right = of 3] {105};
        \path   (1) edge (2) edge (3) edge (4)
                (2) edge (4)
                (3) edge (4) edge (5)
                (4) edge (5);
        \path (4) -- node[pos=1] {Final Square Sum: 54707} +(0, -2.5em);
    \end{tikzpicture}
\end{center}
Here, the square sum gap for both graphs is $q=5$, which is a much better bound for the maximum number of shrink updates applicable. Since applying a shrink update to $a$ and $b$ decreases the square sum by $2(|a-b|-1)$, this bound is reached when all shrink updates decrease the square sum by the minimal amount, which occurs when $|a-b|=2$. We didn't achieve this for our example since we updated an edge with endpoints $2$ and $6$ at the start, but we could achieve this by making the following sequence of updates:
\begin{center}
    \begin{tikzpicture}
        \node[main] (1) at (0,0) {4};
        \node[main] (2) [below = of 1] {2};
        \node[main] (3) [right = of 1] {6};
        \node[main] (4) [right = of 2] {6};
        \node[main] (5) [right = of 3] {5};
        \path   (1) edge (2) edge (3) edge (4)
                (2) edge (4)
                (3) edge (4) edge (5)
                (4) edge (5);
        \path (1) -- node (mid12) {} (2);
        \node [left = of mid12] {$t=0$};
        \path   (2) -- node (mid24) {} (4)
                (mid24) -- node[pos=1] {Square Sum: 117} +(0, -2.5em);
    
        \node[main, redd] (1) at (14em,0) {4};
        \node[main, redd] (2) [below = of 1] {2};
        \node[main] (3) [right = of 1] {6};
        \node[main] (4) [right = of 2] {6};
        \node[main] (5) [right = of 3] {5};
        \path   (1) edge[redd] (2) edge (3) edge (4)
                (2) edge (4)
                (3) edge (4) edge (5)
                (4) edge (5);
    \end{tikzpicture} \\
    \begin{tikzpicture}
        \node[main] (1) at (0,0) {3};
        \node[main] (2) [below = of 1] {3};
        \node[main] (3) [right = of 1] {6};
        \node[main] (4) [right = of 2] {6};
        \node[main] (5) [right = of 3] {5};
        \path   (1) edge (2) edge (3) edge (4)
                (2) edge (4)
                (3) edge (4) edge (5)
                (4) edge (5);
        \path (1) -- node (mid12) {} (2);
        \node [left = of mid12] {$t=1$};
        \path   (2) -- node (mid24) {} (4)
                (mid24) -- node[pos=1] {Square Sum: 115} +(0, -2.5em);
    
        \node[main] (1) at (14em,0) {3};
        \node[main] (2) [below = of 1] {3};
        \node[main, redd] (3) [right = of 1] {6};
        \node[main] (4) [right = of 2] {6};
        \node[main, redd] (5) [right = of 3] {5};
        \path   (1) edge (2) edge (3) edge (4)
                (2) edge (4)
                (3) edge (4) edge[redd] (5)
                (4) edge (5)
                (1) -- node[pos=1] {} +(0, \picgap);
    \end{tikzpicture} \\
    \begin{tikzpicture}
        \node[main] (1) at (0,0) {3};
        \node[main] (2) [below = of 1] {3};
        \node[main] (3) [right = of 1] {5};
        \node[main] (4) [right = of 2] {6};
        \node[main] (5) [right = of 3] {6};
        \path   (1) edge (2) edge (3) edge (4)
                (2) edge (4)
                (3) edge (4) edge (5)
                (4) edge (5);
        \path (1) -- node (mid12) {} (2);
        \node [left = of mid12] {$t=2$};
        \path   (2) -- node (mid24) {} (4)
                (mid24) -- node[pos=1] {Square Sum: 115} +(0, -2.5em);
    
        \node[main, redd] (1) at (14em,0) {3};
        \node[main] (2) [below = of 1] {3};
        \node[main, redd] (3) [right = of 1] {5};
        \node[main] (4) [right = of 2] {6};
        \node[main] (5) [right = of 3] {6};
        \path   (1) edge (2) edge[redd] (3) edge (4)
                (2) edge (4)
                (3) edge (4) edge (5)
                (4) edge (5)
                (1) -- node[pos=1] {} +(0, \picgap);
    \end{tikzpicture} \\
    \begin{tikzpicture}
        \node[main] (1) at (0,0) {4};
        \node[main] (2) [below = of 1] {3};
        \node[main] (3) [right = of 1] {4};
        \node[main] (4) [right = of 2] {6};
        \node[main] (5) [right = of 3] {6};
        \path   (1) edge (2) edge (3) edge (4)
                (2) edge (4)
                (3) edge (4) edge (5)
                (4) edge (5);
        \path (1) -- node (mid12) {} (2);
        \node [left = of mid12] {$t=3$};
        \path   (2) -- node (mid24) {} (4)
                (mid24) -- node[pos=1] {Square Sum: 113} +(0, -2.5em);
    
        \node[main] (1) at (14em,0) {4};
        \node[main] (2) [below = of 1] {3};
        \node[main, redd] (3) [right = of 1] {4};
        \node[main] (4) [right = of 2] {6};
        \node[main, redd] (5) [right = of 3] {6};
        \path   (1) edge (2) edge (3) edge (4)
                (2) edge (4)
                (3) edge (4) edge[redd] (5)
                (4) edge (5)
                (1) -- node[pos=1] {} +(0, \picgap);
    \end{tikzpicture} \\
    \begin{tikzpicture}
        \node[main] (1) at (0,0) {4};
        \node[main] (2) [below = of 1] {3};
        \node[main] (3) [right = of 1] {5};
        \node[main] (4) [right = of 2] {6};
        \node[main] (5) [right = of 3] {5};
        \path   (1) edge (2) edge (3) edge (4)
                (2) edge (4)
                (3) edge (4) edge (5)
                (4) edge (5);
        \path (1) -- node (mid12) {} (2);
        \node [left = of mid12] {$t=4$};
        \path   (2) -- node (mid24) {} (4)
                (mid24) -- node[pos=1] {Square Sum: 111} +(0, -2.5em);
    
        \node[main, redd] (1) at (14em,0) {4};
        \node[main] (2) [below = of 1] {3};
        \node[main] (3) [right = of 1] {5};
        \node[main, redd] (4) [right = of 2] {6};
        \node[main] (5) [right = of 3] {5};
        \path   (1) edge (2) edge (3) edge[redd] (4)
                (2) edge (4)
                (3) edge (4) edge (5)
                (4) edge (5)
                (1) -- node[pos=1] {} +(0, \picgap);
    \end{tikzpicture} \\
    \begin{tikzpicture}
        \node[main] (1) at (0,0) {5};
        \node[main] (2) [below = of 1] {3};
        \node[main] (3) [right = of 1] {5};
        \node[main] (4) [right = of 2] {5};
        \node[main] (5) [right = of 3] {5};
        \path   (1) edge (2) edge (3) edge (4)
                (2) edge (4)
                (3) edge (4) edge (5)
                (4) edge (5);
        \path (1) -- node (mid12) {} (2);
        \node [left = of mid12] {$t=5$};
        \path   (2) -- node (mid24) {} (4)
                (mid24) -- node[pos=1] {Square Sum: 109} +(0, -2.5em);
    
        \node[main, redd] (1) at (14em,0) {5};
        \node[main, redd] (2) [below = of 1] {3};
        \node[main] (3) [right = of 1] {5};
        \node[main] (4) [right = of 2] {5};
        \node[main] (5) [right = of 3] {5};
        \path   (1) edge[redd] (2) edge (3) edge (4)
                (2) edge (4)
                (3) edge (4) edge (5)
                (4) edge (5)
                (1) -- node[pos=1] {} +(0, \picgap);
                
        \node[main] (1) at (0,-10.4em) {4};
        \node[main] (2) [below = of 1] {4};
        \node[main] (3) [right = of 1] {5};
        \node[main] (4) [right = of 2] {5};
        \node[main] (5) [right = of 3] {5};
        \path   (1) edge (2) edge (3) edge (4)
                (2) edge (4)
                (3) edge (4) edge (5)
                (4) edge (5);
        \path (1) -- node (mid12) {} (2);
        \node [left = of mid12] {$t=6$};
        \path   (2) -- node (mid24) {} (4)
                (mid24) -- node[pos=1] {Square Sum: 107} +(0, -2.5em);
    \end{tikzpicture}
\end{center}
We know the square sum gap is an upper bound for the number of shrink updates, and have demonstrated an example of a graph with square sum gap $q=5$ that can actually take $5$ updates in the worst case. But what if $q$ is much larger; are there still graphs that can require $q$ updates in the worst case?

For simplicity, let's design graphs such that the sum is zero, so that the final square sum is also zero. The simplest case for $q=1$ is:
\begin{center}
    \begin{tikzpicture}[node distance = 4em, bigg/.style={minimum size = 2.5em}]
        \node[main,bigg] (1) at (0, 0) {$-1$};
        \node[main,bigg] (2) [right = of 1] {$1$};
        \path (1) edge (2);
    \end{tikzpicture}
\end{center}
Generalising by one level, we can make the $-1$ using a $-2$ and a $0$, and we can make the $1$ using a $0$ and a $2$. Thus, our new graph will have a $-2$, a $2$ and two copies of $0$. We might as well make it complete, so that all pairs can be swapped at any time.
\begin{center}
    \begin{tikzpicture}[node distance = 4em, bigg/.style={minimum size = 2.5em}]
        \node[main,bigg] (1) at (0, 0) {$-2$};
        \node[main,bigg] (2) [right = of 1] {$0$};
        \node[main,bigg] (3) [below = of 1] {$0$};
        \node[main,bigg] (4) [right = of 3] {$2$};
        \path   (1) edge (2) edge (3) edge (4)
                (2) edge (3) edge (4)
                (3) edge (4);
    \end{tikzpicture}
\end{center}
We can then perform shrink updates to get two copies of $-1$ and two copies of $1$, then shrink these with each other. Generalising one more level, we get the following graph of eight nodes:
\begin{center}
    \begin{tikzpicture}[bigg/.style={minimum size = 2.5em}]
        \def\dist{5.5em}
        \node[main,bigg] (0) at (90:\dist) {$-3$};
        \node[main,bigg] (1) at (45:\dist) {$-1$};
        \node[main,bigg] (2) at (0:\dist) {$-1$};
        \node[main,bigg] (3) at (-45:\dist) {$-1$};
        \node[main,bigg] (4) at (-90:\dist) {$3$};
        \node[main,bigg] (5) at (-135:\dist) {$1$};
        \node[main,bigg] (6) at (-180:\dist) {$1$};
        \node[main,bigg] (7) at (-225:\dist) {$1$};

        \path   (0) edge (1) edge (2) edge (3) edge (4) edge (5) edge (6) edge (7)
                (1) edge (2) edge (3) edge (4) edge (5) edge (6) edge (7)
                (2) edge (3) edge (4) edge (5) edge (6) edge (7)
                (3) edge (4) edge (5) edge (6) edge (7)
                (4) edge (5) edge (6) edge (7)
                (5) edge (6) edge (7)
                (6) edge (7);
    \end{tikzpicture}
\end{center}
There is a pattern here; if the nodes range from $-k$ to $k$, then we have $2^k$ nodes, and the counts of each node are the entries in the corresponding row in Pascal's Triangle. For the above example, we have \textbf{one} copy of $-3$, \textbf{three} copies of $-1$, \textbf{three} copies of $1$ and \textbf{one} copy of $3$, corresponding to the row \textbf{1 3 3 1} in Pascal's Triangle. This row can also be written $\binom{3}{0},\binom{3}{1},\binom{3}{2},\binom{3}{3}$. This is our general construction, and it can always be flattened to the absorbing state of all zeroes.
\begin{defn}[Binomial $n$-valuations]\label{defn:binVals}
    For $n\geq 0$, let $G$ be the complete graph on $2^n$ nodes. Then a \emph{binomial $n$-valuation} is a valuation with $\binom{n}{k}$ copies of $-n+2k$ for $0\leq k\leq n$.
\end{defn}
\begin{thm}\label{binValuationsThm}
    Any binomial $n$-valuation can reach an absorbing state using only shrink updates with $|a-b|=2$. The total number of shrink updates is $n2^{n-1}$.
\end{thm}
\begin{proof}
    We prove this by induction. The unique binomial $0$-valuation consists only of $\binom{0}{0}=1$ copy of $0$, so it is already in an absorbing state. This means it has reached an absorbing state in $0\cdot 2^{0-1}$ shrink updates. Let $n\geq 1$. We show any binomial $n$-valuation can reach an absorbing state in $n2^{n-1}$ shrink updates by showing that it can reach two parallel $(n-1)$-valuations, then counting the number of shrink updates.

    For $0\leq k\leq n-1$, combine $\binom{n-1}{k}$ copies of $-n+2k$ with $\binom{n-1}{k}$ copies of $-n+2k+2$ to get $2\binom{n-1}{k}$ copies of $-n+2k+1$. This is possible since for $0\leq j\leq n$, the value $-n+2j$ is combined for $k=j$ and $k=j-1$, so it is combined a total of $\binom{n-1}{j} + \binom{n-1}{j-1}=\binom{n}{j}$ times, which is exactly how many times it appears. Because of this, we have removed all of the original values, and are left with $2\binom{n-1}{k}$ copies of $-n+2k+1=-(n-1)+2k$ for all $0\leq k\leq n-1$.
    
    But then we can split the graph into two induced complete subgraphs with a binomial $(n-1)$-valuation in each, which by the inductive hypothesis, can each be absorbed in $(n-1)2^{n-2}$ shrink updates, resulting in $(n-1)2^{n-1}$ shrink updates total. Since we initially performed $\sum_{k=0}^{n-1}\binom{n-1}{k}=2^{n-1}$ shrink updates, we therefore performed a total of $n2^{n-1}$ shrink updates.
\end{proof}
For clarity, here is what the process looks like for a binomial $3$-valuation:
\begin{center}
    \begin{tikzpicture}[bigg/.style={minimum size = 2.5em}]
        \def\dist{5.5em}
        \node[main,bigg] (0) at (90:\dist) {$-3$};
        \node[main,bigg] (1) at (45:\dist) {$-1$};
        \node[main,bigg] (2) at (0:\dist) {$-1$};
        \node[main,bigg] (3) at (-45:\dist) {$-1$};
        \node[main,bigg] (4) at (-90:\dist) {$3$};
        \node[main,bigg] (5) at (-135:\dist) {$1$};
        \node[main,bigg] (6) at (-180:\dist) {$1$};
        \node[main,bigg] (7) at (-225:\dist) {$1$};

        \path   (0) edge (1) edge (2) edge (3) edge (4) edge (5) edge (6) edge (7)
                (1) edge (2) edge (3) edge (4) edge (5) edge (6) edge (7)
                (2) edge (3) edge (4) edge (5) edge (6) edge (7)
                (3) edge (4) edge (5) edge (6) edge (7)
                (4) edge (5) edge (6) edge (7)
                (5) edge (6) edge (7)
                (6) edge (7);

        \path (16em, 0) -- node[pos=1,main,bigg,redd] (0) {$-3$} +(90:\dist)
                        -- node[pos=1,main,bigg,redd] (1) {$-1$} +(45:\dist)
                        -- node[pos=1,main,bigg,redd] (2) {$-1$} +(0:\dist)
                        -- node[pos=1,main,bigg,redd] (3) {$-1$} +(-45:\dist)
                        -- node[pos=1,main,bigg,redd] (4) {$3$} +(-90:\dist)
                        -- node[pos=1,main,bigg,redd] (5) {$1$} +(-135:\dist)
                        -- node[pos=1,main,bigg,redd] (6) {$1$} +(-180:\dist)
                        -- node[pos=1,main,bigg,redd] (7) {$1$} +(-225:\dist);

        \path   (0) edge[redd] (1) edge (2) edge (3) edge (4) edge (5) edge (6) edge (7)
                (1) edge (2) edge (3) edge (4) edge (5) edge (6) edge (7)
                (2) edge (3) edge (4) edge (5) edge (6) edge[redd] (7)
                (3) edge (4) edge (5) edge[redd] (6) edge (7)
                (4) edge[redd] (5) edge (6) edge (7)
                (5) edge (6) edge (7)
                (6) edge (7);
    \end{tikzpicture} \\
    \begin{tikzpicture}[bigg/.style={minimum size = 2.5em}]
        \def\dist{5.5em}
        \node[main,bigg] (0) at (90:\dist) {$-2$};
        \node[main,bigg] (1) at (45:\dist) {$-2$};
        \node[main,bigg] (2) at (0:\dist) {$0$};
        \node[main,bigg] (3) at (-45:\dist) {$0$};
        \node[main,bigg] (4) at (-90:\dist) {$2$};
        \node[main,bigg] (5) at (-135:\dist) {$2$};
        \node[main,bigg] (6) at (-180:\dist) {$0$};
        \node[main,bigg] (7) at (-225:\dist) {$0$};

        \path   (0) edge (1) edge (2) edge (3) edge (4) edge (5) edge (6) edge (7)
                (1) edge (2) edge (3) edge (4) edge (5) edge (6) edge (7)
                (2) edge (3) edge (4) edge (5) edge (6) edge (7)
                (3) edge (4) edge (5) edge (6) edge (7)
                (4) edge (5) edge (6) edge (7)
                (5) edge (6) edge (7)
                (6) edge (7);

        \path (16em, 0) -- node[pos=1,main,bigg,redd] (0) {$-2$} +(90:\dist)
                        -- node[pos=1,main,bigg,redd] (1) {$-2$} +(45:\dist)
                        -- node[pos=1,main,bigg,redd] (2) {$0$} +(0:\dist)
                        -- node[pos=1,main,bigg,redd] (3) {$0$} +(-45:\dist)
                        -- node[pos=1,main,bigg,redd] (4) {$2$} +(-90:\dist)
                        -- node[pos=1,main,bigg,redd] (5) {$2$} +(-135:\dist)
                        -- node[pos=1,main,bigg,redd] (6) {$0$} +(-180:\dist)
                        -- node[pos=1,main,bigg,redd] (7) {$0$} +(-225:\dist);
        \node at (0, \dist+\picgap) {};

        \path   (0) edge (1) edge (2) edge (3) edge (4) edge (5) edge (6) edge[redd] (7)
                (1) edge[redd] (2) edge (3) edge (4) edge (5) edge (6) edge (7)
                (2) edge (3) edge (4) edge (5) edge (6) edge (7)
                (3) edge[redd] (4) edge (5) edge (6) edge (7)
                (4) edge (5) edge (6) edge (7)
                (5) edge[redd] (6) edge (7)
                (6) edge (7);
    \end{tikzpicture} \\
    \begin{tikzpicture}[bigg/.style={minimum size = 2.5em}]
        \def\dist{5.5em}
        \node[main,bigg] (0) at (90:\dist) {$-1$};
        \node[main,bigg] (1) at (45:\dist) {$-1$};
        \node[main,bigg] (2) at (0:\dist) {$1$};
        \node[main,bigg] (3) at (-45:\dist) {$-1$};
        \node[main,bigg] (4) at (-90:\dist) {$1$};
        \node[main,bigg] (5) at (-135:\dist) {$1$};
        \node[main,bigg] (6) at (-180:\dist) {$1$};
        \node[main,bigg] (7) at (-225:\dist) {$-1$};

        \path   (0) edge (1) edge (2) edge (3) edge (4) edge (5) edge (6) edge (7)
                (1) edge (2) edge (3) edge (4) edge (5) edge (6) edge (7)
                (2) edge (3) edge (4) edge (5) edge (6) edge (7)
                (3) edge (4) edge (5) edge (6) edge (7)
                (4) edge (5) edge (6) edge (7)
                (5) edge (6) edge (7)
                (6) edge (7);

        \path (16em, 0) -- node[pos=1,main,bigg,redd] (0) {$-1$} +(90:\dist)
                        -- node[pos=1,main,bigg,redd] (1) {$-1$} +(45:\dist)
                        -- node[pos=1,main,bigg,redd] (2) {$1$} +(0:\dist)
                        -- node[pos=1,main,bigg,redd] (3) {$-1$} +(-45:\dist)
                        -- node[pos=1,main,bigg,redd] (4) {$1$} +(-90:\dist)
                        -- node[pos=1,main,bigg,redd] (5) {$1$} +(-135:\dist)
                        -- node[pos=1,main,bigg,redd] (6) {$1$} +(-180:\dist)
                        -- node[pos=1,main,bigg,redd] (7) {$-1$} +(-225:\dist);
        \node at (0, \dist+\picgap) {};

        \path   (0) edge (1) edge (2) edge (3) edge (4) edge[redd] (5) edge (6) edge (7)
                (1) edge (2) edge (3) edge[redd] (4) edge (5) edge (6) edge (7)
                (2) edge[redd] (3) edge (4) edge (5) edge (6) edge (7)
                (3) edge (4) edge (5) edge (6) edge (7)
                (4) edge (5) edge (6) edge (7)
                (5) edge (6) edge (7)
                (6) edge[redd] (7);
    \end{tikzpicture} \\
    \begin{tikzpicture}[bigg/.style={minimum size = 2.5em}]
        \def\dist{5.5em}
        \node[main,bigg] (0) at (90:\dist) {$0$};
        \node[main,bigg] (1) at (45:\dist) {$0$};
        \node[main,bigg] (2) at (0:\dist) {$0$};
        \node[main,bigg] (3) at (-45:\dist) {$0$};
        \node[main,bigg] (4) at (-90:\dist) {$0$};
        \node[main,bigg] (5) at (-135:\dist) {$0$};
        \node[main,bigg] (6) at (-180:\dist) {$0$};
        \node[main,bigg] (7) at (-225:\dist) {$0$};

        \path   (0) edge (1) edge (2) edge (3) edge (4) edge (5) edge (6) edge (7)
                (1) edge (2) edge (3) edge (4) edge (5) edge (6) edge (7)
                (2) edge (3) edge (4) edge (5) edge (6) edge (7)
                (3) edge (4) edge (5) edge (6) edge (7)
                (4) edge (5) edge (6) edge (7)
                (5) edge (6) edge (7)
                (6) edge (7);
    \end{tikzpicture}
\end{center}
Since the final square sum is $0$ and the square sum decreased by $2$ each update, the initial square sum must have been $n2^n$. This can also be double-checked by differentiating
\[ (x+1)^n = \sum_{k=0}^n\binom{n}{k}x^k \]
with respect to $x$ to obtain
\begin{align*}
    \sum_{k=0}^nk\binom{n}{k} &= n2^{n-1}\text{ and} \\
    \sum_{k=0}^nk^2\binom{n}{k} &= (n+n^2)2^{n-2},
\end{align*}
then doing some algebra on the initial square sum to show
\[ \sum_{k=0}^n\binom{n}{k}(-n+2k)^2 = n2^n. \]
The offshoot of this is that, if we want to find a graph with a square sum gap of $q=n2^{n-1}$ such that there is an update sequence which exhausts that bound (i.e.\ requires $q$ shrink updates), we can find such a graph with $2^n$ nodes. Redefining $n$ to be the number of nodes, we get that $q=\Theta(n\log n)$, meaning if we want a graph with a square sum gap of $q$ such that an update sequence exhausts that bound, we can find one with $n$ nodes such that $q=\Theta(n\log n)$. Regardless of the numbers of nodes, this family of specific examples justifies the square sum invariant, since it shows its tightness as a bound for the height of the convergence process.

\subsection{The Width}\label{sec:width}
The worst case of the width is simple; given a graph, if all possible updates are shrink updates, then the width is 1, since we are guaranteed to exit the SCC on the next update. Otherwise, there is an edge which, when picked, causes a swap update, so the worst case width is $\infty$, since we can simply pick this edge forever.

The more interesting question is to analyse the width in the average case, in other words, the expected value of the width when the process is run randomly. If the only nodes that the shrink update can be performed on are particularly far apart in the graph, it may take a long time for the right sequence of swaps to occur to align them next to each other, before they are combined. Analysing the expected value is significantly more challenging than analysing the worst case, as our tools for analysing randomness are more limited. For this problem, it helps to consider a motivating example. We will take inspiration from the above idea of having only one possible shrink update, to simplify the example and ensure that there are not multiple interdependent possibilities. The graph will be filled with zeroes, aside from a $-1$ node and a $1$ node.
\begin{center}
    \begin{tikzpicture}[bigg/.style = {minimum size = 2.5em}]
        \node[main,bigg] (1) {$-1$};
        \node[main,bigg] (2) [below = of 1] {$0$};
        \node[main,bigg] (3) [right = of 1] {$0$};
        \node[main,bigg] (4) [right = of 2] {$0$};
        \node[main,bigg] (5) [right = of 3] {$0$};
        \node[main,bigg] (6) [right = of 4] {$1$};
        \node[main,bigg] (7) [right = of 5] {$0$};

        \path   (1) edge (2)
                (2) edge (3) edge (4)
                (3) edge (4) edge (5)
                (4) edge (5)
                (5) edge (6) edge (7)
                (6) edge (7);
    \end{tikzpicture}
\end{center}
To further simplify things, let's organise the nodes of the graph into a line. This can be done in a number of ways, but it helps to pick one which aligns with a path from $-1$ to $1$.
\begin{center}
    \begin{tikzpicture}[bigg/.style = {minimum size = 2.5em}]
        \node[main,bigg] (1) {$-1$};
        \node[main,bigg] (2) [below = of 1] {$0$};
        \node[main,bigg] (3) [right = of 1] {$0$};
        \node[main,bigg] (4) [right = of 2] {$0$};
        \node[main,bigg] (5) [right = of 3] {$0$};
        \node[main,bigg] (6) [right = of 4] {$1$};
        \node[main,bigg] (7) [right = of 5] {$0$};

        \path   (1) edge[bluu] (2)
                (2) edge[bluu] (3) edge (4)
                (3) edge[bluu] (4) edge (5)
                (4) edge[bluu] (5)
                (5) edge[bluu] (6) edge (7)
                (6) edge[bluu] (7) -- +(0, -\picgap);
    \end{tikzpicture} \\
    \begin{tikzpicture}[bigg/.style = {minimum size = 2.5em}]
        \node[main,bigg] (1) {$-1$};
        \node[main,bigg] (2) [right = of 1] {$0$};
        \node[main,bigg] (3) [right = of 2] {$0$};
        \node[main,bigg] (4) [right = of 3] {$0$};
        \node[main,bigg] (5) [right = of 4] {$0$};
        \node[main,bigg] (6) [right = of 5] {$1$};
        \node[main,bigg] (7) [right = of 6] {$0$};

        \path   (1) edge[bluu] (2)
                (2) edge[bluu] (3) edge[bend left = 35] (4)
                (3) edge[bluu] (4) edge[bend right = 35] (5)
                (4) edge[bluu] (5)
                (5) edge[bluu] (6) edge[bend left = 35] (7)
                (6) edge[bluu] (7);
    \end{tikzpicture}
\end{center}
As swap updates are performed and $-1$ and $1$ shuffle along the line, their positions vary. They may be pulled closer together, or pushed further apart. Eventually, we expect them to be placed close to each other at least once, and if the edge between them is picked, the shrink update finally happens. The width is the amount of time until this occurs.

This visualisation suggests a \emph{worst case} of this process, where $-1$ and $1$ are on opposite sides of a line, as far apart as possible, and there are no `shortcut' edges that can be used to jump ahead. Note that shortcut edges are not necessarily always good, they can send the numbers further away from each other, but intuition leads us to believe the expected value is maximised by removing them; justification for this is provided at the end of this section. This would create a straight-line graph, and the random process on this graph would be similar to the \emph{gambler's ruin} game. In this game, a gambler starts with $i$ dollars and either makes or loses a dollar on each iteration with equal probability. The game ends when the gambler runs out of money (i.e.\ has 0 dollars), or reaches $n$ dollars. It can be shown using Markov chains that the expected number of steps until the game ends is $i(n-i)$, and the probability of success of the gambler is $i/n$. If we make the following definition, then the technique used to prove this fact can be applied to our graph problem.
\begin{defn}[Gambler's ruin valuations]\label{defn:grVals}
    For $n\geq 2$, let $G$ be a straight-line graph of $n$ nodes (i.e.\ a tree with maximum degree 2). Then a \emph{gambler's ruin valuation} is a valuation which assigns $-1$ to one leaf, $1$ to the other leaf, and zero to all non-leaves.
\end{defn}
Below is an example of a gambler's ruin valuation for $n=5$.
\begin{center}
    \begin{tikzpicture}[bigg/.style={minimum size=2.5em}];
        \node[main,bigg] (1) {$-1$};
        \node[main,bigg] (2) [right = of 1] {$0$};
        \node[main,bigg] (3) [right = of 2] {$0$};
        \node[main,bigg] (4) [right = of 3] {$0$};
        \node[main,bigg] (5) [right = of 4] {$1$};
    
        \path       (1) edge (2)
                    (2) edge (3)
                    (3) edge (4)
                    (4) edge (5);
    \end{tikzpicture}
\end{center}
Even for this simplified graph, finding the expected number of swaps before the shrink update is tricky because both the $-1$ and the $1$ can move. This means the Markov chain has $\Theta(n^2)$ states, corresponding to the positions $i$ and $j$ of the $-1$ and $1$ respectively, which can range over all values such that $1\leq i\leq j\leq n$. Instead, we will create simpler processes that are upper and lower bounds of this process, which are easier to analyse.
\begin{lem}\label{ohNcubed}
    Let $G$ be a straight-line graph of $n$ nodes and let $f_0$ be a gambler's ruin valuation on $G$. Then the expected time until the shrink update is $\bigO(n^3)$.
\end{lem}
\begin{proof}
    Enumerate the edges of $G$ as $1,2,\cdots,n-1$. Instead of choosing an edge randomly at each timestep, let's fix the entire infinitely-long random sequence in advance. For example:
    \[ (1, 4, 2, 4, 4, 3, 1, 2, \cdots) \]
    Then, we can step through the sequence and update those edges in order. Then for each sequence of this form, the width $t$ is the point in the sequence at which the shrink update occurred. For $n=5$, the first three swaps in the example sequence above bring $-1$ and $1$ next to each other, the next values (4 and 4) accomplish nothing, and the 3 swaps them, meaning $t=6$.
    
    Let's consider a new process where only the $1$ is capable of moving, and the $-1$ is fixed in place. This means the $1$ must be swapped all the way across the graph until it reaches the $-1$. Let $M$ be defined similarly to $t$ -- it is the point in the sequence at which the shrink update occurs in this new process. Then for a fixed infinite sequence of edge updates, we know $t\leq M$. To see why, suppose we have reached point $M$, meaning the $1$ has made its way to the end. Then at some point it must have crossed over the $-1$ in the original process used to define $t$, meaning $t$ has already occurred.

    We now show that $M=\Theta(n^3)$ using the technique which solves the original gambler's ruin problem. Each of the possible positions of $1$ corresponds to a unique state in the Markov chain, and we can transition between the states by swapping $1$ to be closer to or further away from $-1$. Thus, the probabilities of transitioning between the states can be found:
    \begin{center}
        \begin{tikzpicture}[bigg/.style={minimum size=3em}];
            \node[main,bigg] (1) {$e_0$};
            \node[main,bigg] (2) [right = of 1] {$e_1$};
            \node[main,bigg] (3) [right = of 2] {$e_2$};
            \node[bigg] (4) [right = of 3] {$\cdots$};
            \node[main,bigg] (5) [right = of 4] {$e_{n-2}$};
            \node[main,bigg] (6) [right = of 5] {$e_{n-1}$};
        
            \path[->, bend left = 25]
                (2) edge node[below] {$\frac{1}{n-1}$} (1) edge node[above] {$\frac{1}{n-1}$} (3)
                (3) edge node[below] {$\frac{1}{n-1}$} (2) edge node[above] {$\frac{1}{n-1}$} (4)
                (4) edge node[below] {$\frac{1}{n-1}$} (3) edge node[above] {$\frac{1}{n-1}$} (5)
                (5) edge node[below] {$\frac{1}{n-1}$} (4) edge node[above] {$\frac{1}{n-1}$} (6)
                (6) edge node[below] {$\frac{1}{n-1}$} (5);
            \path[->, out=110, in=70, looseness=8]
                (2) edge node[above] {$\frac{n-3}{n-1}$} (2)
                (3) edge node[above] {$\frac{n-3}{n-1}$} (3)
                (5) edge node[above] {$\frac{n-3}{n-1}$} (5)
                (6) edge node[above] {$\frac{n-2}{n-1}$} (6);
        \end{tikzpicture}
    \end{center}
    Note that $e_0$ corresponds to $1$ reaching the end and merging with $-1$, causing a shrink update, and $e_{n-1}$ corresponds to the initial position where $1$ is as far away from $-1$ as possible.
    
    Let $e_k$ be the expected number of moves taken to reach the far left cell from cell $k$. We know $e_0=0$, and we wish to find $e_{n-1}$, which is the expected value of $M$. We will show the following by induction:
    \[ e_k=\frac{k(n-1)}{2}+\frac{k}{k+1}e_{k+1} \]
    This is clearly true for $k=0$, since $e_0=0$. For $1\leq k\leq n-2$, suppose it holds for $k-1$. We will construct a formula for $e_k$; after taking one step, there is a $\frac{1}{n-1}$ chance of moving to $e_{k-1}$, where there are $e_{k-1}$ expected extra steps until the end. Similarly, there is a $\frac{1}{n-1}$ chance of moving to $e_{k+1}$, where there are $e_{k+1}$ expected extra steps until the end. Finally, there is a $\frac{n-3}{n-1}$ chance of staying still, where there are $e_k$ expected extra steps until the end. Therefore:
    \begin{align*}
        e_k &= 1 + \frac{1}{n-1}e_{k-1} + \frac{1}{n-1}e_{k+1} + \frac{n-3}{n-1}e_k \\
        \frac{2}{n-1}e_k &= 1 + \frac{1}{n-1}e_{k-1} + \frac{1}{n-1}e_{k+1}
    \end{align*}
    By the inductive hypothesis:
    \begin{align*}
        \frac{2}{n-1}e_k &= 1 + \frac{1}{n-1}\left(\frac{(k-1)(n-1)}{2}+\frac{k-1}{k}e_k\right) + \frac{1}{n-1}e_{k+1} \\
        2\left(\frac{1}{n-1}e_k\right) &= 1 + \frac{k-1}{2}+\frac{k-1}{k}\left(\frac{1}{n-1}e_k\right) + \frac{1}{n-1}e_{k+1} \\
        \left(2-\frac{k-1}{k}\right)\left(\frac{1}{n-1}e_k\right) &= 1 + \frac{k-1}{2} + \frac{1}{n-1}e_{k+1} \\
        \frac{k+1}{k}\left(\frac{1}{n-1}e_k\right) &= 1 + \frac{k-1}{2} + \frac{1}{n-1}e_{k+1} \\
        \frac{1}{n-1}e_k &= \frac{k}{k+1} + \frac{k(k-1)}{2(k+1)} + \frac{k}{k+1}\cdot\frac{1}{n-1}e_{k+1} \\
        \frac{1}{n-1}e_k &= \frac{2k}{2(k+1)} + \frac{k^2-k}{2(k+1)} + \frac{k}{k+1}\cdot\frac{1}{n-1}e_{k+1} \\
        \frac{1}{n-1}e_k &= \frac{k(k+1)}{2(k+1)} + \frac{k}{k+1}\cdot\frac{1}{n-1}e_{k+1} \\
        \frac{1}{n-1}e_k &= \frac{k}{2} + \frac{k}{k+1}\cdot\frac{1}{n-1}e_{k+1} \\
        e_k &= \frac{k(n-1)}{2} + \frac{k}{k+1}e_{k+1}
    \end{align*}
    Thus, the result holds for all $1\leq k\leq n-2$ by induction.
    
    Now, by substituting $k=n-2$, we get the following relation between $e_{k-1}$ and $e_{k-2}$:
    \[ e_{n-2} = \frac{(n-1)(n-2)}{2} + \frac{n-2}{n-1}e_{n-1} \]
    We can get another relation between $e_{k-1}$ and $e_{k-2}$ by constructing an explicit formula for $e_{n-1}$, similar to above. There is a $\frac{1}{n-1}$ chance of requiring an extra $e_{n-2}$ moves, and a $\frac{n-2}{n-1}$ chance of staying still. Therefore:
    \begin{align*}
        e_{n-1} &= 1 + \frac{1}{n-1}e_{n-2} + \frac{n-2}{n-1}e_{n-1} \\
        \frac{1}{n-1}e_{n-1} &= 1 + \frac{1}{n-1}e_{n-2} \\
        e_{n-1} &= n-1 + e_{n-2}
    \end{align*}
    Using the previous relation:
    \begin{align*}
        e_{n-1} &= n-1 + \frac{(n-1)(n-2)}{2} + \frac{n-2}{n-1}e_{n-1} \\
        \frac{1}{n-1}e_{n-1} &= \frac{2(n-1)}{2} + \frac{(n-1)(n-2)}{2} \\
        &= \frac{n(n-1)}{2} \\
        e_{n-1} &= \frac{n(n-1)^2}{2} \\
        &= \Theta(n^3)
    \end{align*}
    Since $e_{n-1}$ was the expected value of $M$, and $t\leq M$, we know $t=\bigO(n^3)$. But this was the expected width of the gambler's ruin valuation of $n$ nodes, so we are done.
\end{proof}
Since this \emph{feels} like the worst case valuation, this gives a strong idea that the average case for the width is $\bigO(n^3)$ on all graphs. We now provide a proof sketch using a similar technique to show that the worst case for the gambler's ruin graph is $\Omega(n^3)$; combining this with the previous result would yield that the worst case is $\Theta(n^3)$.
\begin{lem}\label{omegaNcubed}
    Let $G$ be a straight-line graph of $n$ nodes and let $f_0$ be a gambler's ruin valuation on $G$. Then the expected time until the shrink update is $\Omega(n^3)$.
\end{lem}
\begin{proofsketch}
    Consider the same process, but let $m$ be the point in the sequence at which either $-1$ or $1$ crosses the center for the first time. Then $m\leq t$, since at the point $t$ we know $-1$ and $1$ have combined, so at least one of them has crossed the center and thus we have passed point $m$.
    \begin{center}
        \begin{tikzpicture}[bigg/.style = {minimum size = 2.5em}]
            \node[main,bigg] (1) {$-1$};
            \node[main,bigg] (2) [right = of 1] {$0$};
            \node[main,bigg] (3) [right = of 2] {$0$};
            \node[main,bigg] (4) [right = of 3] {$0$};
            \node[main,bigg] (5) [right = of 4] {$0$};
            \node[main,bigg] (6) [right = of 5] {$1$};

            \path   (1) edge (2)
                    (2) edge (3)
                    (3) edge node (3p5) {} (4)
                    (4) edge (5)
                    (5) edge (6);

            \def\dist{5em}
            \path   (3p5) -- ++(0,-\dist) edge[redd] +(0,2*\dist) -- node[pos=1] {Center} +(0, -2em);
        \end{tikzpicture}
    \end{center}
    We can now split across the center and view the process as two independent processes:
    \begin{center}
        \begin{tikzpicture}[bigg/.style = {minimum size = 2.5em}]
            \node[main,bigg] (1) {$-1$};
            \node[main,bigg] (2) [right = of 1] {$0$};
            \node[main,bigg] (3) [right = of 2] {$0$};
            \node[bigg] (4) [right = of 3] {$\cdots$};
            \node[main,bigg] (5) [below = of 1] {$1$};
            \node[main,bigg] (6) [right = of 5] {$0$};
            \node[main,bigg] (7) [right = of 6] {$0$};
            \node[bigg] (8) [right = of 7] {$\cdots$};

            \path   (1) edge (2)
                    (2) edge (3)
                    (3) edge node (35) {} (4)
                    (5) edge (6)
                    (6) edge (7)
                    (7) edge node (78) {} (8)
                    
                    (35) -- node (mid) {} (78);

            \def\dist{5em}
            \path
                    (mid) -- ++(0,-\dist) edge[redd] +(0,2*\dist) -- node[pos=1] {Center} +(0, -2em);
        \end{tikzpicture}
    \end{center}
    Now, we are playing two independent games \emph{almost simultaneously}; all edges are still uniformly randomly likely to be picked, but only one game is played at a time, and we stop when either game ends. But the individual games are each just copies of the game from \Cref{ohNcubed}, and since the expected value of each of them was shown to be in $\Theta\left((n/2)^3\right)=\Theta(n^3)$, playing two of them in parallel with these rules and stopping when one of them finishes should different only by a constant, and thus also be $\Theta(n^3)$. Since $m=\Theta(n^3)$ and $m\leq t$, this would imply that $t=\Omega(n^3)$.
\end{proofsketch}
Note that this is not a formal proof; attempts with Markov's inequality and double summations did not appear to help tie everything up at the end. However, in \Cref{sec:expAnalysis1}, we show with experiments that the gambler's ruin valuation is likely to converge in $\Theta(n^3)$ time, verifying the lemma statements above.

Combining \Cref{ohNcubed} with \Cref{omegaNcubed} would show that the expected width of the gambler's ruin valuation is $\Theta(n^3)$. If the gambler's ruin valuation was indeed the worst case for expected width, i.e.\ if the width of any valuation was $\bigO(n^3)$, then combining this with the analysis of the height in \Cref{sec:height} would yield an overall upper bound of $\bigO(n^3q)$ for the expected convergence time of any valuation. But the gambler's ruin valuation has a square sum gap of only $1$; proving the tightness of the $\bigO(n^3q)$ bound for general $q$ would be a challenge, and would require finding a graph with a width that is consistently as bad as the gambler's ruin valuation, and proving this to be the case.

We finish off this section by tying up a few loose ends:
\begin{itemize}
    \item What is the intuition for the gambler's ruin valuation being the worst case? Why would adding \emph{shortcut edges} to it make the expected width lower?
    \item How could we use this intuition to formally show that the expected width of any graph is upper bounded by the gambler's ruin valuation?
\end{itemize}
Then, we close the chapter with experimental analysis of the load balancing model on a few different classes of graph.

To answer the first question, let's simplify the process by supposing that one value is fixed, and we are trying to move the other to it. A random walk on the integer coordinate plane is expected to essentially breadth-first search out from the center, spiralling out and reaching many lattice points close to the center before moving on to lattice points further away. Similarly, we can expect that if we want to maximise the time taken for one value to reach the other, then we want those values to be as far away from each other as possible. Shortcut edges may be traversed forwards, backwards, or both many times, but the most important thing that they do is they shorten the distance between the values. Since the process is random, this means over time, the probability that they do not eventually run into each other becomes more and more unlikely. In other words, we do not want to create opportunities for the start to reach the end by introducing shortcut edges, since the process will just spiral out of both endpoints of the shortcut edge, and thus eventually reach the end faster.

Another way to see this is with some examples. Suppose for simplicity that at each step, the moving value picks a neighbour uniformly at random and moves to it. We can apply the same technique used in the proof of \Cref{ohNcubed} to analyse the following graph where $n=5$, and calculate the expected number of steps until we reach the far right (assuming the walk ends once we do so).
\begin{center}
    \begin{tikzpicture}[bigg/.style = {minimum size = 2.5em}]
        \node[main,bigg] (1) {$e_1$};
        \node[main,bigg] (2) [right = of 1] {$e_2$};
        \node[main,bigg] (3) [right = of 2] {$e_3$};
        \node[main,bigg] (4) [right = of 3] {$e_4$};
        \node[main,bigg] (5) [right = of 4] {$e_5$};

        \def\dist{2em};
        \path   (1) edge (2)    -- node[pos=1, redd] {16} +(0, -\dist)
                (2) edge (3)    -- node[pos=1, redd] {15} +(0, -\dist)
                (3) edge (4)    -- node[pos=1, redd] {12} +(0, -\dist)
                (4) edge (5)    -- node[pos=1, redd] {7} +(0, -\dist)
                (5)             -- node[pos=1, redd] {0} +(0, -\dist);
    \end{tikzpicture}
\end{center}
Note that this means placing the two values at opposite endpoints is optimal. We can then perform the same analysis on the edge with a shortcut edge added, which requires solving a system of three equations.
\begin{center}
    \begin{tikzpicture}[bigg/.style = {minimum size = 2.5em}]
        \node[main,bigg] (1) {$e_1$};
        \node[main,bigg] (2) [right = of 1] {$e_2$};
        \node[main,bigg] (3) [right = of 2] {$e_3$};
        \node[main,bigg] (4) [right = of 3] {$e_4$};
        \node[main,bigg] (5) [right = of 4] {$e_5$};

        \def\dist{2em};
        \path   (1) edge (2) edge[bend left = 30] (4) -- node[pos=1, redd] {12} +(0, -\dist)
                (2) edge (3)    -- node[pos=1, redd] {13} +(0, -\dist)
                (3) edge (4)    -- node[pos=1, redd] {12} +(0, -\dist)
                (4) edge (5)    -- node[pos=1, redd] {9} +(0, -\dist)
                (5)             -- node[pos=1, redd] {0} +(0, -\dist);
    \end{tikzpicture}
\end{center}
Since there is now a possibility of falling back to the beginning, the value of $e_4$ has increased. The value of $e_3$ stayed the same, but the values of $e_1$ and $e_2$ have both decreased, and $e_1$ is no longer the optimal value to place the moving endpoint at. Since all of these values are now below 16, the graph with no shortcut edges was better than the graph with one shortcut edge; it seems likely that the graph with no shortcut edges is in fact optimal.

Now, if we were able to show this formally, we could take any graph and organise the nodes in any order, like we did earlier by drawing the blue path. Then this graph would be like a gambler's ruin game with shortcut edges, and would thus be upper bounded by the worst case analysed earlier, making it $\bigO(n^3)$ by \Cref{ohNcubed}. Turning these sketches into a complete proof that the convergence time is upper bounded by $\bigO(n^3q)$ could be a subject of future research into the load balancing model.

\section{Experimental Analysis}\label{sec:expAnalysis1}
The bounds shown above are theoretical, meaning they are guaranteed to hold in the worst case. However, in practice, on random graphs or on social networks, there may be additional properties which result in different behaviour. For example, \emph{scale free networks}, such as social networks, have a vast majority of nodes with few connections, and a few important nodes (\emph{hubs}) with many connections. Thus, it is worthwhile to analyse such graphs through experimentation.

\subsection{\ER graphs and social networks}
The methodology behind generating random graphs was covered in the background (\Cref{analysedGraphs}), but we are now presented with the challenge of generating random valuations. The width is in terms of $n$, the number of nodes, which is a parameter of $G(n,p)$ and is fixed by the social network datasets. But the height is in terms of $q$, the square sum gap, and it is difficult to generate a `random valuation with square sum gap $q$'. Fortunately however, it is possible.

Let's try to make the sum $0$ for convenience. For a given node, if we give it a random value from $-b$ to $b$ (inclusive), for some $b\geq 0$, then its expected square sum can be calculated using the formula for the sum of the first $b$ squares:
\begin{align*}
    \E[X] &= \frac{1}{2b+1}\sum_{x=-b}^bx^2 \\
    &= \frac{2}{2b+1}\sum_{x=1}^bx^2 \\
    &= \frac{2}{2b+1}\cdot\frac{b(b+1)(2b+1)}{6} \\
    &= \frac{1}{3}b(b+1)
\end{align*}
Thus, if we assign all $n$ nodes a random value from $-b$ to $b$, we can get expected square sum gaps close to $\frac{\E[X]-0}{2}=\frac{1}{6}nb(b+1)$. However, with just one parameter $b$, this is not helpful unless we want to test for some very strange square sum gap values. The tables below test for square sum gaps 10, 100, 1\,000 and 10\,000, which given a fixed $n$, are unlikely to all be writable in the form $\frac{1}{6}nb(b+1)$.

To remedy this, instead of assigning a random value to all $n$ nodes, we assign a random value to $0\leq k\leq n$ of them, giving us an expected square sum gap close to $\frac{1}{6}kb(b+1)$. By choosing $k$ and $b$ such that this value is close to (or equal to) our desired $q$, but $k$ is as large as possible (as we want a large number of nodes to have random values), then for small values of $q$ (e.g.\ values $\leq 10\,000$) it only takes a few attempts until we generate a valuation with square sum gap $q$. For larger $q$ (e.g.\ 100\,000), we simply allow for a $0.01\%$ error bound. That means the results may be the average of valuations with square sum gap 99\,994 or 100\,006, but this is acceptable as the convergence time (particularly the width) has a high variance between trials anyway.
\begin{mytable}
    \centering
    \begin{tabular}{ccccccccccccc}
        \toprule
        \multirow{2}{*}{\shortstack{Desired\vspace{0.3em}\\sqsum gap\vspace{-0.8em}}} & \multicolumn{2}{c}{$n=10$} & \multicolumn{2}{c}{$n=100$} & \multicolumn{2}{c}{$n=1\,000$} & \multicolumn{2}{c}{$n=4\,039$} & \multicolumn{2}{c}{$n=7\,126$} & \multicolumn{2}{c}{$n=10\,000$} \\
        \cmidrule(lr){2-3} \cmidrule(lr){4-5} \cmidrule(lr){6-7} \cmidrule(lr){8-9} \cmidrule(lr){10-11} \cmidrule(lr){12-13}
        & $k$ & $b$ & $k$ & $b$ & $k$ & $b$ & $k$ & $b$ & $k$ & $b$ & $k$ & $b$ \\
        \midrule
        10 & 10 & 2 & 30 & 1 & 30 & 1 & 30 & 1 & 30 & 1 & 30 & 1 \\
        100 & 10 & 7 & 100 & 2 & 300 & 1 & 300 & 1 & 300 & 1 & 300 & 1 \\
        1\,000 & 10 & 24 & 100 & 7 & 1\,000 & 2 & 3\,000 & 1 & 3\,000 & 1 & 3\,000 & 1 \\
        10\,000 & 10 & 77 & 100 & 24 & 833 & 8 & 3\,000 & 4 & 5\,000 & 3 & 10\,000 & 2 \\
        100\,000 & 10 & 245 & 100 & 77 & 1\,000 & 24 & 3\,846 & 12 & 6\,666 & 9 & 8\,333 & 8 \\
        \bottomrule
    \end{tabular}
    \caption{Choices of $k$ and $b$ for various desired square sum gaps}
\end{mytable}

For example, if we have $n=10\,000$ and want a square sum gap of $q=10\,000$, we cannot let $b=1$ since then solving $\frac{1}{6}kb(b+1)=q$ would give us $k=30\,000$, but we need $k\leq n$. However, letting $b=2$ gives us $k=10\,000$, which is the largest possible $k$.

Note that all values presented in the tables are integers as they have been rounded, meaning they may not exactly satisfy $\frac{1}{6}kb(b+1)=q$, but they will be close enough that random generation will produce a correct one in only a couple hundred attempts at worst.

The below experiments were run for 100 trials each. Note that $p=0.5$ for \ER graphs and $m=9$ for \BA graphs.
\begin{mytable}
    \centering
    \begin{tabular}{clccccc}
        \toprule
        \multicolumn{2}{c}{\multirow{2}{*}{Graph}} & \multicolumn{5}{c}{Square sum gap} \\
        \cmidrule{3-7}
        && 10 & 100 & 1\,000 & 10\,000 & 100\,000 \\
        \midrule
        \multirow{4}{*}{\rotatebox[origin=c]{90}{\Erdos}} & 10 & 41.86 & 98.34 & 216.36 & 608.68 & 1872.68 \\
        & 100 & 4\,338.59 & 2\,159.98 & 1\,746.22 & 2\,558.52 & 5\,743.85 \\
        & 1\,000 & 410\,991.19 & 310\,077.53 & 107\,514.51 & 48\,554.61 & 48\,503.97 \\
        & 10\,000 & 39\,871\,750.57 & 31\,101\,818.94 & 14\,740\,931.59 & 5\,206\,490.84 & 1\,767\,497.17 \\
        \multirow{4}{*}{\rotatebox[origin=c]{90}{\;\Barabasi}} & 10 & 42.20 & 97.79 & 235.03 & 678.82 & 2\,063.18 \\
        & 100 & 4\,372.88 & 2\,079.02 & 1\,935.90 & 3\,300.98 & 10\,186.92 \\
        & 1\,000 & 414\,683.06 & 301\,538.37 & 107\,003.94 & 58\,842.71 & 51\,125.03 \\
        & 10\,000 & 40\,124\,384.78 & 31\,291\,793.44 & 16\,397\,386.54 & 7\,232\,802.91 & 2\,790\,517.71 \\
        \multicolumn{2}{c}{Facebook} & 11\,154\,471.65 & 9\,798\,811.93 & 9\,718\,185.96 & 7\,172\,158.85 & 7\,153\,170.34 \\
        \multicolumn{2}{c}{Twitch} & 23\,005\,902.49 & 13\,414\,254.35 & 7\,083\,177.81 & 2\,750\,397.61 & 1\,404\,999.78 \\
        \bottomrule
    \end{tabular}
    \caption{Convergence time of various graphs}
\end{mytable}
As expected, an increase in the number of nodes implies an increase in convergence time. This seemed to consistently be the case both for sparse graphs (e.g.\ Twitch), denser graphs (e.g.\ Facebook, which has half the number of nodes but double the number of edges) and very dense graphs (e.g.\ \ER graphs, for which the expected number of edges is $n(n-1)/4$). It also did not seem to depend on whether the graph was scale-free (e.g.\ \BA graphs and social networks versus \ER graphs); in fact, the results were remarkably consistent among \ER and \BA graphs, with \BA graphs taking longer to converge only for larger values of $q$.

One might also expect an increase in the square sum gap to cause an increase in the convergence time, since the initial valuation is further from the absorbing state. The results show this for the random graphs of 10 nodes, but for 100 nodes, the convergence time actually appears to \emph{decrease} with respect to the square sum up until $q=1\,000$, when it begins to increase again. This suggests the existence of an optimal square sum value for which the expected convergence time is minimised, which is counter-intuitive, and would have been difficult to foresee without experimental analysis.

One possible explanation as to why increasing the square sum gap does not make a large impact on the convergence time is that the vast majority of the time taken occurs after the values are within 2 of each other (i.e.\ taking three adjacent values). Due to the design of the experiment, we almost always end up with a vast quantity of zeroes and a few copies of $1$ and $-1$. The only way to make progress is for these few copies to run into each other and collide, becoming zeroes. For example, for one trial (on a \BA graph with $n=10\,000$ and $q=100\,000$), the process took 1\,149\,218 moves and 36\,060 shrink updates. However, all values became within 2 after only 175\,723 moves, after which 35\,678 shrink updates had been made; the remaining 1.06\% of the shrink updates took the remaining 84.7\% of the time.

This means including some large values does not make a big difference, since almost every time they are swapped with anything, they will decrease until they become close to zero. As the square sum increases, the number of nonzero values in the graph increases also. It is quite possible that this causes these values to destroy each other, getting closer to zero. In contrast, with a low square sum, the graph immediately starts in a position with a few copies of $1$ and $-1$ floating around, but there may be more of these, and thus, the long phase may take longer.

To further investigate this hypothesis, further testing was performed on the \BA graph of 1\,000 nodes, since the optimal square sum gap value appears to be large enough as to have not come up in the original experiments.
\begin{mytable}
    \centering
    \begin{tabular}{ccc}
        \toprule
        \multirow{2}{*}{\shortstack{Desired\vspace{0.3em}\\sqsum gap\vspace{-0.8em}}} & \multicolumn{2}{c}{$n=1\,000$} \\
        \cmidrule(lr){2-3}
        & $k$ & $b$ \\
        \midrule
        1\,000\,000 & 999 & 77 \\
        10\,000\,000 & 1\,000 & 244 \\
        100\,000\,000 & 1\,000 & 774 \\
        \bottomrule
    \end{tabular}
    \caption{Choices of $k$ and $b$ for additional square sum gaps}
\end{mytable}
These experiments were also run for 100 trials.
\begin{mytable}
    \centering
    \begin{tabular}{cccccc}
        \toprule
        \multirow{2}{*}{Graph} & \multicolumn{5}{c}{Square sum gap} \\
        \cmidrule{2-6}
        & 10\,000 & 100\,000 & 1\,000\,000 & 10\,000\,000 & 100\,000\,000 \\
        \midrule
        \Erdos & \multirow{2}{*}{58\,842.71} & \multirow{2}{*}{51\,125.03} & \multirow{2}{*}{119\,931.8} & \multirow{2}{*}{272\,085.98} & \multirow{2}{*}{849\,721.74} \\
        $n=1\,000$ \\
        \bottomrule
    \end{tabular}
    \caption{Convergence time of the \ER graph of 1\,000 nodes}
\end{mytable}
It is worth noting that over half of the updates for the test with $q=100\,000\,000$ were shrink updates. These results reinforce the above conjecture about an optimal value, and suggest that for $n=1\,000$, this optimal value is likely to be between $100\,000$ and $1\,000\,000$ (and almost definitely no smaller than $10\,000$).

The code also recorded the number of shrink updates required until convergence, i.e.\ the height. Within the 100 trials, this varied much less than the convergence time, mostly due to the aforementioned fault in the methodology; by setting $k$ as high as possible, we set $b$ as low as possible, making it more likely that the graph either starts off with, or eventually ends up in, the position where all values are within 2 of each other.
\begin{mytable}
    \centering
    \begin{tabular}{clccccc}
        \toprule
        \multicolumn{2}{c}{\multirow{2}{*}{Graph}} & \multicolumn{5}{c}{Square sum gap} \\
        \cmidrule{3-7}
        && 10 & 100 & 1\,000 & 10\,000 & 100\,000 \\
        \midrule
        \multirow{4}{*}{\rotatebox[origin=c]{90}{\Erdos}} & 10 & 7.52 & 32.19 & 112.96 & 373.61 & 1\,218.56 \\
        & 100 & 10.00 & 78.25 & 329.23 & 1\,171.99 & 3\,814.90 \\
        & 1\,000 & 10.00 & 100.00 & 794.47 & 3\,326.18 & 11\,674.16 \\
        & 10\,000 & 10.00 & 100.00 & 1\,000.00 & 7\,978.60 & 33\,515.83 \\
        \multirow{4}{*}{\rotatebox[origin=c]{90}{\;\Barabasi}} & 10 & 7.77 & 34.66 & 122.96 & 411.86 & 1\,312.57 \\
        & 100 & 10.00 & 78.86 & 352.03 & 1\,265.33 & 4\,226.42 \\
        & 1\,000 & 10.00 & 100.00 & 812.24 & 3\,561.42 & 13\,182.14 \\
        & 10\,000 & 10.00 & 100.00 & 1\,000.00 & 8\,172.87 & 35\,989.43 \\
        \multicolumn{2}{c}{Facebook} & 10.00 & 100.00 & 1\,000.00 & 6\,102.37 & 25\,473.41 \\
        \multicolumn{2}{c}{Twitch} & 10.00 & 100.00 & 1\,000.00 & 6\,789.42 & 35\,041.03 \\
        \bottomrule
    \end{tabular}
    \caption{Height of various graphs}
\end{mytable}
Note that when $q\leq n/10$, we set $b=1$ and thus the range of values is from $-1$ to $1$ inclusive, meaning all shrink updates are optimal and we hit the worst case of $q$; this is why many of the table values are exact.

While the theoretical worst case bound for the number of shrink updates depended only on $q$, the average case in practice appears to depend on both $n$ and $q$. This is likely because in smaller graphs, the further apart values are more likely to run into each other and skip many shrink updates in the worst-case sequence. With a small number of nodes, there is not much room to have the values decrease by two each time, whereas on large graphs, they can swap past each other to locate the values close to them, and thus not decrease the square sum by as much.

In terms of the relation to $n$, it appears that for a fixed $q$, the number of shrink updates converges towards $q$, getting marginal returns as $n$ increases to $10\,000$. It is possible that this is simply caused by the fact that $b$ decreases as $n$ increases, meaning that the range becomes tighter and the possible shrink update window decreases. The relation between the expected number of shrink updates and the worst case (i.e.\ $q$) as $n$ increases could be a topic of interest for future research.

\subsection{Binomial valuations}
We can also consider the special valuations that arose earlier from analysing the height and width. For the height, we considered the $n$-valuation (\Cref{defn:binVals}), which has a square sum gap of $q=n2^{n-1}$, and showed it had a worst case of $q$ shrink updates. The below experiments were run for 1\,000 trials.
\begin{mytable}
    \centering
    \begin{tabular}{ccccccc}
        \toprule
        $n$-val.\ & Nodes & Sqsum Gap & Conv Time & Height & Width & Worst case (\%) \\
        \midrule
        2 & 4 & 4 & 9.51 & 3.09 & 3.08 & 77.25 \\
        3 & 8 & 12 & 48.40 & 8.18 & 5.92 & 68.17 \\
        4 & 16 & 32 & 205.42 & 19.95 & 10.30 & 62.34 \\
        5 & 32 & 80 & 854.65 & 46.46 & 18.40 & 58.08 \\
        6 & 64 & 192 & 3\,421.41 & 104.54 & 32.73 & 54.45 \\
        7 & 128 & 448 & 13\,745.82 & 230.54 & 59.62 & 51.46 \\
        8 & 256 & 1\,024 & 52\,307.78 & 499.30 & 104.76 & 48.76 \\
        9 & 512 & 2\,304 & 214\,564.3 & 1\,072.57 & 200.06 & 46.55 \\
        \bottomrule
    \end{tabular}
    \caption{Convergence time breakdown for binomial $n$-valuations}
\end{mytable}
The `worst case \%' column shows the height obtained from the experiments as a percentage of the worst case height, which is the square sum gap. This decreases over time, since the larger the graph, the more values make up the complete graph, and thus the more likely we are to skip steps in the shrink update sequence. For example, at $n=3$, the value $-3$ can be swapped with either $-1$, $1$ or $3$, and for the worst case we want it to be swapped with $-1$. However, at $n=9$, the value $-9$ can be swapped with anything from $-7, -5, \cdots, 7, 9$, which is a much larger range of values.

The width was calculated as the convergence time divided by the height. We can plot this with respect to $n$, since we expect the width to depend only on the number of nodes.
\begin{myfigure}
    \centering
    \includegraphics[width=0.8\linewidth]{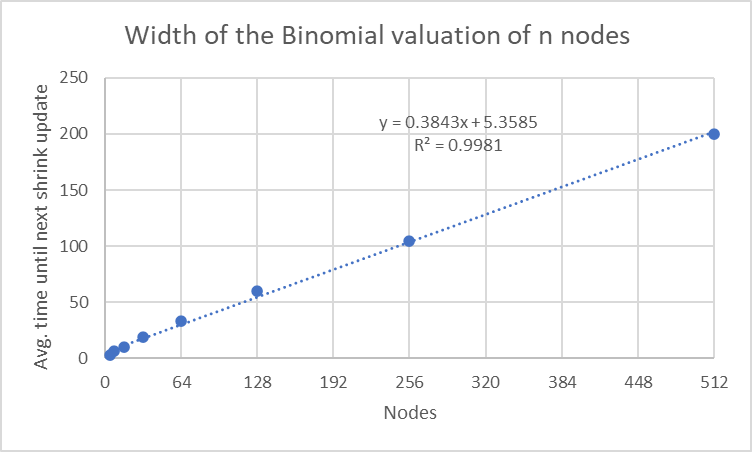}
\end{myfigure}
Then, we can plot the height with respect to $q$, since we expect the height to depend only on the square sum gap.
\begin{myfigure}
    \centering
    \includegraphics[width=0.8\linewidth]{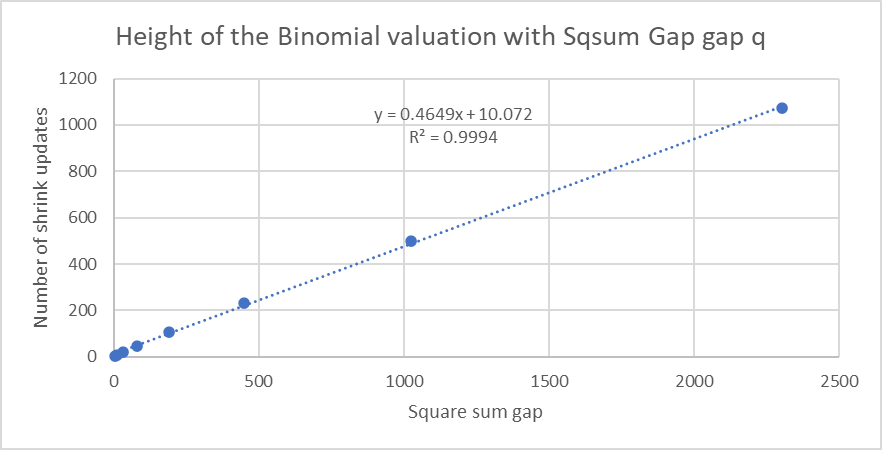}
\end{myfigure}
The high R squared values support the hypothesis that the width depends on $n$ and the height depends on $q$, making the convergence time, i.e.\ the product of the width and height, proportional to $nq$. Since the binomial $k$-valuation has $n=2^k$ nodes and a square sum gap of $q=k2^{k-1}$, this means $q$ is proportional to $n\log n$, and thus the convergence time appears to be $\Theta(n^2\log n)$.

Note that the complete graph is dense with $\Theta(n^2)$ edges, which made the random graphs more likely to bounce around without performing a shrink update. However, for binomial valuations, this appears to proportional to $n$, rather than $\Theta(n^3)$ as shown by the theoretical bound. A potential research question could be to see if there are any dense graphs that also obtain a $\Theta(n^3)$ worst case bound like the Gambler's ruin $n$-valuations, and to find a link between the density of a graph and the convergence time.

\subsection{Gambler's ruin valuations}\label{sec:grVals}
For the width, we considered the gambler's ruin valuation (\Cref{defn:grVals}), and showed the expected convergence time is $\Theta(n^3)$. The below experiments were run for 1\,000 trials.
\begin{mytable}
    \centering
    \begin{tabular}{cc}
        \toprule
        Nodes ($n$) & Conv Time \\
        \midrule
        5 & 19.61 \\
        10 & 156.58 \\
        15 & 529.50 \\
        20 & 1\,230.18 \\
        25 & 2\,377.34 \\
        30 & 4\,110.76 \\
        35 & 6\,472.44 \\
        40 & 9\,428.93 \\
        45 & 13\,752.96 \\
        50 & 19\,517.74 \\
        55 & 24\,766.43 \\
        60 & 31\,817.11 \\
        \bottomrule
    \end{tabular}
    \caption{Convergence Time of Gambler's Ruin $n$-valuations}
\end{mytable}
Plotting these points, a polynomial trendline of degree 2 starts off decreasing, and does not grow fast enough on the right-hand side. A trendline of degree 3, however, fits the points with $R^2=0.9994$. For degree 4 (shown below), the coefficient of $x^4$ is $-0.0015$; the low magnitude indicates that this term is unimportant, and the sign indicates a further mistake, as it will cause the graph to slope downwards in the long term.
\begin{center}
    \includegraphics[width=0.8\linewidth]{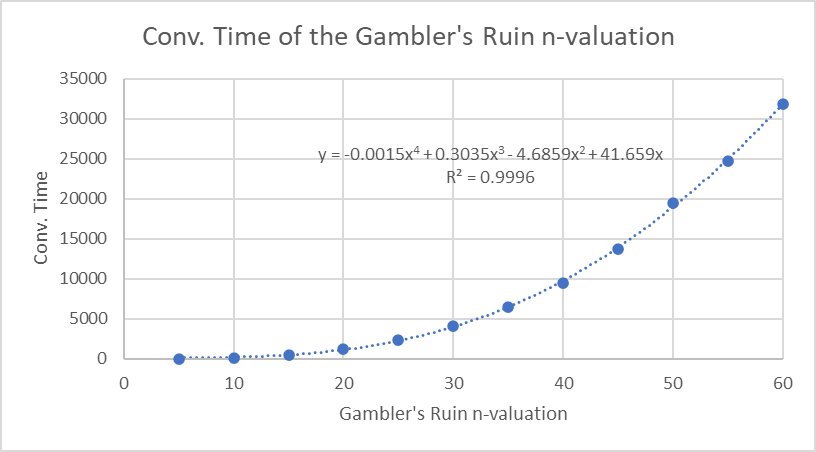}    
\end{center}
Thus, the experiment backs up the theoretical proof sketch that the width is $\Theta(n^3)$ in this particular example. It remains to show that there are no other examples with higher convergence times, but this appears unlikely.                          
\chapter{Maximum Model}\label{chap:max-model}

In the synchronous maximum model, each node takes the maximum of its neighbours in each iteration, overriding its own value. We first consider the undirected case, presenting a view of the process that makes the convergence behaviour and convergence time clear (\Cref{undPeriod1or2}). Then, we proceed to the directed case where the behaviour is much more complex. We start with strongly connected graphs where we can place a restriction on the period (\Cref{pDividesG}) and provide descriptions of the absorbing state (\Cref{maximumSpreadsStrong}, \Cref{unintuitiveMaximumSpreadsStrong}). We can then show that all possible values $p$ for the period that obey the restriction can arise from a possible graph (\Cref{constructivePeriodP}), showing that the restriction is tight. Finally, in the weakly connected case, we present an example where the period is exponential (\Cref{ahadExample}).

\section{Undirected Graphs}
Similar to \Cref{chap:load-balancing}, we begin with a motivating example to see how the process plays out.
\begin{center}
    \begin{tikzpicture}
        \node[main] (1) {7};
        \node[main] (2) [right = of 1] {3};
        \node[main] (3) [right = of 2] {1};
        \node[main] (4) [below = of 1] {5};
        \node[main] (5) [right = of 4] {4};
        \node[main] (6) [right = of 5] {6};
        \path   (1) -- node (14) {} (4)
                (3) -- node (36) {} (6);
        \node [left = of 14] {$t=0$};
        \node[main] (7) [right = of 36] {2};

        \path   (1) edge (2) edge (4)
                (2) edge (3) edge (5)
                (3) edge (6) edge (7)
                (4) edge (5)
                (5) edge (6)
                (6) edge (7);
    \end{tikzpicture}
\end{center}
On the first iteration, the $6$ node sets itself to $4$, the maximum of its three neighbours $4$, $1$ and $2$. All other nodes undergo the same transformation. The first few iterations proceed as follows:
\begin{center}
    \begin{tikzpicture}
        \node[main] (1) {5};
        \node[main] (2) [right = of 1] {7};
        \node[main] (3) [right = of 2] {6};
        \node[main] (4) [below = of 1] {7};
        \node[main] (5) [right = of 4] {6};
        \node[main] (6) [right = of 5] {4};
        \path   (1) -- node (14) {} (4)
                (3) -- node (36) {} (6);
        \node [left = of 14] {$t=1$};
        \node[main] (7) [right = of 36] {6};

        \path   (1) edge (2) edge (4)
                (2) edge (3) edge (5)
                (3) edge (6) edge (7)
                (4) edge (5)
                (5) edge (6)
                (6) edge (7);
    \end{tikzpicture} \\
    \begin{tikzpicture}
        \node[main] (1) {7};
        \node[main] (2) [right = of 1] {6};
        \node[main] (3) [right = of 2] {7};
        \node[main] (4) [below = of 1] {6};
        \node[main] (5) [right = of 4] {7};
        \node[main] (6) [right = of 5] {6};
        \path   (1) -- node (14) {} (4)
                (3) -- node (36) {} (6);
        \node [left = of 14] {$t=2$};
        \node[main] (7) [right = of 36] {6};

        \path   (1) edge (2) edge (4) -- node[pos=1] {} +(0, \picgap)
                (2) edge (3) edge (5)
                (3) edge (6) edge (7)
                (4) edge (5)
                (5) edge (6)
                (6) edge (7);
    \end{tikzpicture} \\
    \begin{tikzpicture}
        \node[main] (1) {6};
        \node[main] (2) [right = of 1] {7};
        \node[main] (3) [right = of 2] {6};
        \node[main] (4) [below = of 1] {7};
        \node[main] (5) [right = of 4] {6};
        \node[main] (6) [right = of 5] {7};
        \path   (1) -- node (14) {} (4)
                (3) -- node (36) {} (6);
        \node [left = of 14] {$t=3$};
        \node[main] (7) [right = of 36] {7};

        \path   (1) edge (2) edge (4) -- node[pos=1] {} +(0, \picgap)
                (2) edge (3) edge (5)
                (3) edge (6) edge (7)
                (4) edge (5)
                (5) edge (6)
                (6) edge (7);
    \end{tikzpicture}
\end{center}
Observe that the 7 and 6 nodes have each spread to every second node, creating a checkerboard pattern. Initially, it looks like the pattern will oscillate, but observe:
\begin{center}
    \begin{tikzpicture}
        \node[main] (1) {7};
        \node[main] (2) [right = of 1] {6};
        \node[main] (3) [right = of 2] {7};
        \node[main] (4) [below = of 1] {6};
        \node[main] (5) [right = of 4] {7};
        \node[main] (6) [right = of 5] {7};
        \path   (1) -- node (14) {} (4)
                (3) -- node (36) {} (6);
        \node [left = of 14] {$t=4$};
        \node[main] (7) [right = of 36] {7};

        \path   (1) edge (2) edge (4)
                (2) edge (3) edge (5)
                (3) edge (6) edge (7)
                (4) edge (5)
                (5) edge (6)
                (6) edge (7);
    \end{tikzpicture} \\
    \begin{tikzpicture}
        \node[main] (1) {6};
        \node[main] (2) [right = of 1] {7};
        \node[main] (3) [right = of 2] {7};
        \node[main] (4) [below = of 1] {7};
        \node[main] (5) [right = of 4] {7};
        \node[main] (6) [right = of 5] {7};
        \path   (1) -- node (14) {} (4)
                (3) -- node (36) {} (6);
        \node [left = of 14] {$t=5$};
        \node[main] (7) [right = of 36] {7};

        \path   (1) edge (2) edge (4) -- node[pos=1] {} +(0, \picgap)
                (2) edge (3) edge (5)
                (3) edge (6) edge (7)
                (4) edge (5)
                (5) edge (6)
                (6) edge (7);
    \end{tikzpicture} \\
    \begin{tikzpicture}
        \node[main] (1) {7};
        \node[main] (2) [right = of 1] {7};
        \node[main] (3) [right = of 2] {7};
        \node[main] (4) [below = of 1] {7};
        \node[main] (5) [right = of 4] {7};
        \node[main] (6) [right = of 5] {7};
        \path   (1) -- node (14) {} (4)
                (3) -- node (36) {} (6);
        \node [left = of 14] {$t=6$};
        \node[main] (7) [right = of 36] {7};

        \path   (1) edge (2) edge (4) -- node[pos=1] {} +(0, \picgap)
                (2) edge (3) edge (5)
                (3) edge (6) edge (7)
                (4) edge (5)
                (5) edge (6)
                (6) edge (7);
    \end{tikzpicture}
\end{center}
Eventually, the maximum 7 node managed to leverage the 3-cycle on the right of the graph to spread to every node. If there was no such 3-cycle, it seems like the checkerboard pattern would have oscillated forever, which prompts us to consider bipartite graphs separately. We can also observe that the $7$ spread from the left to the right through every second node, then spread back to the left again through the nodes it missed. This is reminiscent of a breadth-first search (BFS), so let's consider what happens when we run a BFS algorithm from a maximum node and rebuild the graph accordingly.
\begin{center}
    \begin{tikzpicture}
        \node[main] (1) {7};
        \node[main] (2) [right = of 1] {3};
        \node[main] (3) [right = of 2] {1};
        \node[main] (4) [below = of 1] {5};
        \node[main] (5) [right = of 4] {4};
        \node[main] (6) [right = of 5] {6};
        \path   (1) -- node (14) {} (4)
                (3) -- node (36) {} (6);
        \node [left = of 14] {$t=0$};
        \node[main] (7) [right = of 36] {2};

        \path   (1) edge (2) edge (4)
                (2) edge (3) edge (5)
                (3) edge (6) edge (7)
                (4) edge (5)
                (5) edge (6)
                (6) edge (7);
    \end{tikzpicture}
\end{center}
Let $u$ be a maximum node. In this case $f_0(u)=7$.
\begin{center}
    \begin{tikzpicture}
        \node[main] (7) {7};
        \def\off{2.2em};
        \def\labeloff{3em};
        \path   (7) -- node (above7) {} +(-\off, \off)
                (7) -- node (below7) {} +(\off, -\off)
                (7) -- node[pos=1, bfslabel] {\huge\textbf{0}} +(0, \labeloff);
        \draw[bfs, rounded corners] (above7) rectangle (below7);
    \end{tikzpicture}
\end{center}
On the first step, we add the nodes in $\Gamma^1(u)$ that have not yet been added:
\begin{center}
    \begin{tikzpicture}
        \node[main] (7) {7};
        \node[main] (3) [above right = of 7] {3};
        \node[main] (5) [below right = of 7] {5};

        \def\off{2.2em};
        \def\labeloff{3em};
        \path   (7) edge (3) edge (5)
                    -- node (above7) {} +(-\off, \off)
                (7) -- node (below7) {} +(\off, -\off)
                (7) -- node[pos=1, bfslabel] {\huge\textbf{0}} +(0, \labeloff)
                (3) -- node (off3) {} +(-\off, \off)
                (3) -- node[pos=1, bfslabel] {\huge\textbf{1}} +(0, \labeloff)
                (5) -- node (off5) {} +(\off, -\off);
        \draw[bfs, rounded corners] (above7) rectangle (below7)
                                    (off3) rectangle (off5);
    \end{tikzpicture}
\end{center}
Then, we add the nodes in $\Gamma^2(u)$ that have not yet been added:
\begin{center}
    \begin{tikzpicture}
        \node[main] (7) {7};
        \node[main] (3) [above right = of 7] {3};
        \node[main] (5) [below right = of 7] {5};
        \node[main] (1) [right = of 3] {1};
        \node[main] (4) [right = of 5] {4};

        \def\off{2.2em};
        \def\labeloff{3em};
        \path   (7) edge (3) edge (5)
                    -- node (above7) {} +(-\off, \off)
                (7) -- node (below7) {} +(\off, -\off)
                (7) -- node[pos=1, bfslabel] {\huge\textbf{0}} +(0, \labeloff)
                (3) edge (1) edge (4)
                    -- node (off3) {} +(-\off, \off)
                (3) -- node[pos=1, bfslabel] {\huge\textbf{1}} +(0, \labeloff)
                (5) edge (4)
                    -- node (off5) {} +(\off, -\off)
                (1) -- node (off1) {} +(-\off, \off)
                (1) -- node[pos=1, bfslabel] {\huge\textbf{2}} +(0, \labeloff)
                (4) -- node (off4) {} +(\off, -\off);
        \draw[bfs, rounded corners] (above7) rectangle (below7)
                                    (off3) rectangle (off5)
                                    (off1) rectangle (off4);
    \end{tikzpicture}
\end{center}
Finally, we add the nodes in $\Gamma^3(u)$ that have not yet been added, which adds all remaining nodes:
\begin{center}
    \begin{tikzpicture}
        \node[main] (7) {7};
        \node[main] (3) [above right = of 7] {3};
        \node[main] (5) [below right = of 7] {5};
        \node[main] (1) [right = of 3] {1};
        \node[main] (4) [right = of 5] {4};
        \node[main] (2) [right = of 1] {2};
        \node[main] (6) [right = of 4] {6};

        \def\off{2.2em};
        \def\labeloff{3em};
        \path   (7) edge (3) edge (5)
                    -- node (above7) {} +(-\off, \off)
                (7) -- node (below7) {} +(\off, -\off)
                (7) -- node[pos=1, bfslabel] {\huge\textbf{0}} +(0, \labeloff)
                (3) edge (1) edge (4)
                    -- node (off3) {} +(-\off, \off)
                (3) -- node[pos=1, bfslabel] {\huge\textbf{1}} +(0, \labeloff)
                (5) edge (4)
                    -- node (off5) {} +(\off, -\off)
                (1) edge (2) edge (6)
                    -- node (off1) {} +(-\off, \off)
                (1) -- node[pos=1, bfslabel] {\huge\textbf{2}} +(0, \labeloff)
                (4) edge (6)
                    -- node (off4) {} +(\off, -\off)
                (2) edge (6)
                    -- node (off2) {} +(-\off, \off)
                (2) -- node[pos=1, bfslabel] {\huge\textbf{3}} +(0, \labeloff)
                (6) -- node (off6) {} +(\off, -\off);
        \draw[bfs, rounded corners] (above7) rectangle (below7)
                                    (off3) rectangle (off5)
                                    (off1) rectangle (off4)
                                    (off2) rectangle (off6);
    \end{tikzpicture}
\end{center}
Now, at $t=0$ the entire zeroth component has the maximum value. At $t=1$ it spreads to the first component, but leaves the zeroth component. At $t=2$ it spreads to the second component and leaves the first component, but also returns to the zeroth component again. In general, at time $t$ all components $t'\leq t$ with the same parity as $t$ (modulo 2) will be guaranteed to have the maximum. But on top of this, if a component has an edge to itself, such as the third component above, then the maximum will not leave that component, meaning it can propagate out to the rest of the graph.

This gives us a complete understanding of the convergence process; for bipartite graphs, the nodes an even distance away from the maximum will have the maximum at the opposite times to the nodes an odd distance away, and for non-bipartite graphs, we first spread the maximum to the component with an edge to itself, then the maximum propagates out from that component. For the rest of the analysis, we will assume that $G$ is a connected graph, since otherwise we could simply consider each connected component separately. Since they run in parallel, the answer is the maximum of the times for each connected component.
\begin{defn}[The $k$th component $\Lambda_k$]
    Let $G=(V,E)$ be a connected undirected graph. Let $s$ be the source, i.e.\ any node such that $f_0(s)\geq f_0(v)$ for all $v\in V$. Then let $\Lambda_k(s)$ be the set of vertices reached on the $k$th step of the BFS from $s$. That is, for $k\geq 0$:
    \[ \Lambda_k(s) \coloneqq \Gamma^k(s) \setminus \bigcup_{i=0}^{k-1}\Gamma^i(s) \]
\end{defn}
Note that $\Lambda_0(s)=\Gamma^0(s) = \{s\}$. In other words, $\Lambda_k(s)$ is the set of vertices a distance of $k$ away from $s$ ($\Gamma^k(s)$) excluding the ones seen already ($\bigcup_{i=0}^{k-1}\Gamma^i(s)$). Note that this partitions the graph into distinct components.
\begin{prop}
    For all $v\in V$ there is a unique $k$ such that $v\in\Lambda_k(s)$. \qed
\end{prop}
We first show that all edges in this BFS graph are within a component, or between two adjacent components.
\begin{lem}\label{edgesWithin1}
    Let $(u,v)\in E$ be an edge, and let $u\in\Lambda_k(s)$ and $v\in\Lambda_\ell(s)$. Then $|k-\ell|\leq 1$.
\end{lem}
\begin{proof}
    Suppose $\ell\geq k+2$. Then there is a path to $v$ of length $k+1$ through $u$, contradicting that $v$ is in a component $\ell\geq k+2$. Formally, since $u\in\Lambda_k(s)$ we know $u\in\Gamma^k(s)$, so since $(u,v)\in E$, we know $v\in\Gamma^{k+1}(s)$. Then since $v\in\Lambda_\ell(s)$ we know $v\notin\bigcup_{i=0}^{\ell-1}\Gamma^i(s)$. But $k+1\leq\ell-1$, so this contradicts that $v\in\Gamma^{k+1}(s)$. A symmetric argument works for when $k\geq\ell+2$.
\end{proof}
Now we can show the connection between bipartite graphs and our BFS graph.
\begin{lem}\label{selfEdges}
    Let $G=(V,E)$ be an undirected graph. Then $G$ is bipartite if and only if no edges are between two vertices in the same component $\Lambda_k$, including self-edges.
\end{lem}
\begin{proof}
    It is well known that a graph is bipartite if and only if it has no odd-length cycles, so we will show that a graph has an odd-length cycle if and only if there is an edge between two vertices in the same component.

    $(\impliedby)$ \\
    Suppose there is an edge $(u,v)\in E$ such that $u,v\in\Lambda_k(s)$. Then there are paths $s,w_1,\cdots,w_{k-1},u$ and $s,x_1,\cdots,x_{k-1},v$ from $s$ to each of $u$ and $v$, each of which is length $k$. Then the following is an odd cycle of length $2k+1$:
    \[ s,w_1,\cdots,w_{k-1},u,v,x_{k-1},\cdots,x_1,s \]
    Note that this cycle may visit the same nodes multiple times, and that the proof is valid even if $u=v$.

    $(\implies)$ \\
    Now, suppose there is a cycle $w_1,\cdots,w_{2k+1},w_1$ of odd length. Consider the parity mod 2 of the component that we are in. Since the cycle has odd length, this parity cannot alternate every step, so there is a step in which the parity remains the same. If $k$ and $\ell$ are the components at this step, then $k-\ell$ is even and $|k-\ell|\leq 1$ by \Cref{edgesWithin1}, meaning $k=\ell$.
\end{proof}
With this knowledge, for bipartite graphs we can call the two components $B$ and $B'$ and denote them as follows:
\begin{defn}[Bipartite components]\label{bipartDefn}
    Let $G=(V,E)$ be a bipartite undirected graph. Then the two bipartite components $B$ and $B'$ are:
    \begin{align*}
        B &= \bigcup_{k=0}^\infty\Lambda_{2k}(s) \\
        B' &= \bigcup_{k=0}^\infty\Lambda_{2k+1}(s)
    \end{align*}
\end{defn}
Now that we understand the internal structure of our BFS graph, we can ask ourselves, how many components are there? We can bound this by $n$, the number of nodes in the graph, since at least one node is introduced per component. However, certain classes of graphs, such as the `star graph' below, produce BFS graphs with a small number of components.
\begin{center}
    \begin{tikzpicture}[bigg/.style={minimum size = 2em}]
        \node[main,bigg] (0) {0};
        \def\n{10};
        \def\lessn{9}; 
        \foreach \angle in {1, ..., \n} {
            \node[main,bigg] (\angle) at (90-\angle*360/\n:6em) {\angle};
            \path (0) edge (\angle);
        }

        \def\gap{12em};
        \def\off{2.2em};
        \def\labeloff{3em};
        \node[main,bigg] (10) at (\gap, 0) {10};
        \node[main,bigg] (0) at (\gap+6em, 0) {0};
        \path   (10) edge (0) -- node (above10) {} +(-\off, \off)
                (10) -- node (below10) {} +(\off, -\off)
                (10) -- node[pos=1, bfslabel] {\huge\textbf{0}} +(0, \labeloff)
                (0) -- node (above0) {} +(-\off, \off)
                (0) -- node (below0) {} +(\off, -\off)
                (0) -- node[pos=1, bfslabel] {\huge\textbf{1}} +(0, \labeloff);
        \foreach \angle in {1, 3, ..., \n} {
            \node[main,bigg] (\angle) at (\gap+12em, {2em*(5-\angle)}) {\angle};
            \path (0) edge (\angle);
        }
        \foreach \angle in {2, 4, ..., \lessn} {
            \node[main,bigg] (\angle) at (\gap+14em, {2em*(5-\angle)}) {\angle};
        }
        \path   (0) edge[out=43, in=180] (2)
                    edge[out=16, in=180] (4)
                    edge[out=-16, in=180] (6)
                    edge[out=-43, in=180] (8)
                (1) -- node (above1) {} +(-\off, \off)
                (1) -- node (mid12) {} (2)
                (mid12) -- node[pos=1, bfslabel] {\huge\textbf{2}} +(0, \labeloff+1em)
                (9) -- node (below9) {} +(\off + 3.7em, -\off);

        \draw[bfs, rounded corners] (above10) rectangle (below10)
                                    (above0) rectangle (below0)
                                    (above1) rectangle (below9);
    \end{tikzpicture}
\end{center}
Since each node is included in the component corresponding to its distance from the source, the maximum component value is the maximum distance of any node to the source. If we do not know what the source is in advance, we can bound this by the maximum distance between any two nodes, which is the diameter from \Cref{defn:diameter}. For example, the diameter of the star graph is $2$ (by taking a path from one leaf, through the center and to another leaf), regardless of the number $n$ of nodes.
\begin{lem}
    The number of components, i.e.\ the smallest $k$ such that $\Lambda_k(s)=\emptyset$, is bounded by the diameter $D$.
\end{lem}
\begin{proof}
    If $u\in\Lambda_k(s)$ then the shortest distance from $s$ to $u$ is $k$, and the largest of these shortest distances is less than or equal to the diameter $D$.
\end{proof}
With this visualisation in place, and the formalisation in terms of $\Lambda_k$, as well as the previous propositions and lemmas, we are now ready to prove the following key theorem.
\begin{thm}\label{undPeriod1or2}
    Let $G=(V,E)$ be a connected undirected graph. Then if $G$ is not bipartite then the period is $1$. Otherwise, let $B$ and $B'$ be the bipartite components as defined in \Cref{bipartDefn}, and let $m=\max_{v\in B}f_0(v)$ and $m'=\max_{v\in B'}f_0(v)$. Then if $m=m'$ the period is $1$, otherwise the period is $2$. Additionally, the convergence time is at most $2D$, where $D$ is defined in \Cref{defn:diameter}.
\end{thm}
\begin{proof}
    Suppose $G$ is not bipartite, so there is an edge $(u,v)\in E$ and component $k\geq 0$ such that $u,v\in\Lambda_k(s)$. We can show that at time $t$ the maximum spreads to $\Lambda_\ell(s)$ for all $0\leq \ell\leq t$ such that $\ell$ and $t$ have same parity modulo 2 (i.e.\ $\ell-t$ is even). Note that if there are multiple maxima then other components (including the ones with opposite parity) may have the maximum value too.

    The above statement is clearly true for $t=0$ as the maximum value has spread to all vertices in $\Lambda_0(s)=\{s\}$. Consider $t\geq 1$ and let $0\leq\ell\leq t$ have the same parity as $t$. If $\Lambda_\ell(s)$ is empty then the maximum has trivially spread to all vertices in it; otherwise, let $w\in\Lambda_\ell(s)$ be arbitrary -- we will show that $w$ has acquired the maximum value at time $t$. If $\ell=0$ then $w=s$ and $t\geq 2$ is even. By \Cref{assNeighbour}, $s$ has at least one neighbour $x\in\Lambda_1(s)$. Then $1$ satisfies $0\leq 1\leq t-1$ and has the same parity as $t-1$, so the maximum value spreads to $s$. In the other case, $\ell\geq 1$, so $\ell-1$ satisfies $0\leq \ell-1\leq t-1$ and has the same parity as $t-1$. In either case, the maximum value spreads to $w$, and so our statement is true by induction.
    
    Thus, at $t=k$ steps, the maximum has spread to $\Lambda_k(s)$, so $f_k(u)$ and $f_k(v)$ are maximum. Since $k\leq D$, this means $f_D(u)$ and $f_D(v)$ are maximum. But then $u$ and $v$ can pull from each other in all future updates, so $f_t(u)$ and $f_t(v)$ are maximum for all $t\geq D$. With a similar inductive argument as above, but running BFS from $u$ instead of $s$ and without the condition about parity, the maximum spreads to all components $\Lambda_\ell(u)$. Since there are at most $D$ of these, then all components are maximum by time $t=2D$. Clearly the period ends up as $1$ in this case.

    Now, suppose $G$ is bipartite. Let $m$ and $m'$ be as defined in \Cref{bipartDefn}, and let $u$ and $u'$ be vertices that achieve those respective values in the corresponding bipartite components. By running a similar inductive argument to above, we know after $D$ steps that $m$ has spread to all components $\Lambda_0(u),\cdots,\Lambda_D(u)$ and $m'$ has spread to all components $\Lambda_0(u'),\cdots,\Lambda_D(u')$. If $m=m'$ then all vertices are maximum, so the period is $1$. Otherwise, at any time $t\geq D$, either all vertices in $B$ have value $m$ and all vertices in $B'$ have value $m'$, or vice versa, all vertices in $B$ have value $m'$ and all vertices in $B'$ have value $m$. Either way, since there are no edges from $B$ to itself or from $B'$ to itself, the maxima $m$ and $m'$ spread to their corresponding bipartite components, and thus we reach an absorbing state with a period of $2$.
\end{proof}
Considering connected components separately, we also get the following:
\begin{cor}
    Let $G=(V,E)$ be any undirected graph. Then the period is 2 if $G$ has a connected component which is
    \begin{itemize}
        \item bipartite with bipartite components $B$ and $B'$ as defined in \Cref{bipartDefn}, and
        \item $m\neq m'$, where $m=\max_{v\in B}f_0(v)$ and $m'=\max_{v\in B'}f_0(v)$.
    \end{itemize}
    If no connected component satisfies both of these criteria, then the period is 1. Additionally, the convergence time is at most $2D$, where $D$ is defined in \Cref{defn:diameter}.
\end{cor}
\begin{proof}
    We can apply \Cref{undPeriod1or2} to each connected component to get its convergence behaviour. The overall convergence time is the maximum of the convergence times for each component. After each component has converged, those with period 1 have their valuations remain fixed in each iteration, and the components with period 2 alternate between two valuations. Even if this alternation is out-of-sync with other components, the overall period will still be 2. Thus, the overall period is 2 if and only if there exists a component with period 2.
\end{proof}
Note that this would not be the case for values larger than 2. For example, if a component alternated every 6 iterations, and another alternated every 10 iterations, then the overall period would be the least common multiple (LCM) of these, which is 30.

If we were to go into a bit more detail, it would be possible to lower the upper bound to $D-1$ for connected bipartite graphs and $2D-1$ for connected non-bipartite graphs. Below are examples of graph families which hit these bounds:
\begin{center}
    \begin{tikzpicture}[bigg/.style = {minimum size = 3.2em}]
        \node[main,bigg] (n) {$n$};
        \node[main,bigg] (n-1) [right = of n] {$n-1$};
        \node (dots) [right = of n-1] {$\cdots$};
        \node[main,bigg] (2) [right = of dots] {$2$};
        \node[main,bigg] (1) [right = of 2] {$1$};

        \path   (n) edge (n-1)
                (n-1) edge (dots)
                (dots) edge (2)
                (2) edge (1);
    \end{tikzpicture} \\
    \begin{tikzpicture}[bigg/.style = {minimum size = 3.2em}]
        \node[main,bigg] (n) {$n$};
        \node[main,bigg] (n-1) [right = of n] {$n-1$};
        \node (dots) [right = of n-1] {$\cdots$};
        \node[main,bigg] (4) [right = of dots] {$4$};
        \node[main,bigg] (3) [right = of 4] {$3$};
        \def\dist{5.8em};
        \path   (3) -- node[pos=1,main,bigg] (2) {$2$} +(30:\dist)
                (3) -- node[pos=1,main,bigg] (1) {$1$} +(-30:\dist);

        \path   (n) edge (n-1)
                (n-1) edge (dots)
                (dots) edge (4)
                (4) edge (3)
                (3) edge (2) edge (1)
                (2) edge (1) -- node[pos=1] {} +(0, \picgap+1em);
    \end{tikzpicture}
\end{center}
Rewriting in terms of $n$, the worst case convergence time is $n-2$ for bipartite graphs and $2n-5$ for non-bipartite graphs, since a straight line is not bipartite, so $D\leq n-2$ for non-bipartite graphs.

\section{Strongly Connected Graphs}
We are now ready to tackle the directed case, starting off with strongly connected graphs, where for any two nodes $u,v\in V$ it is possible for $u$ to reach $v$ and vice versa.

\subsection{Motivating Example}
Let's see how the process plays out on a strongly connected graph composed of two cycles, one of length 3 and one of length 6.
\begin{center}
    \begin{tikzpicture}
        \def\dist{5em};
        \node[main] (1) at (120:\dist) {3};
        \node[main] (2) at (60:\dist) {5};
        \node[main] (3) at (0:\dist) {2};
        \node[main] (4) at (-60:\dist) {6};
        \node[main] (5) at (-120:\dist) {5};
        \node[main,redd] (6) at (180:\dist) {1};
        \path (5) -- node[pos=1,main] (7) {4} +(-\dist, 0);
        \path (6) -- node[pos=1] {$t=0$} +(-\dist-2em, 0);
        \foreach \from/\to in {1/2, 2/3, 3/4, 4/5, 5/6, 6/1}
            \path[->] (\from) edge[bend left = 20] (\to);
        \foreach \from/\to in {6/7, 7/5}
            \path[->] (\from) edge[bend right = 20] (\to);
    \end{tikzpicture}
\end{center}
For this graph, each node has one outgoing edge to pull from with the exception of the node highlighted in red, which takes the maximum of its two outgoing edges. As a result, most of the values simply shift backwards along their corresponding cycle, but the $3$ is replaced with the $4$, which is higher.
\begin{center}
    \begin{tikzpicture}
        \def\dist{5em};
        \node[main] (1) at (120:\dist) {5};
        \node[main] (2) at (60:\dist) {2};
        \node[main] (3) at (0:\dist) {6};
        \node[main] (4) at (-60:\dist) {5};
        \node[main] (5) at (-120:\dist) {1};
        \node[main] (6) at (180:\dist) {4};
        \path (5) -- node[pos=1,main] (7) {5} +(-\dist, 0);
        \path (6) -- node[pos=1] {$t=1$} +(-7em, 0);
        \foreach \from/\to in {1/2, 2/3, 3/4, 4/5, 5/6, 6/1}
            \path[->] (\from) edge[bend left = 20] (\to);
        \foreach \from/\to in {6/7, 7/5}
            \path[->] (\from) edge[bend right = 20] (\to);
    \end{tikzpicture}
\end{center}
Continuing this process:
\begin{center}
    \begin{tikzpicture}
        \def\dist{5em};
        \node[main] (1) at (120:\dist) {2};
        \node[main] (2) at (60:\dist) {6};
        \node[main] (3) at (0:\dist) {5};
        \node[main] (4) at (-60:\dist) {1};
        \node[main] (5) at (-120:\dist) {4};
        \node[main] (6) at (180:\dist) {5};
        \path (5) -- node[pos=1,main] (7) {1} +(-\dist, 0);
        \path (6) -- node[pos=1] {$t=2$} +(-7em, 0);
        \foreach \from/\to in {1/2, 2/3, 3/4, 4/5, 5/6, 6/1}
            \path[->] (\from) edge[bend left = 20] (\to);
        \foreach \from/\to in {6/7, 7/5}
            \path[->] (\from) edge[bend right = 20] (\to);
    \end{tikzpicture} \\
    \begin{tikzpicture}
        \def\dist{5em};
        \node[main] (1) at (120:\dist) {6};
        \node[main] (2) at (60:\dist) {5};
        \node[main] (3) at (0:\dist) {1};
        \node[main] (4) at (-60:\dist) {4};
        \node[main] (5) at (-120:\dist) {5};
        \node[main] (6) at (180:\dist) {2};
        \path (5) -- node[pos=1,main] (7) {4} +(-\dist, 0);
        \path (6) -- node[pos=1] {$t=3$} +(-7em, 0);
        \path (1) -- node[pos=1] {} +(0, \picgap);
        \foreach \from/\to in {1/2, 2/3, 3/4, 4/5, 5/6, 6/1}
            \path[->] (\from) edge[bend left = 20] (\to);
        \foreach \from/\to in {6/7, 7/5}
            \path[->] (\from) edge[bend right = 20] (\to);
    \end{tikzpicture} \\
    \begin{tikzpicture}
        \def\dist{5em};
        \node[main] (1) at (120:\dist) {5};
        \node[main] (2) at (60:\dist) {1};
        \node[main] (3) at (0:\dist) {4};
        \node[main] (4) at (-60:\dist) {5};
        \node[main] (5) at (-120:\dist) {2};
        \node[main] (6) at (180:\dist) {6};
        \path (5) -- node[pos=1,main] (7) {5} +(-\dist, 0);
        \path (6) -- node[pos=1] {$t=4$} +(-7em, 0);
        \path (1) -- node[pos=1] {} +(0, \picgap);
        \foreach \from/\to in {1/2, 2/3, 3/4, 4/5, 5/6, 6/1}
            \path[->] (\from) edge[bend left = 20] (\to);
        \foreach \from/\to in {6/7, 7/5}
            \path[->] (\from) edge[bend right = 20] (\to);
    \end{tikzpicture} \\
    \begin{tikzpicture}
        \def\dist{5em};
        \node[main] (1) at (120:\dist) {1};
        \node[main] (2) at (60:\dist) {4};
        \node[main] (3) at (0:\dist) {5};
        \node[main] (4) at (-60:\dist) {2};
        \node[main] (5) at (-120:\dist) {6};
        \node[main] (6) at (180:\dist) {5};
        \path (5) -- node[pos=1,main] (7) {2} +(-\dist, 0);
        \path (6) -- node[pos=1] {$t=5$} +(-7em, 0);
        \path (1) -- node[pos=1] {} +(0, \picgap);
        \foreach \from/\to in {1/2, 2/3, 3/4, 4/5, 5/6, 6/1}
            \path[->] (\from) edge[bend left = 20] (\to);
        \foreach \from/\to in {6/7, 7/5}
            \path[->] (\from) edge[bend right = 20] (\to);
    \end{tikzpicture} \\
    \begin{tikzpicture}
        \def\dist{5em};
        \node[main] (1) at (120:\dist) {4};
        \node[main] (2) at (60:\dist) {5};
        \node[main] (3) at (0:\dist) {2};
        \node[main] (4) at (-60:\dist) {6};
        \node[main] (5) at (-120:\dist) {5};
        \node[main] (6) at (180:\dist) {2};
        \path (5) -- node[pos=1,main] (7) {6} +(-\dist, 0);
        \path (6) -- node[pos=1] {$t=6$} +(-7em, 0);
        \path (1) -- node[pos=1] {} +(0, \picgap);
        \foreach \from/\to in {1/2, 2/3, 3/4, 4/5, 5/6, 6/1}
            \path[->] (\from) edge[bend left = 20] (\to);
        \foreach \from/\to in {6/7, 7/5}
            \path[->] (\from) edge[bend right = 20] (\to);
    \end{tikzpicture} \\
    \begin{tikzpicture}
        \def\dist{5em};
        \node[main] (1) at (120:\dist) {5};
        \node[main] (2) at (60:\dist) {2};
        \node[main] (3) at (0:\dist) {6};
        \node[main] (4) at (-60:\dist) {5};
        \node[main] (5) at (-120:\dist) {2};
        \node[main] (6) at (180:\dist) {6};
        \path (5) -- node[pos=1,main] (7) {5} +(-\dist, 0);
        \path (6) -- node[pos=1] {$t=7$} +(-7em, 0);
        \path (1) -- node[pos=1] {} +(0, \picgap);
        \foreach \from/\to in {1/2, 2/3, 3/4, 4/5, 5/6, 6/1}
            \path[->] (\from) edge[bend left = 20] (\to);
        \foreach \from/\to in {6/7, 7/5}
            \path[->] (\from) edge[bend right = 20] (\to);
    \end{tikzpicture} \\
    \begin{tikzpicture}
        \def\dist{5em};
        \node[main] (1) at (120:\dist) {2};
        \node[main] (2) at (60:\dist) {6};
        \node[main] (3) at (0:\dist) {5};
        \node[main] (4) at (-60:\dist) {2};
        \node[main] (5) at (-120:\dist) {6};
        \node[main] (6) at (180:\dist) {5};
        \path (5) -- node[pos=1,main] (7) {2} +(-\dist, 0);
        \path (6) -- node[pos=1] {$t=8$} +(-7em, 0);
        \path (1) -- node[pos=1] {} +(0, \picgap);
        \foreach \from/\to in {1/2, 2/3, 3/4, 4/5, 5/6, 6/1}
            \path[->] (\from) edge[bend left = 20] (\to);
        \foreach \from/\to in {6/7, 7/5}
            \path[->] (\from) edge[bend right = 20] (\to);
    \end{tikzpicture} \\
    \begin{tikzpicture}
        \def\dist{5em};
        \node[main] (1) at (120:\dist) {6};
        \node[main] (2) at (60:\dist) {5};
        \node[main] (3) at (0:\dist) {2};
        \node[main] (4) at (-60:\dist) {6};
        \node[main] (5) at (-120:\dist) {5};
        \node[main] (6) at (180:\dist) {2};
        \path (5) -- node[pos=1,main] (7) {6} +(-\dist, 0);
        \path (6) -- node[pos=1] {$t=9$} +(-7em, 0);
        \path (1) -- node[pos=1] {} +(0, \picgap);
        \foreach \from/\to in {1/2, 2/3, 3/4, 4/5, 5/6, 6/1}
            \path[->] (\from) edge[bend left = 20] (\to);
        \foreach \from/\to in {6/7, 7/5}
            \path[->] (\from) edge[bend right = 20] (\to);
    \end{tikzpicture} \\
    \begin{tikzpicture}
        \def\dist{5em};
        \node[main] (1) at (120:\dist) {5};
        \node[main] (2) at (60:\dist) {2};
        \node[main] (3) at (0:\dist) {6};
        \node[main] (4) at (-60:\dist) {5};
        \node[main] (5) at (-120:\dist) {2};
        \node[main] (6) at (180:\dist) {6};
        \path (5) -- node[pos=1,main] (7) {5} +(-\dist, 0);
        \path (6) -- node[pos=1] {$t=10$} +(-7em, 0);
        \path (1) -- node[pos=1] {} +(0, \picgap);
        \foreach \from/\to in {1/2, 2/3, 3/4, 4/5, 5/6, 6/1}
            \path[->] (\from) edge[bend left = 20] (\to);
        \foreach \from/\to in {6/7, 7/5}
            \path[->] (\from) edge[bend right = 20] (\to);
    \end{tikzpicture}
\end{center}
We can stop here as the graph for $t=10$ looks the same as $t=7$, meaning that we will cycle around those three states forever. This means the convergence time for this particular example is 7 and the period is 3.

Examining the graph after stabilisation, notice that any path we take alternates 2, 6, 5, 2, 6, 5, $\cdots$. It looks like we have divided the graph into three separate interlocking \emph{classes}, and that each class has a single value that all its members have assumed:
\begin{center}
    \begin{tikzpicture}
        \def\dist{5em};
        \node[main,redd] (1) at (120:\dist) {5};
        \node[main,bluu] (2) at (60:\dist) {2};
        \node[main,grnn] (3) at (0:\dist) {6};
        \node[main,redd] (4) at (-60:\dist) {5};
        \node[main,bluu] (5) at (-120:\dist) {2};
        \node[main,grnn] (6) at (180:\dist) {6};
        \path (5) -- node[pos=1,main,redd] (7) {5} +(-\dist, 0);
        \path (6) -- node[pos=1] {$t=10$} +(-7em, 0);
        \path[->]   (1) edge[redd,bend left = 20] (2)
                    (2) edge[bluu,bend left = 20] (3)
                    (3) edge[grnn,bend left = 20] (4)
                    (4) edge[redd,bend left = 20] (5)
                    (5) edge[bluu,bend left = 20] (6)
                    (6) edge[grnn,bend left = 20] (1)
                        edge[grnn,bend right = 20] (7)
                    (7) edge[redd,bend right = 20] (5);
    \end{tikzpicture}
\end{center}
The period can thus be determined by the number of classes. The fact that there are 3 classes makes sense; if there were 2 classes for example, we would be able to colour the 6-cycle red, blue, red, blue, red, blue, but we would not be able to colour the 3-cycle as there would be two adjacent nodes of the same colour. This gives us the idea that \emph{the period must divide all cycle lengths}, a fact we show in the subsequent section.

Afterwards, we will use \emph{equivalence relations} to define \emph{equivalence classes}, giving us a formal notion of the colour of each node depicted above. This will enable us to gain a complete understanding of the convergence of this process on strongly connected graphs, and create examples of graphs with arbitrary periods.

\subsection{The Period and Cycles}
So far, we have thought of the update process as multiple iterations of the same idea, which is for each node to pull from its immediate neighbours. This does not give us much freedom to look ahead and consider the value of a node many iterations down the line, but fortunately we can show that this is possible. If we are at time $t$ and wish to know the value of $u$ in $k$ steps time, we only need to check the nodes a distance of $k$ away from $u$, as these are precisely the vertices that will affect $u$ in $k$ steps time.
\begin{lem}\label{inducto}
    Let $G=(V,E)$ be a graph at time $t\geq 0$. Then the following holds for all $u\in V$ and $k\geq 0$:
    \[ f_{t+k}(u) = \max_{v\in\Gamma^k(u)}f_t(v) \]
\end{lem}
\begin{proof}
    We can show this by induction. Since $\Gamma^0(u)=\{u\}$, it clearly holds for $k=0$. Suppose it holds for $k-1$. Then for $k$, we can apply the inductive hypothesis, followed by the update rule:
    \begin{align*}
        f_{t+k}(u)  &= f_{t+1+k-1}(u) \\
                    &= \max_{v\in\Gamma^{k-1}(u)}f_{t+1}(v) \\
                    &= \max_{v\in\Gamma^{k-1}(u)}\left(\max_{w\in\Gamma(v)}f_t(w)\right)
    \end{align*}
    But then the set of $w$ being considered is:
    \[ \{w\in V : w\in\Gamma(v)\text{ for some }v\in\Gamma^{k-1}(u)\} \]
    By \Cref{eqn:Gamma1} (the equivalence definition of $\Gamma$), we get:
    \[ \{w\in V : (v,w)\in E\text{ for some }v\in\Gamma^{k-1}(u)\} \]
    But this is just the definition of $\Gamma^k(u)$, and thus:
    \[ f_{t+k}(u) = \max_{w\in\Gamma^k(u)}f_t(w)\]
\end{proof}
This fact will prove itself useful in future proofs, but the first thing to note is that we only ever pull values from vertices in $V$, meaning that the range of possible values is bounded. Using a similar argument to \Cref{markovFiniteLB} in the previous chapter, we can see that the Markov chain for this process is finite too, so by \Cref{absorb}, we eventually converge.
\begin{prop}\label{markovFiniteSM}
    Let $G=(V,E)$ be a graph and $f_0$ be an initial valuation. Then only finitely many valuations are reachable from $f_0$.
\end{prop}
\begin{proof}
    By setting $t=0$ in \Cref{inducto}, for all $v\in V$ any future value $f_k(v)$ is in the finite set $\{f_0(u) : u\in V\}$, and thus there are only finitely many possible future valuations.
\end{proof}
Now, we wish to analyse the convergence behaviour of the Markov chain, which describes the entire graph. This is a complex task; it is possible for different areas of the graph to converge into different-length repetitive cycles with distinct periods from each other, and it is not immediately clear how these interact over time. To make things easier, let's consider the convergence behaviours of one vertex $u\in V$ at a time, i.e.\ the sequence $f_0(u),f_1(u),f_2(u),\cdots$.
\begin{center}
    \begin{tikzpicture}[]
        \def\dist{5em};
        \node[main] (1) at (120:\dist) {3};
        \node[main] (2) at (60:\dist) {5};
        \node[main] (3) at (0:\dist) {2};
        \node[main] (4) at (-60:\dist) {6};
        \node[main] (5) at (-120:\dist) {5};
        \node[main] (6) at (180:\dist) {1};
        \path (5) -- node[pos=1,main] (7) {4} +(-\dist, 0);
        \foreach \from/\to in {1/2, 2/3, 3/4, 4/5, 5/6, 6/1}
            \path[->] (\from) edge[bend left = 20] (\to);
        \foreach \from/\to in {6/7, 7/5}
            \path[->] (\from) edge[bend right = 20] (\to);

        \path (2) -- node[pos=1, align=center, text=red] {\textbf{5}\\\textbf{2}\\\textbf{6}\\\textbf{$\vdots$}} ++(2.1em, 2em) -- node[pos=1, text=red] {\Forward} +(-1em, 2em);
    \end{tikzpicture}
\end{center}
We can similarly consider the convergence of all vertices:
\begin{center}
    \begin{tikzpicture}[]
        \def\dist{5em};
        \node[main] (1) at (120:\dist) {3};
        \node[main] (2) at (60:\dist) {5};
        \node[main] (3) at (0:\dist) {2};
        \node[main] (4) at (-60:\dist) {6};
        \node[main] (5) at (-120:\dist) {5};
        \node[main] (6) at (180:\dist) {1};
        \path (5) -- node[pos=1,main] (7) {4} +(-\dist, 0);
        \foreach \from/\to in {1/2, 2/3, 3/4, 4/5, 5/6, 6/1}
            \path[->] (\from) edge[bend left = 20] (\to);
        \foreach \from/\to in {6/7, 7/5}
            \path[->] (\from) edge[bend right = 20] (\to);

        \path   (1) -- node[pos=1, align=center, text=red] {\textbf{3}\\\textbf{5}\\\textbf{2}\\\textbf{$\vdots$}} ++(2em, 2em) -- node[pos=1, text=red] {\Forward} +(-1em, 2em)
                (2) -- node[pos=1, align=center, text=red] {\textbf{5}\\\textbf{2}\\\textbf{6}\\\textbf{$\vdots$}} ++(2.1em, 2em) -- node[pos=1, text=red] {\Forward} +(-1em, 2em)
                (3) -- node[pos=1, align=center, text=red] {\textbf{2}\\\textbf{6}\\\textbf{5}\\\textbf{$\vdots$}} ++(2.1em, -0.6em) -- node[pos=1, text=red] {\Forward} +(-1em, 2em)
                (4) -- node[pos=1, align=center, text=red] {\textbf{6}\\\textbf{5}\\\textbf{1}\\\textbf{$\vdots$}} ++(2.1em, -0.6em) -- node[pos=1, text=red] {\Forward} +(-1em, 2em)
                (5) -- node[pos=1, align=center, text=red] {\textbf{5}\\\textbf{1}\\\textbf{4}\\\textbf{$\vdots$}} ++(2.1em, -0.6em) -- node[pos=1, text=red] {\Forward} +(-1em, 2em)
                (6) -- node[pos=1, align=center, text=red] {\textbf{1}\\\textbf{4}\\\textbf{5}\\\textbf{$\vdots$}} ++(-2.1em, 2em) -- node[pos=1, text=red] {\Forward} +(-1em, 2em)
                (7) -- node[pos=1, align=center, text=red] {\textbf{4}\\\textbf{5}\\\textbf{1}\\\textbf{$\vdots$}} ++(-2.1em, -2em) -- node[pos=1, text=red] {\Forward} +(-1em, 2em);
    \end{tikzpicture}
\end{center}
Since the overall graph eventually converges, any individual vertex will eventually converge and have its own \emph{local period}. For example, if a vertex alternates 3, 4, 3, 4, $\cdots$ forever, then its local period is 2, regardless of the periods of other vertices or the overall \emph{global period} of the graph.
\begin{defn}[Local period]
    Let $G=(V,E)$ be a graph and $f_0$ an initial valuation, and let $u\in V$. Since the Markov chain of $G$ is finite, by \Cref{direcPeriod}, the sequence $f_0,f_1,f_2,\cdots$ eventually enters a cycle, so there exists $T\geq 0$ and $p\geq 1$ such that $f_{t+p}=f_t$ for all $t\geq T$. The \textbf{local period} is the minimum such $p$.
\end{defn}
Now that we have defined local periods, the primary objects of interest are cycles, as they can help us determine the local periods of vertices on them. Throughout the next section, we will denote cycles of length $k$ as $v_1\to v_2\to\cdots\to v_k\to v_1$. The subscripts wrap around, so $v_{k+1}\coloneqq v_1$, and $v_{2k+5}=v_5$, and so on.

It turns out that the sum of any cycle is a potential function which can only increase over time, and cannot decrease.
\begin{lem}\label{cycleMonotonic}
    Let $v_1\to v_2\to\cdots\to v_k\to v_1$ be a cycle of length $k$. Then for all $t\geq 0$:
    \[ \sum_{i=1}^kf_t(v_i) \leq \sum_{i=1}^kf_{t+1}(v_i) \]
\end{lem}
\begin{proof}
    After one update, for all $i$, the new value of $v_i$ is at least the old value of $v_{i+1}$ (possibly more), since vertices update off of their neighbours and there is an edge $v_i\to v_{i+1}$. Formally, we know $f_{t+1}(v_i) \geq f_t(v_{i+1})$ since $(v_i,v_{i+1})\in E$, so:
    \begin{align*}
        \sum_{i=1}^kf_{t+1}(v_i) &\geq \sum_{i=1}^kf_t(v_{i+1}) \\
        &= \sum_{i=1}^kf_t(v_i)
    \end{align*}
\end{proof}
Since we eventually reach a cycle in the Markov chain (i.e.\ an absorbing state where no further changes are possible), there must come a point where the sum of any cycle can increase no further. At this point, we can predict the value of any node along the cycle by simply looking further along the cycle.
\begin{lem}\label{cyclePredict}
    Let $v_1\to v_2\to\cdots\to v_k\to v_1$ be a cycle of length $k$, and suppose $f_t$ is in an absorbing state (i.e.\ the process has converged). Then $f_{t+\ell}(v_i) = f_t(v_{i+\ell})$ for all $\ell\geq 0$ and all $i$.
\end{lem}
\begin{proof}
    Since $v_{i+\ell}\in\Gamma^\ell(v_i)$, by \Cref{inducto} $v_i$ pulls from $v_{i+\ell}$ and is thus at least as big, that is, $f_{t+\ell}(v_i) \geq f_t(v_{i+\ell})$. Suppose for the sake of contradiction that it was strictly larger for some $\ell\geq 0$ and $i$, that is, $f_{t+\ell}(v_i) > f_t(v_{i+\ell})$. Note that for all $j$ we have $f_{t+\ell}(v_j) \geq f_t(v_{j+\ell})$. Applying the same argument as \Cref{cycleMonotonic}:
    \begin{align*}
        \sum_{j=1}^kf_{t+\ell}(v_j) &> \sum_{j=1}^kf_t(v_{j+\ell}) \\
        &= \sum_{j=1}^kf_t(v_j)
    \end{align*}
    But this contradicts \Cref{cycleMonotonic}, and thus we are done.
\end{proof}
Thus, after convergence, values will just be pushed around cycles, and any given cycle will shuffle the same values around forever in sequence, without the sum growing any larger. Two important observations follow this fact, which, when combined, paint a clear picture of the convergence process. First, if a cycle of length $k$ is pushing some values around, then the length of the sequence of values being pushed around (i.e.\ the period) had better divide $k$, otherwise it will not fit neatly in the cycle.
\begin{prop}\label{cyclePeriodDivide}
    Let $v_1\to v_2\to\cdots\to v_k\to v_1$ be a cycle of length $k$. Then for all $i$, the local period of $v_i$ divides $k$.
\end{prop}
\begin{proof}
    Let $T\geq 0$ such that $f_T$ is in an absorbing state, so the process has converged, then fix $i$. Then for all $t\geq T$, by \Cref{cyclePredict} but with $\ell=k$, we know $f_{t+k}(v_i)=f_t(v_{i+k})=f_t(v_i)$. Thus, the value at $v_i$ is the same every $k$ iterations, so its local period is at most $k$. Below is an elementary argument that shows this period divides $k$.

    Let $p$ be the local period of $v_i$, i.e.\ the minimum $p\geq 1$ such that $f_{t+p}(v_i)=f_t(v_i)$ for all $t\geq T$. Since $k$ satisfies this property, $p\leq k$, so using the division algorithm we can write $k=mp+r$ for some $m\geq 0$ and $0\leq r<p$. Then for all $t\geq T$:
    \begin{align*}
        f_{t+k}(v_i) &= f_t(v_i) \\
        f_{t+mp+r}(v_i) &= f_t(v_i) \\
        f_{t+(m-1)p+r}(v_i) &= f_t(v_i) \\
        \vdots \\
        f_{t+r}(v_i) &= f_t(v_i)
    \end{align*}
    But $r<p$ and $p$ is supposed to be the minimum $x\geq 1$ which satisfies this property. Thus, the only possibility is $r=0$, meaning $k=mp$ and thus $p$ divides $k$.
\end{proof}
This means we can determine information about a vertex by looking at the cycles it lies on. For example, if $v$ is on a cycle of length 6 and another cycle of length 15, then the local period of $v$ must divide both 6 and 15, and is thus either 1 or 3.

The second important observation is that if the values of a node in the cycle are simply the values of a node further down the cycle, but delayed, then the local periods of these two nodes must be identical.
\begin{prop}\label{localsIdenticalWeak}
    Let $v_1\to v_2\to\cdots\to v_k\to v_1$ be a cycle of length $k$. Then the local periods of $v_1,v_2,\cdots,v_k$ are identical.
\end{prop}
\begin{proof}
    Let $T\geq 0$ such that $f_T$ is in an absorbing state, and fix $i$. Let $p$ be the local period of $v_{i+1}$. Then by repeatedly applying \Cref{cyclePredict}, we know for all $t\geq T+1$:
    \begin{align*}
        f_{t+p}(v_i) &= f_{t+p-1}(v_{i+1}) \\
        &= f_{t-1}(v_{i+1}) \\
        &= f_t(v_i)
    \end{align*}
    Thus, $p$ satisfies the local period property for $v_i$, meaning the real local period of $v_i$ is less than or equal to $p$, the local period of $v_{i+1}$. But since the choice of $i$ was arbitrary, this holds for all vertices in the cycle. The only possibility is that they all have the same local period.
\end{proof}
When we combine this fact with the assumption that the graph is strongly connected, meaning any two nodes can be placed on a cycle that goes through both of them, we arrive at an important observation that justifies the detour into analysing local periods.
\begin{prop}\label{localsIdenticalStrong}
    If $G=(V,E)$ is strongly connected, then the local periods of all vertices are identical and equal to the global period $p$.
\end{prop}
\begin{proof}
    Let $u,v\in V$ be arbitrary. Since $G$ is strongly connected, there is a path from $u$ to $v$ and a path from $v$ to $u$, forming a cycle that passes through both $u$ and $v$. Note that it does not matter if this cycle contains duplicate vertices. By \Cref{localsIdenticalWeak}, the local periods of all vertices on this cycle must be equivalent, and thus $u$ and $v$ have the same local period. Therefore, all vertices in $G$ have the same local period $\ell$.

    Let $f_T$ be in an absorbing state and suppose $p$ is the global period, so $f_{t+p}=f_t$ (as functions) for all $t\geq T$. Then for any fixed $v\in V$ we have $f_{t+p}(v)=f_t(v)$, so the local period of $v$ is at most $p$ and thus $\ell\leq p$. But conversely, for all $v\in V$ we have $f_{t+\ell}(v)=f_t(v)$, and thus $f_{t+\ell}=f_t$ as functions, so the global period is at most $\ell$ and thus $p\leq\ell$. Thus, the local periods $\ell$ are equal to the global period $p$.
\end{proof}
Thus, we have established that all periods, local or global, are in fact one singular number tied to the graph and valuation, consistent across all vertices. This means information we learn by focussing on a specific cycle buried somewhere in the graph will carry across to all other cycles, and the global period itself. Recall from earlier that the local period of a vertex must divide any cycle that it is on. When viewing all periods as the same, this brings us to the first key result of this section. Note that a \emph{simple cycle} refers to a cycle with no repeating vertices.
\begin{thm}\label{pDividesG}
    Let $G=(V,E)$ be a graph with initial valuation $f_0$. Let $g$ be the greatest common divisor of all simple cycle lengths in $G$:
    \[ g=\gcd_{\substack{C\subseteq V \text{ is a} \\ \text{simple cycle}}}|C| \]
    Then the period $p$ of the maximum model divides $g$.
\end{thm}
\begin{proof}
    By \Cref{markovFiniteSM} the period exists, and by \Cref{localsIdenticalStrong} it is equal to the local period of any vertex. If there is a simple cycle $C$ of length $C\geq 1$, then $C$ contains at least one vertex $v\in C$. By \Cref{cyclePeriodDivide} the local period of $v$ divides $|C|$, and thus the global period $p$ divides $|C|$. This means $p$ divides all simple cycle lengths. Since simple cycles cannot repeat vertices, there are at most $n!$ of them, i.e.\ we are taking the GCD of finitely many values. An inductive argument from elementary number theory shows that if $p$ divides a finite set of values then it also divides the GCD of that finite set of values, and thus we are done.
\end{proof}
Note that the purpose of insisting on simple cycles was to enforce that the set were finite; but what if $p$ divides other, non-simple cycles that introduce more information? This turns out not to be the case, since we can decompose non-simple cycles into simple ones.
\begin{prop}\label{decomposeSimple}
    The GCD of all simple cycle lengths $g$ as defined in \Cref{pDividesG} divides the length of any cycle.
\end{prop}
\begin{proof}
    This is shown by induction. It is clearly true for all cycles of length $1$ as these are self-loops which are necessarily simple. Suppose it were true for all (not necessarily simple) cycles of length less than $k$, and let $v_1\to v_2\to\cdots\to v_k\to v_1$ be an arbitrary cycle that is also not necessarily simple, i.e.\ duplicate vertices are permitted. Start at $v_1$ and iterate right until we find a vertex we have already seen before. If this does not happen then the cycle is simple, so $g$ divides its length. Otherwise, let this duplicate be $v_i=v_j$ where $i<j$. Then $v_i\to v_{i+1}\to\cdots\to v_j$ is a cycle of length $j-i<k$, so $g$ divides $j-i$ by the inductive hypothesis (alternatively, $g$ divides $j-i$ because it is a simple cycle). Then since $v_i=v_j$, the remaining vertices $v_1\to v_2\to\cdots\to v_i\to v_{j+1}\to v_{j+2}\to\cdots\to v_k\to v_1$ form a cycle of length $k-(j-i)<k$ (since $j-i>0$), so $g$ also divides $k-(j-i)$ by the inductive hypothesis. Then $g$ divides $j-i$ and $k-(j-i)$, so $g$ divides their sum $k$.
\end{proof}
Let's use our results to analyse the example from earlier.
\begin{center}
    \begin{tikzpicture}[bigg/.style = {scale = 1.2}]
        \def\dist{6.5em};
        \node[main,bigg,redd] (1) at (120:\dist) {5};
        \node[main,bigg,bluu] (2) at (60:\dist) {2};
        \node[main,bigg,grnn] (3) at (0:\dist) {6};
        \node[main,bigg,redd] (4) at (-60:\dist) {5};
        \node[main,bigg,bluu] (5) at (-120:\dist) {2};
        \node[main,bigg,grnn] (6) at (180:\dist) {6};
        \path (5) -- node[pos=1,main,bigg,redd] (7) {5} +(-\dist, 0);
        \path[->]   (1) edge[redd,bend left = 20] (2)
                    (2) edge[bluu,bend left = 20] (3)
                    (3) edge[grnn,bend left = 20] (4)
                    (4) edge[redd,bend left = 20] (5)
                    (5) edge[bluu,bend left = 20] (6)
                    (6) edge[grnn,bend left = 20] (1)
                        edge[grnn,bend right = 20] (7)
                    (7) edge[redd,bend right = 20] (5);
        \node[scale=4.5] {$\circlearrowright$};
        \node[scale=1.8] {\textbf{6}};
        \coordinate (A) at (-\dist-0.6em, -\dist/1.732-0.6em);
        \node[scale=3.25] at (A) {$\circlearrowleft$};
        \node[scale=1.3] at (A) {\textbf{3}};
    \end{tikzpicture}
\end{center}
Notice that this graph has two simple cycles, one of length 6 and one of length 3. The greatest common divisor of these is 3, meaning that the 3-cycle causes the values in the 6-cycle to merge into just three classes. In our example, after stabilisation, the 3-cycle eventually rotates 2, 6, 5, and the 6-cycle eventually rotates 2, 6, 5, 2, 6, 5. Since the GCD is 3, the only possible periods for this graph are 3 and 1; the latter occurs when the graph converges and the red, blue and green classes end up with the same value.

Note that there are infinitely many non-simple cycles in this graph, including cycles of length 9, 12, 15 and so on. However, all cycle lengths are multiples of $3$ by \Cref{decomposeSimple}, so we need not consider these.

Consider another graph:
\begin{center}
    \begin{tikzpicture}
        \def\dist{5em};
        \node[main] (1) {$v_1$};
        \node[main] (2) at (-0.866*\dist, 1.5*\dist) {$v_2$};
        \node[main] (3) at (-0.866*2*\dist, 0) {$v_3$};
        \node[main] (4) at (0.866*\dist, 1.5*\dist) {$v_4$};
        \node[main] (5) at (0.866*3*\dist, 1.5*\dist) {$v_5$};
        \node[main] (6) at (0.866*2*\dist, 0) {$v_6$};

        \path[->]   (1) edge (2) edge (4)
                    (2) edge (3)
                    (3) edge (1)
                    (4) edge (5)
                    (5) edge (6)
                    (6) edge (1);

        \coordinate (A) at (-0.866*\dist, 0.6*\dist);
        \coordinate (B) at (0.866*1.5*\dist, 0.75*\dist);
        \node[scale=4] at (A) {$\circlearrowleft$};
        \node[scale=1.6] at (A) {\textbf{3}};
        \node[scale=4] at (B) {$\circlearrowright$};
        \node[scale=1.6] at (B) {\textbf{4}};
    \end{tikzpicture}
\end{center}
The simple cycles of this graph have lengths 3 and 4, meaning the greatest common divisor is 1. An inductive argument shows that non-simple cycles of length $k$ exist for all $k\geq 6$. This means regardless of the initial valuation, the period of the Markov chain will be one, meaning all vertices will converge to a value which will never change. In fact, this value will be the same across all vertices, since if any two were different, there would be an edge between them, causing a change in the state. Note that this only applies to strongly connected graphs; otherwise we can have examples such as the following:
\begin{center}
    \begin{tikzpicture}
        \def\dist{5em}; 
        
        \node[main] (1) {5};
        \node[main] (2) at (-0.866*\dist, 1.5*\dist) {5};
        \node[main] (3) at (-0.866*2*\dist, 0) {5};
        \node[main] (4) at (0.866*\dist, 1.5*\dist) {4};
        \node[main] (5) at (0.866*3*\dist, 1.5*\dist) {4};
        \node[main] (6) at (0.866*2*\dist, 0) {4};

        \path[->]   (1) edge (2) edge (4) edge (6)
                    (2) edge (3) edge (4)
                    (3) edge (1)
                    (4) edge (5)
                    (5) edge (6)
                    (6) edge (4);

        \coordinate (cA) at (-0.866*\dist, 0.5*\dist);
        \path   (cA)    -- coordinate[pos=2] (c1) (1)
                (cA)    -- coordinate[pos=2] (c2) (2)
                (cA)    -- coordinate[pos=2] (c3) (3);
        \draw[scc, rounded corners = 2.5em] (c1) -- (c2) -- (c3) -- cycle;

        \coordinate (cB) at (0.866*2*\dist, \dist);
        \path   (cB)    -- coordinate[pos=2] (c4) (4)
                (cB)    -- coordinate[pos=2] (c5) (5)
                (cB)    -- coordinate[pos=2] (c6) (6);
        \draw[scc, rounded corners = 2.5em] (c4) -- (c5) -- (c6) -- cycle;
    \end{tikzpicture}
\end{center}
This graph has converged with a period of 1, but not all values are identical.

Finally, recall that a connected undirected graph can be viewed as a strongly connected directed graph with the property that $(u,v)\in E\iff(v,u)\in E$. Any non-trivial undirected graph must have at least one edge by \Cref{assNeighbour}, and thus has at least two edges. These two edges form a simple cycle of length 2, meaning that the GCD of all simple cycle lengths is either 1 or 2. A non-bipartite graph necessarily has an odd length simple cycle, and since 2 does not divide the length of such a cycle, the GCD and thus the period must be 1. Otherwise, for bipartite graphs, the GCD is 2, so the period can be either 1 or 2; this all co-incides with the result we showed in the last section.

However, in the last section we showed that for bipartite graphs, the period is 1 if and only if $m=m'$, where $m$ and $m'$ are the maxima of the respective bipartite components. These maxima eventually spread to their entire bipartite component, and if they are equal, then the entire graph assumes the same value. In the next section, we develop the tools to show that this extends to strongly connected graphs as well. That is, each node will have a certain colour, and eventually the entire colour will assume the value of the maximum of that colour. This means we can predict the convergence state from the initial valuation by looking at the maximum value within each colour.

\subsection{The Period and Equivalence Classes}
Consider the multi-coloured example from before:
\begin{center}
    \begin{tikzpicture}[bigg/.style = {scale = 1.2}]
        \def\dist{6.5em};
        \node[main,bigg,redd] (1) at (120:\dist) {5};
        \node[main,bigg,bluu] (2) at (60:\dist) {2};
        \node[main,bigg,grnn] (3) at (0:\dist) {6};
        \node[main,bigg,redd] (4) at (-60:\dist) {5};
        \node[main,bigg,bluu] (5) at (-120:\dist) {2};
        \node[main,bigg,grnn] (6) at (180:\dist) {6};
        \path (5) -- node[pos=1,main,bigg,redd] (7) {5} +(-\dist, 0);
        \path[->]   (1) edge[redd,bend left = 20] (2)
                    (2) edge[bluu,bend left = 20] (3)
                    (3) edge[grnn,bend left = 20] (4)
                    (4) edge[redd,bend left = 20] (5)
                    (5) edge[bluu,bend left = 20] (6)
                    (6) edge[grnn,bend left = 20] (1)
                        edge[grnn,bend right = 20] (7)
                    (7) edge[redd,bend right = 20] (5);
        \node[scale=4.5] {$\circlearrowright$};
        \node[scale=1.8] {\textbf{6}};
        \coordinate (A) at (-\dist-0.6em, -\dist/1.732-0.6em);
        \node[scale=3.25] at (A) {$\circlearrowleft$};
        \node[scale=1.3] at (A) {\textbf{3}};
    \end{tikzpicture}
\end{center}
We can see that the nodes have separated themselves into three interlocking classes, but how can we define these? More specifically, let's say $u$ is the topmost red node. How can we locate other nodes that should be coloured red? We can see that we have to take three steps to get to the next red node, then six to get to the one after. This gives us the idea that we can define these equivalence classes based on the neighbourhoods $\Gamma$ defined in \Cref{chap:background}.

Note that the next section on equivalence relations are properties of the graph, not any particular valuation, and not the period. Even if all values in the above example were identical and therefore never changed, resulting in a period of 1, we would still colour with three colours, as that is the GCD of all simple cycle lengths. To define our equivalence classes, we will define an equivalence relation $uPv$, noting that we already defined $uRv$ as `$v$ is reachable from $u$' and $uSv$ as `$u$ and $v$ are in the same SCC' in \Cref{chap:background}.
\begin{defn}[$g$-path relation]\label{pDefn}
    Let $G=(V,E)$ be a strongly connected directed graph, let $g$ be the GCD of all simple cycles as defined in \Cref{pDividesG} and let $u,v\in V$ be nodes. Then $uPv$ if and only if there exists a path of length $k$ from $u$ to $v$ where $g$ divides $k$.
\end{defn}
Note that this path does not have to be simple. We can now show that this is an equivalence relation.
\begin{prop}
    $P$ is an equivalence relation, meaning it is reflexive, symmetric and transitive.
\end{prop}
\begin{proof}
    \phantom{a}

    \textbf{Reflexive:} \\
    Clearly $uPu$ as there is a path of length $0$ from $u$ to itself, and $g$ divides $0$.

    \textbf{Symmetric:} \\
    Suppose $uPv$, so there is a path of length $kg$ from $u$ to $v$, for some $k\geq 0$. Then since $G$ is strongly connected, there is a path from $v$ to $u$; let $\ell$ be its length. Then we have a cycle from $u$ to itself of length $kg+\ell$, so $g$ divides $kg+\ell$ by \Cref{decomposeSimple}, and thus $g$ divides $\ell$. Thus, $g$ divides the length of a path from $v$ to $u$, so $vRu$.
    \begin{center}
        \begin{tikzpicture}[node distance = 8em]
            \node[main] (u) {$u$};
            \node[main] (v) [right = of u] {$v$};
            \path[->]   (u) edge[bend left = 25] node[fill=white] {$\cdots$} node[above] {$kg$} (v)
                        (v) edge[bend left = 25] node[fill=white] {$\cdots$} node[below] {$\ell$} (u);
        \end{tikzpicture}
    \end{center}

    \textbf{Transitive:} \\
    Suppose $uPv$ and $vPw$, so there is a path of length $kg$ from $u$ to $v$ and a path of length $\ell g$ from $v$ to $w$. Then we have a path of length $kg+\ell g=(k+\ell)g$ from $u$ to $w$, so $uPw$.
    \begin{center}
        \begin{tikzpicture}[node distance = 8em]
            \node[main] (u) {$u$};
            \node[main] (v) [right = of u] {$v$};
            \node[main] (w) [right = of v] {$w$};
            \path[->]   (u) edge node[fill=white] {$\cdots$} node[above] {$kg$} (v)
                        (v) edge node[fill=white] {$\cdots$} node[above] {$\ell g$} (w);
        \end{tikzpicture}
    \end{center}
\end{proof}
Note that the reflexive and transitive proofs did not use the fact that $g$ was the GCD of all simple cycle lengths, or that the graph was strongly connected, but the symmetric proof required both.

The relation holding only guarantees the existence of a path whose length is a multiple of $g$; it does not say anything about whether the other paths' lengths are multiples of $g$. It would be challenging to work with a relation like this, where there are multiple paths from $u$ to $v$, each with different lengths modulo $g$. Fortunately, the next result shows that this is impossible, and demonstrates that we can exchange `there exists' with `for all' without changing anything.
\begin{prop}\label{existsForall}
    $uPv$ if and only if $g$ divides the lengths of all paths from $u$ to $v$.
\end{prop}
\begin{proof}
    \phantom{a}

    $(\impliedby)$ \\
    Suppose $g$ divides the lengths of all paths from $u$ to $v$. Since $G$ is strongly connected, there is a path from $u$ to $v$, which $g$ divides the length of. Thus, $uPv$.

    $(\implies)$ \\
    Suppose $uPv$, so there exists a path from $u$ to $v$ of length $kg$ for some $k\geq 0$. Since $P$ is symmetric, $vPu$, so there is a path from $v$ to $u$ of length $\ell g$ for some $\ell\geq 0$. Now, any path from $u$ to $v$ of length $m$ induces a cycle from $u$ to itself of length $m+\ell g$. Since $g$ divides the length of any cycle by \Cref{decomposeSimple}, we get that $g$ divides $m+\ell g$ and thus $g$ divides $m$.
    \begin{center}
        \begin{tikzpicture}[node distance = 12em]
            \node[main] (u) {$u$};
            \node[main] (v) [right = of u] {$v$};
            \path[->]   (u) edge[bend left = 35] node[fill=white] {$\cdots$} node[above] {$kg$} (v) edge node[fill=white] {$\cdots$} node[above] {$m$} (v)
                        (v) edge[bend left = 35] node[fill=white] {$\cdots$} node[below] {$\ell g$} (u);
        \end{tikzpicture}
    \end{center}
\end{proof}
Since $P$ is an equivalence relation, it divides the vertices into equivalence classes. That is, for $u\in V$ we denote the class of $u$ as:
\[ [u] \coloneqq \{v\in V : uPv\} \]
In other words, $[u]$ is the set of vertices such that paths from $u$ have lengths which are multiples of $g$. We can first show an equivalent definition which proves to be useful.
\begin{lem}\label{bigcupEquiv}
    \[ [u] = \bigcup_{i=0}^\infty\Gamma^{ig}(u) \]
\end{lem}
\begin{proof}
    \phantom{a}

    $(\subseteq)$ \\
    Let $v\in[u]$, so $uPv$. Then there is a path of length $ig$ from $u$ to $v$ for some $i\geq 0$, so $v\in\Gamma^{ig}(u)$ and thus $v\in\bigcup_{i=0}^\infty\Gamma^{ig}(u)$.

    $(\supseteq)$ \\
    Let $v\in\bigcup_{i=0}^\infty\Gamma^{ig}(u)$. Then for some $i\geq 0$ we have $v\in\Gamma^{ig}(u)$. We can apply an inductive argument to show there is a path from $u$ to $v$ of length $ig$, which $g$ divides, so we are done.
\end{proof}
Now that we have these equivalence classes, it is about time to put them to use. We will show a couple of lemmas and build up to a key result that the maximum initial value of an equivalence class eventually spreads to the entire equivalence class. But there is one major problem. Observe what happens after one iteration of the process on our multi-coloured graph:
\begin{center}
    \begin{tikzpicture}
        \def\dist{5em};
        \node[main,redd] (1) at (120:\dist) {5};
        \node[main,bluu] (2) at (60:\dist) {2};
        \node[main,grnn] (3) at (0:\dist) {6};
        \node[main,redd] (4) at (-60:\dist) {5};
        \node[main,bluu] (5) at (-120:\dist) {2};
        \node[main,grnn] (6) at (180:\dist) {6};
        \path (5) -- node[pos=1,main,redd] (7) {5} +(-\dist, 0);
        \path (6) -- node[pos=1] {$t=7$} +(-7em, 0);
        \path[->]   (1) edge[redd,bend left = 20] (2)
                    (2) edge[bluu,bend left = 20] (3)
                    (3) edge[grnn,bend left = 20] (4)
                    (4) edge[redd,bend left = 20] (5)
                    (5) edge[bluu,bend left = 20] (6)
                    (6) edge[grnn,bend left = 20] (1)
                        edge[grnn,bend right = 20] (7)
                    (7) edge[redd,bend right = 20] (5);
    \end{tikzpicture} \\
    \begin{tikzpicture}
        \def\dist{5em};
        \node[main,redd] (1) at (120:\dist) {2};
        \node[main,bluu] (2) at (60:\dist) {6};
        \node[main,grnn] (3) at (0:\dist) {5};
        \node[main,redd] (4) at (-60:\dist) {2};
        \node[main,bluu] (5) at (-120:\dist) {6};
        \node[main,grnn] (6) at (180:\dist) {5};
        \path (5) -- node[pos=1,main,redd] (7) {2} +(-\dist, 0);
        \path (6) -- node[pos=1] {$t=8$} +(-7em, 0);
        \path (1) -- node[pos=1] {} +(0, \picgap);
        \path[->]   (1) edge[redd,bend left = 20] (2)
                    (2) edge[bluu,bend left = 20] (3)
                    (3) edge[grnn,bend left = 20] (4)
                    (4) edge[redd,bend left = 20] (5)
                    (5) edge[bluu,bend left = 20] (6)
                    (6) edge[grnn,bend left = 20] (1)
                        edge[grnn,bend right = 20] (7)
                    (7) edge[redd,bend right = 20] (5);
    \end{tikzpicture}
\end{center}
If we fix the red, blue and green colouring, then the values within each colouring misalign. In fact, they are only aligned every $g$ iterations. Thus, up until we prove the aforementioned key result, we will only consider the process at times $0,g,2g,\cdots$, so that the values are in their correct locations. After that, we will define an operator that allows us to view the process at intermediate times.
\begin{lem}\label{easyLem}
    For all $t\geq 0$ and $u\in V$, after $gt$ iterations $u$ only pulls from values in its equivalence class, i.e.\ $\Gamma^{gt}(u)\subseteq[u]$.
\end{lem}
\begin{proof}
    This immediately follows from \Cref{bigcupEquiv}:
    \[ \Gamma^{gt}(u)\subseteq\bigcup_{i=0}^\infty\Gamma^{ig}(u)=[u] \]
\end{proof}
The fact that after $gt$ steps $u$ can only reach values in its equivalence class followed quickly from the equivalent definition of $[u]$. We also know that $u$ can reach all values in its equivalence class, since the graph is strongly connected, but we would like to show that eventually the maximum spreads to the entire class; by \emph{eventually}, we mean that for \emph{all} large enough $t$ we have $\Gamma^{gt}(u)=[u]$, so we pull from the entire class on every iteration. At this point, every $g$ steps we will query the entire class, not just a subset of it. If all nodes do this, then the entire class will display the maximum value on every iteration.

This is not immediately clear; considering the above example, the topmost red node can reach the bottom-right red node in 3 steps, but it is not possible in 6. However, for all multiples of three greater than or equal to 9, it is possible, since we can use the little cycle of length 3 to adjust ourselves to the correct parity.

As another example, consider the following graph:
\begin{center}
    \begin{tikzpicture}
        \node[main] (u) {u};
        \path[->]   (u) edge[loop, out=120, in=60, looseness=15] node[sloped, fill=white] {$\cdots$} node[above] {6} (u) edge[loop, out=0, in=-60, looseness=15] node[sloped, fill=white] {$\cdots$} node[sloped, below] {10} (u) edge[loop, out=-120, in=180, looseness=15] node[sloped, fill=white] {$\cdots$} node[sloped, below] {15} (u);
    \end{tikzpicture}
\end{center}
Here, the edges labelled 6, 10 and 15 represent cycle of lengths 6, 10 and 15. Then $u$ can reach itself in 6, 10 and 15. It can reach itself in 12 by taking the cycle of length 6 twice, but it cannot reach itself in 13, since that would require taking an odd cycle, but 15 is too many. Calculating this manually, we can see that $u$ can reach itself in the following amounts:
\[ [0, 6, 10, 12, 15, 16, 18, 20, 21, 22, 24, 25, 26, 27, 28, \geq 30] \]
The smallest integer $T$ for which $u$ can reach itself for all $t\geq T$ is 30, which happens to be the LCM of 6, 10 and 15. We can reach $31,\cdots,34$ by taking the ways to reach $25,\cdots,28$ and simply appending an additional visit to the 6-cycle. To achieve 35, we simply take the 10-cycle twice and the 15-cycle once. Finally, to achieve 36 or higher, we take the known ways to achieve $30,\cdots,35$ and append many visits to the 6-cycle (this can be shown by induction).

These ideas serve as motivation for the following challenging lemma.
\begin{lem}\label{hardLem}
    There exists $T\geq 0$ such that $\Gamma^{gt}(u)=[u]$ for all $t\geq T$ and $u\in V$.
\end{lem}
\begin{proof}
    We know $\Gamma^{gt}(u)\subseteq[u]$ for all $t\geq 0$ by \Cref{easyLem}, so we just need to find $T\geq 0$ such that $[u]\subseteq\Gamma^{gt}(u)$ for all $t\geq T$ and $u\in V$. Let $v\in[u]$; we will show that eventually $v\in\Gamma^{gt}(u)$. We actually need to show that this `eventually' is earlier than some fixed value $T$ which does not depend on $u$ or $v$, but we can ignore this problem for now and resolve it quickly at the end.

    Since $g$ is the GCD of finitely many elements, we can find $k$ cycles of lengths $\ell_1,\cdots,\ell_k$ such that $g=\gcd(\ell_1,\cdots,\ell_k)$. WLOG, assume $\ell_1\leq\cdots\leq\ell_k$. Let $v_i$ be any vertex on the $i$th cycle. Since $G$ is strongly connected, there is a path from $u$ to $v_1$, then from $v_1$ to $v_2$, and so on up to $v_k$, then to $v$.
    \begin{center}
        \begin{tikzpicture}[node distance = 5em, lp/.style = {out=120, in=60, looseness=12}]
            \node[main] (u) {$u$};
            \node[main] (v1) [right = of u] {$v_1$};
            \node[main] (v2) [right = of v1] {$v_2$};
            \node[main] (vk) [right = of v2] {$v_k$};
            \node[main] (v) [right = of vk] {$v$};

            \path   (u) -- node[pos=1] (uu) {} +(0, -2em)
                    (vk) -- node[pos=1] (vkvk) {} +(0, -2em);

            \path[->]   (u) edge[red] node[fill=white] {$\cdots$} (v1)
                        (v1) edge[red] node[fill=white] {$\cdots$} (v2) edge[lp] node[fill=white] {$\cdots$} node[above] {$\ell_1$} (v1)
                        (v2) edge[red] node[fill=white] {$\cdots$} (vk) edge[lp] node[fill=white] {$\cdots$} node[above] {$\ell_2$} (v2)
                        (vk) edge[blue] node[fill=white] {$\cdots$} node[above] {$\beta$} (v) edge[lp] node[fill=white] {$\cdots$} node[above] {$\ell_k$} (vk)
                        (uu) edge[red] node[below] {$\alpha$} (vkvk);
        \end{tikzpicture}
    \end{center}
    Let $\alpha$ be the sum of the lengths of the edges indicated in red and $\beta$ be the length of the blue path.
    
    An elementary number theory argument tells us that there are $a_1,\cdots,a_k\in\Z$ such that:
    \[ a_1\ell_1 + a_2\ell_2 + \cdots + a_{k-1}\ell_{k-1} + a_k\ell_k = g \]
    However, some of $a_1,\cdots,a_k$ may be negative, which will annoy us later on. Fortunately, we can perform a trick; we can add $\ell_k>0$ to $a_1$ and balance it out by subtracting $\ell_1>0$ from $a_k$:
    \[ (a_1+\ell_k)\ell_1 + a_2\ell_2 + \cdots + a_{k-1}\ell_{k-1} + (a_k-\ell_1)\ell_k = g \]
    Expanding this gives us $+\ell_k\ell_1-\ell_1\ell_k$ which cancels out. By doing this multiple times, and to all values $a_1,\cdots,a_{k-1}$ which are negative, we can `pool' all of the negativity into $a_k$, ensuring that $a_1,\cdots,a_{k-1}\geq 0$, and that $a_k$ is the only possible negative value.
    
    The key idea is that we can head straight from $u$ to $v$ in $\alpha+\beta$ steps, but we can also loop around the intermediate cycles a bunch of times to modify our parity. Since $a_1\ell_1+\cdots+a_k\ell_k=g$, this adds $g$ to our journey, and thus allows us to reach $v$ in $\alpha+\beta$ steps, $\alpha+\beta+g$ steps, $\alpha+\beta+2g$ steps and so on. The catch is that $a_k$ is negative, so we cannot `go $a_k$ times around cycle $k$' -- this is why we made $a_1,\cdots,a_{k-1}$ non-negative. To fix this, rather than start off with $\alpha+\beta$, we start off with $\alpha+\beta+K\ell_k$ for a large value of $K$. To increment by $g$, we increase the number of trips around cycles $1,\cdots,k-1$ and decrease the number of trips around cycle $k$, and since $K$ is large, this takes a while to dip below zero. In fact, we can get this to work $\ell_1/g$ times, and past that we can simply take our existing solutions and go around cycle $1$ an extra time. Let's now show this formally.
    
    It turns out the following value of $K$ is enough:
    \[ K \coloneqq \left(\frac{\ell_1}{g}-1\right)(-a_k) \]
    Note that $g$ divides $\ell_1$ and that $a_k\leq 0$, so $K$ is a non-negative integer.
    
    Now, since $v\in[u]$ we know $uPv$, so $g$ divides the length of any path from $u$ to $v$. This means $\alpha+\beta$ is a multiple of $g$. Now, observe that by going around $K$ times around cycle $k$, and not going around the other cycles, we get a path from $u$ to $v$ of length $\alpha+\beta+K\ell_k$. More generally, for all $0\leq j<\ell_1/g$, let's go $ja_i$ times around cycle $i$ for $1\leq i\leq k-1$ and $K+ja_k$ times around cycle $k$. First, we will make sure we are not going around a cycle a negative number of times, which is not allowed. We made $a_1,\cdots,a_{k-1}\geq 0$, so the only cycle which could be a problem is cycle $k$, which we go around $K+ja_k$ times. However, recall that $K\coloneqq(\ell_1/g-1)(-a_k)$, so:
    \begin{align*}
        K - \left(\frac{\ell_1}{g}-1\right)(-a_k) &\geq 0 \\
        K - j(-a_k) &\geq 0 \\
        K + ja_k &\geq 0
    \end{align*}
    Therefore, cycle $k$ is not a problem either. Now, observe what happens if we enact this plan for $0\leq j<\ell_1/g$:
    \begin{align*}
        \alpha+\beta+\left(\sum_{i=1}^{k-1}ja_i\ell_i\right)+(K+ja_k)\ell_k &= \alpha+\beta+j\left(\sum_{i=1}^ka_i\ell_i\right)+K\ell_k \\
        &= \alpha+\beta+K\ell_k+jg
    \end{align*}
    This gives us $\ell_1/g$ distinct paths from $u$ to $v$ with the following lengths:
    \begin{align*}
        &\alpha+\beta+K\ell_k \\
        &\alpha+\beta+K\ell_k+g \\
        &\alpha+\beta+K\ell_k+2g \\
        &\vdots \\
        &\alpha+\beta+K\ell_k+\left(\frac{\ell_1}{g}-1\right)g
    \end{align*}
    We can show by induction that for all $t\geq 0$ there is a path of this type (i.e.\ one which loops around the cycles like this) from $u$ to $v$ with length $\alpha+\beta+K\ell_k+gt$. We just showed the base case, which was for $0\leq t<\ell_1/g$. Suppose $t\geq\ell_1/g$. Then we can simply take the path for $t-\ell_1/g\geq 0$ and go around cycle $1$ an extra time. The length of this new path is:
    \[ \alpha+\beta+K\ell_k+g\left(t-\frac{\ell_1}{g}\right) + \ell_1 = \alpha+\beta+K\ell_k+gt \]
    Therefore, the result is true for all $t\geq 0$ by induction. Thus, if we let
    \[ T_{u,v} \coloneqq \left\lceil\frac{\alpha+\beta+K\ell_k}{g}\right\rceil \]
    then we have succeeded in showing that we can eventually reach $v$ from $u$ in $gt$ steps for all $t\geq T_{u,v}$. However, we want to show it for a $T$ independent of both $u$ and $v$. Fortunately, $|V|$ is finite, so there are only finitely many $u\in V$ and only finitely many $v\in[u]$; thus, we can define:
    \[ T \coloneqq \max_{u\in V}\left(\max_{v\in[u]}T_{u,v}\right) \]
    Then for all $t\geq T$ we know $t\geq T_{u,v}$, and thus $v\in[u]$ implies $v\in\Gamma^{gt}(u)$ for all $u\in V$ and $v\in[u]$. This means $[u]\subseteq\Gamma^{gt}(u)$ for all $t\geq T$ as desired.
\end{proof}
Fortunately, as a reward for proving that difficult lemma, the simplfied version of a key result for this chapter follows reasonably effortlessly:
\begin{prop}\label{maximumSpreadsWeak}
    Let $G=(V,E)$ be a strongly connected directed graph. Let $g$ be the GCD of all simple cycles as defined in \Cref{pDividesG}, let $P$ be the equivalence relation defined in \Cref{pDefn}, and let its equivalence classes be $[u]$ for $u\in V$. Then there exists $T\geq 0$ such that
    \[ f_{gt}(u)=\max_{v\in[u]}f_0(v) \]
    for all $t\geq T$ and $u\in V$.
\end{prop}
\begin{proof}
    By \Cref{hardLem} there exists $T\geq 0$ such that $\Gamma^{gt}(u)=[u]$ for all $t\geq T$. Then by \Cref{inducto}:
    \begin{align*}
        f_{0+gt}(u) &= \max_{v\in\Gamma^{gt}(u)}f_0(v) \\
        f_{gt}(u) &= \max_{v\in[u]}f_0(v)
    \end{align*}
\end{proof}
This result is powerful because it grants us the values of the entire graph at states which are multiples of $g$ after convergence. Intuitively, we know what the states between these multiples of $g$ after convergence look like too; the values simply cycle around forever, before aligning with the equivalence classes again at the next multiple of $g$. The final few results of this section formalise this notion and clean up the corresponding loose ends.
\begin{center}
    \begin{tikzpicture}[bigg/.style = {scale = 1.2}]
        \def\dist{6.5em};
        \node[main,bigg,redd] (1) at (120:\dist) {5};
        \node[main,bigg,bluu] (2) at (60:\dist) {2};
        \node[main,bigg,grnn] (3) at (0:\dist) {6};
        \node[main,bigg,redd] (4) at (-60:\dist) {5};
        \node[main,bigg,bluu] (5) at (-120:\dist) {2};
        \node[main,bigg,grnn] (6) at (180:\dist) {6};
        \path (5) -- node[pos=1,main,bigg,redd] (7) {5} +(-\dist, 0);
        \path[->]   (1) edge[redd,bend left = 20] (2)
                    (2) edge[bluu,bend left = 20] (3)
                    (3) edge[grnn,bend left = 20] (4)
                    (4) edge[redd,bend left = 20] (5)
                    (5) edge[bluu,bend left = 20] (6)
                    (6) edge[grnn,bend left = 20] (1)
                        edge[grnn,bend right = 20] (7)
                    (7) edge[redd,bend right = 20] (5);
        \node[scale=4.5] {$\circlearrowright$};
        \node[scale=1.8] {\textbf{6}};
        \coordinate (A) at (-\dist-0.6em, -\dist/1.732-0.6em);
        \node[scale=3.25] at (A) {$\circlearrowleft$};
        \node[scale=1.3] at (A) {\textbf{3}};
    \end{tikzpicture}
\end{center}
First, we know that the blue equivalence class follows the red one, the green one follows the blue one, and the red one follows the green one. Let's formalise this by introducing a \emph{successor} function $S$, which takes in an equivalence class and outputs the next one along the chain. In this example $S(\text{red})=\text{blue}$, $S(\text{blue})=\text{green}$ and $S(\text{green})=\text{red}$. We can extend this for $r\geq 0$ to get a successor function $S^r$ which looks $r$ steps ahead in the chain; for example, $S^2(\text{red})=\text{green}$ and $S^6(\text{blue})=\text{blue}$. This will build off the alternate definition of $[u]$ from \Cref{bigcupEquiv}.
\begin{defn}[Successor function]\label{defn:succ}
    Let $G=(V,E)$ be a strongly connected directed graph, let $g$ be the GCD of all simple cycles as defined in \Cref{pDividesG} and let $u\in V$ be any node. Then the $r$th successor of $[u]$ is the equivalence class $S^r[u]$, where:
    \[ S^r[u] = \bigcup_{i=0}^\infty\Gamma^{ig+r}(u) \]
\end{defn}
Note that for $r=0$ we get $S^0[u]=[u]$ by \Cref{bigcupEquiv}. For general $r\geq 0$, it is not immediately clear that this will be equal to one of the equivalence classes of $V$, i.e.\ $[v]$ for some vertex $v\in V$. Fortunately, we can show this.
\begin{prop}\label{srIsEquivalenceClass}
    Let $u\in V$ be any vertex. Then $S^r[u]=[v]$ for all $v\in\Gamma^r(u)$.
\end{prop}
\begin{proof}
    Let $v\in\Gamma^r(u)$ be arbitrary. Then there is a path of length $r$ from $u$ to $v$. Expanding out with \Cref{defn:succ}, we need to show:
    \[ \bigcup_{i=0}^\infty\Gamma^{ig+r}(u) = [v] \]
    $(\supseteq)$ \\
    Let $w\in V$ and suppose $vPw$. Then $g$ divides the length of a path from $v$ to $w$; let the length be $ig$ for $i\geq 0$. Then there is a path of length $r+ig$ from $u$ to $w$ through $v$, so $w\in\Gamma^{ig+r}(u)$.
    \begin{center}
        \begin{tikzpicture}[node distance = 5em, bigg/.style={minimum size = 2.3em}]
            \node[main,bigg] (u) {$u$};
            \node[main,bigg] (v) [right = of u] {$v$};
            \node[main,bigg] (w) [right = of v] {$w$};

            \path[->]   (u) edge node[fill=white] {$\cdots$} node[above] {$r$} (v)
                        (v) edge node[fill=white] {$\cdots$} node[above] {$ig$} (w);
        \end{tikzpicture}
    \end{center}
    $(\subseteq)$ \\
    Let $w\in\Gamma^{ig+r}(u)$ for some $i\geq 0$, so there is a path of length $ig+r$ from $u$ to $w$. Since $G$ is strongly connected, there is a path from $v$ to $u$ of length $\ell$. Then there is a cycle from $u$ to itself via $v$ of length $r+\ell$, meaning $g$ divides $r+\ell$ by \Cref{decomposeSimple}. But there is a path of length $\ell+ig+r$ from $v$ to $w$ through $u$, and $g$ divides $ig+r+\ell$, so $vRw$.
    \begin{center}
        \begin{tikzpicture}[node distance = 5em, bigg/.style={minimum size = 2.3em}]
            \node[main,bigg] (u) {$u$};
            \node[main,bigg] (v) [right = of u] {$v$};
            \node[main,bigg] (w) [below = of u] {$w$};

            \path[->]   (u) edge[bend left=20] node[fill=white] {$\cdots$} node[above] {$r$} (v)
                            edge node[fill=white, sloped] {$\cdots$} node[right] {$ig+r$} (w)
                        (v) edge[bend left=20] node[fill=white] {$\cdots$} node[below] {$\ell$} (u);
        \end{tikzpicture}
    \end{center}
\end{proof}
Note that there exists $v\in\Gamma^r(u)$ for all $r\geq 0$ since $G$ is strongly connected, since $u$ lies on some cycle. We can traverse this cycle arbitrarily many times to get nodes of any arbitrarily high distance away.

We have been using $r$ as our notation, which due to our earlier analysis using the division algorithm, implicitly states that we can assume $0\leq r<g$. Indeed, writing $S^k$ for $k\geq g$ does not make much sense, as we would loop around the cycle multiple times. Formally, we can say that $S^g[u]=[u]$, or more generally:
\begin{lem}\label{canAddG}
    Let $u\in V$ be any vertex. Then $S^k[u] = S^{k+g}[u]$ for all $k\geq 0$.
\end{lem}
\begin{proof}
    Expanding out the left-hand side:
    \begin{align*}
        S^k[u] &= \bigcup_{i=0}^\infty\Gamma^{ig+k}(u) \\
        &= \Gamma^k(u) \cup \bigcup_{i=1}^\infty\Gamma^{ig+k}(u)
    \end{align*}
    Expanding out the right-hand side:
    \begin{align*}
        S^{k+g}[u] &= \bigcup_{i=0}^\infty\Gamma^{ig+k+g}(u) \\
        &= \bigcup_{i=0}^\infty\Gamma^{(i+1)g+k}(u) \\
        &= \bigcup_{i=1}^\infty\Gamma^{ig+k}(u)
    \end{align*}
    Thus, we need only show the following:
    \[ \Gamma^k(u) \subseteq \bigcup_{i=1}^\infty\Gamma^{ig+k}(u) \]
    Let $v\in\Gamma^k(u)$ be arbitrary, so there is a path from $u$ to $v$ of length $k$. Since $G$ is strongly connected, there is a path from $v$ back to $u$ of length $\ell$. Then there is a cycle from $u$ to itself via $v$ of length $k+\ell$, so $g$ divides $k+\ell$ by \Cref{decomposeSimple}, and thus $k+\ell=ig$ for some $i\geq 1$. But then there is a path from $u$ to $v$ of length $k+\ell+k=ig+k$.
    \begin{center}
        \begin{tikzpicture}[node distance = 5em, bigg/.style={minimum size = 2.3em}]
            \node[main,bigg] (u) {$u$};
            \node[main,bigg] (v) [right = of u] {$v$};

            \path[->]   (u) edge[bend left=20] node[fill=white] {$\cdots$} node[above] {$k$} (v)
                        (v) edge[bend left=20] node[fill=white] {$\cdots$} node[below] {$\ell$} (u);
        \end{tikzpicture}
    \end{center}
\end{proof}

This is a useful result, because it tells us that once we go past the first $g$ classes, we will just loop back to the start again. Inductively, we can reason that we can consider the exponent of $S$ modulo $g$ whenever working with it. The next result finally confirms our intuition that this sequence of equivalence classes separates our graph into distinct partitions.
\begin{lem}\label{sPartitions}
    Let $u\in V$ be any vertex. Then $[u]$, $S[u]$, $S^2[u]$, $\cdots$, $S^{g-1}[u]$ partitions $V$.
\end{lem}
\begin{proof}
    We first show that the above sets cover $V$, then show that they are disjoint.

    Let $v\in V$. Since $G$ is strongly connected, there is a path from $u$ to $v$. Using the division algorithm, we can write its length as $kg+r$ for some $k\geq 0$ and $0\leq r<g$. Then $v\in\Gamma^{kg+r}(u)$, so $v\in S^{kg+r}[u]$. By repeatedly applying \Cref{canAddG} we can show that:
    \[ S^r[u]=S^{r+g}[u]=S^{r+2g}[u]=\cdots=S^{kg+r}[u] \]
    Thus, $v\in S^r[u]$ for some $0\leq r<g$, so the sets cover $V$.

    Now, suppose we have $v\in V$ such that $v\in S^k[u]$ and $v\in S^\ell[u]$ for $0\leq k,\ell<g$. Then there are paths from $u$ to $v$ of lengths $ig+k$ and $jg+\ell$. Since $G$ is strongly connected, there is a path from $v$ back to $u$ of length $m\geq 0$. This can be used to form cycles from $u$ to itself through $v$ of lengths $ig+k+m$ and $jg+\ell+m$, so $g$ divides $ig+k+m$ and $jg+\ell+m$ by \Cref{decomposeSimple}, and therefore divides $k+m$ and $\ell+m$. Thus, $g$ divides their difference $(k+m)-(\ell+m)=k-\ell$. But since $-(g-1)\leq k-\ell\leq(g-1)$, this means $k-\ell=0$, so $k=\ell$, and thus the above sets are disjoint.
    \begin{center}
        \begin{tikzpicture}[node distance = 12em]
            \node[main] (u) {$u$};
            \node[main] (v) [right = of u] {$v$};
            \path[->]   (u) edge[bend left = 35] node[fill=white] {$\cdots$} node[above] {$ig+k$} (v) edge node[fill=white] {$\cdots$} node[above] {$jg+\ell$} (v)
                        (v) edge[bend left = 35] node[fill=white] {$\cdots$} node[below] {$m$} (u);
        \end{tikzpicture}
    \end{center}
\end{proof}
With these useful lemmas, we are now ready to conclude the chapter with some key results.
\begin{thm}\label{maximumSpreadsStrong}
    Let $G=(V,E)$ be a strongly connected directed graph. Let $g$ be the GCD of all simple cycles as defined in \Cref{pDividesG}, let $P$ be the equivalence relation defined in \Cref{pDefn}, and let its equivalence classes be $[u]$ for $u\in V$. Let $S^k$ denote the successor function on equivalence classes. Then there exists $T\geq 0$ such that
    \[ f_{gt+r}(u)=\max_{v\in S^r[u]}f_0(v) \]
    for all $t\geq T$, $0\leq r<g$ and $u\in V$.
\end{thm}
\begin{proof}
    By applying \Cref{inducto} and then \Cref{maximumSpreadsWeak}:
    \begin{align*}
        f_{gt+r}(u) &= \max_{v\in\Gamma^r(u)}f_{gt}(v) \\
        &= \max_{v\in\Gamma^r(u)}\left(\max_{w\in[v]}f_0(w)\right)
    \end{align*}
    But $[v]=S^r[u]$ by \Cref{srIsEquivalenceClass} since $v\in\Gamma^r(u)$. Thus:
    \begin{align*}
        f_{gt+r}(u) &= \max_{v\in\Gamma^r(u)}\left(\max_{w\in S^r[u]}f_0(w)\right) \\
        &= \max_{w\in S^r[u]}f_0(w)
    \end{align*}
\end{proof}
An alternative, simpler (but arguably less intuitive) statement is the following:
\begin{thm}\label{unintuitiveMaximumSpreadsStrong}
    There exists $T\geq 0$ such that:
    \[ f_t(u)=\max_{v\in S^t[u]}f_0(v) \]
    for all $t\geq T$ and $u\in V$.
\end{thm}
\begin{proof}
    Using the division algorithm, write $t=gk+r$ for $0\leq r<g$. By repeatedly applying \Cref{canAddG} we can show that $S^{gk+r}[u]=S^r[u]$. Then, by \Cref{maximumSpreadsStrong}:
    \begin{align*}
        f_{gk+r}(u) &= \max_{v\in S^r[u]}f_0(v) \\
        &= \max_{v\in S^{gk+r}[u]}f_0(v)
    \end{align*}
\end{proof}
These results allow us to predict the convergent states of any graph by choosing any $u$ and finding the $g$ equivalence classes $[u]$, $S[u]$, $S^2[u]$, $\cdots$, $S^{g-1}[u]$, which can be achieved using a BFS, similar to the undirected case, and getting the shortest distance $d$ (i.e.\ component $\Lambda_d$) from $u$ for all other vertices $v$. Then, $v\in S^r[u]$ where $d=kg+r$ (i.e.\ we take the distance modulo $g$).

We can get the maxima of these equivalence classes and consider the sequence, then check for repeats in order to find the period. We will show this with some examples; consider the following graph:
\begin{center}
    \begin{tikzpicture}
        \def\dist{6.5em};
        \node[main] (0) {2};
        \node[main] (1) at (105:\dist) {5};
        \node[main] (2) at (75:\dist) {2};
        \node[main] (3) at (45:\dist) {4};
        \node[main] (4) at (15:\dist) {3};
        \node[main] (5) at (-15:\dist) {3};
        \node[main] (6) at (-45:\dist) {1};
        \node[main] (7) at (-75:\dist) {5};
        \node[main] (8) at (-105:\dist) {2};
        \node[main] (9) at (-135:\dist) {4};
        \node[main] (10) at (-165:\dist) {6};
        \node[main] (11) at (165:\dist) {1};
        \node[main] (12) at (135:\dist) {3};
        \foreach \from/\to in {1/2, 2/3, 3/4, 4/5, 5/6, 6/7, 7/8, 8/9, 9/10, 10/11, 11/12, 12/1, 7/0, 0/1}
            \path[->] (\from) edge (\to);
    \end{tikzpicture}
\end{center}
First, notice that there is a simple cycle of length 12 around the outside, and a simple cycle of length 8 on the right side. There are no simple cycles on the left side, so the GCD is 4. Thus, we can let $u$ be the top-left node (with value $5$), and paint every node in its class $[u]$ red. Then, we can paint $S[u]$ blue, $S^2[u]$ green and $S^3[u]$ yellow:
\begin{center}
    \begin{tikzpicture}
        \def\dist{6.5em};
        \node[main,yelw] (0) {2};
        \node[main,redd] (1) at (105:\dist) {5};
        \node[main,bluu] (2) at (75:\dist) {2};
        \node[main,grnn] (3) at (45:\dist) {4};
        \node[main,yelw] (4) at (15:\dist) {3};
        \node[main,redd] (5) at (-15:\dist) {3};
        \node[main,bluu] (6) at (-45:\dist) {1};
        \node[main,grnn] (7) at (-75:\dist) {5};
        \node[main,yelw] (8) at (-105:\dist) {2};
        \node[main,redd] (9) at (-135:\dist) {4};
        \node[main,bluu] (10) at (-165:\dist) {6};
        \node[main,grnn] (11) at (165:\dist) {1};
        \node[main,yelw] (12) at (135:\dist) {3};
        \foreach \from/\to in {1/2, 5/6, 9/10}
            \path[->] (\from) edge[redd] (\to);
        \foreach \from/\to in {2/3, 6/7, 10/11}
            \path[->] (\from) edge[bluu] (\to);
        \foreach \from/\to in {3/4, 7/8, 11/12, 7/0}
            \path[->] (\from) edge[grnn] (\to);
        \foreach \from/\to in {4/5, 8/9, 12/1, 0/1}
            \path[->] (\from) edge[yelw] (\to);
    \end{tikzpicture}
\end{center}
Now we can get the maximum of each equivalence class:
\begin{center}
    \begin{tabular}{cc}
        \toprule
        Class & Maximum \\
        \midrule
        Red & 5 \\
        Blue & 6 \\
        Green & 5 \\
        Yellow & 3 \\
        \bottomrule
    \end{tabular}
\end{center}
Now, \Cref{maximumSpreadsStrong} tells us that, after convergence, at all times which are multiples of 4, all red nodes will be 5, all blue will be 6, all green will be 5 and all yellow will be 3. For times $4t+1$, nodes will instead be the maximum of the class after them in the cycle, so all red will be 6, all blue will be 5, all green will be 3 and all yellow will be 5.
\begin{center}
    \begin{tikzpicture}
        \def\dist{6.5em};
        \node[main,yelw] (0) {3};
        \node[main,redd] (1) at (105:\dist) {5};
        \node[main,bluu] (2) at (75:\dist) {6};
        \node[main,grnn] (3) at (45:\dist) {5};
        \node[main,yelw] (4) at (15:\dist) {3};
        \node[main,redd] (5) at (-15:\dist) {5};
        \node[main,bluu] (6) at (-45:\dist) {6};
        \node[main,grnn] (7) at (-75:\dist) {5};
        \node[main,yelw] (8) at (-105:\dist) {3};
        \node[main,redd] (9) at (-135:\dist) {5};
        \node[main,bluu] (10) at (-165:\dist) {6};
        \node[main,grnn] (11) at (165:\dist) {5};
        \node[main,yelw] (12) at (135:\dist) {3};
        \foreach \from/\to in {1/2, 5/6, 9/10}
            \path[->] (\from) edge[redd] (\to);
        \foreach \from/\to in {2/3, 6/7, 10/11}
            \path[->] (\from) edge[bluu] (\to);
        \foreach \from/\to in {3/4, 7/8, 11/12, 7/0}
            \path[->] (\from) edge[grnn] (\to);
        \foreach \from/\to in {4/5, 8/9, 12/1, 0/1}
            \path[->] (\from) edge[yelw] (\to);
        \node at (0, -\dist-2em) {Time: $4t$};
    \end{tikzpicture}
    \begin{tikzpicture}
        \def\dist{6.5em};
        \node[main,yelw] (0) {5};
        \node[main,redd] (1) at (105:\dist) {6};
        \node[main,bluu] (2) at (75:\dist) {5};
        \node[main,grnn] (3) at (45:\dist) {3};
        \node[main,yelw] (4) at (15:\dist) {5};
        \node[main,redd] (5) at (-15:\dist) {6};
        \node[main,bluu] (6) at (-45:\dist) {5};
        \node[main,grnn] (7) at (-75:\dist) {3};
        \node[main,yelw] (8) at (-105:\dist) {5};
        \node[main,redd] (9) at (-135:\dist) {6};
        \node[main,bluu] (10) at (-165:\dist) {5};
        \node[main,grnn] (11) at (165:\dist) {3};
        \node[main,yelw] (12) at (135:\dist) {5};
        \foreach \from/\to in {1/2, 5/6, 9/10}
            \path[->] (\from) edge[redd] (\to);
        \foreach \from/\to in {2/3, 6/7, 10/11}
            \path[->] (\from) edge[bluu] (\to);
        \foreach \from/\to in {3/4, 7/8, 11/12, 7/0}
            \path[->] (\from) edge[grnn] (\to);
        \foreach \from/\to in {4/5, 8/9, 12/1, 0/1}
            \path[->] (\from) edge[yelw] (\to);
        \path (10) -- node[pos=1] {} +(-\picgap, 0);
        \node at (0, -\dist-2em) {Time: $4t+1$};
    \end{tikzpicture} \\
    \begin{tikzpicture}
        \def\dist{6.5em};
        \node[main,yelw] (0) {6};
        \node[main,redd] (1) at (105:\dist) {5};
        \node[main,bluu] (2) at (75:\dist) {3};
        \node[main,grnn] (3) at (45:\dist) {5};
        \node[main,yelw] (4) at (15:\dist) {6};
        \node[main,redd] (5) at (-15:\dist) {5};
        \node[main,bluu] (6) at (-45:\dist) {3};
        \node[main,grnn] (7) at (-75:\dist) {5};
        \node[main,yelw] (8) at (-105:\dist) {6};
        \node[main,redd] (9) at (-135:\dist) {5};
        \node[main,bluu] (10) at (-165:\dist) {3};
        \node[main,grnn] (11) at (165:\dist) {5};
        \node[main,yelw] (12) at (135:\dist) {6};
        \foreach \from/\to in {1/2, 5/6, 9/10}
            \path[->] (\from) edge[redd] (\to);
        \foreach \from/\to in {2/3, 6/7, 10/11}
            \path[->] (\from) edge[bluu] (\to);
        \foreach \from/\to in {3/4, 7/8, 11/12, 7/0}
            \path[->] (\from) edge[grnn] (\to);
        \foreach \from/\to in {4/5, 8/9, 12/1, 0/1}
            \path[->] (\from) edge[yelw] (\to);
        \path (1) -- node[pos=1] {} +(0, \picgap);
        \node at (0, -\dist-2em) {Time: $4t+2$};
    \end{tikzpicture}
    \begin{tikzpicture}
        \def\dist{6.5em};
        \node[main,yelw] (0) {5};
        \node[main,redd] (1) at (105:\dist) {3};
        \node[main,bluu] (2) at (75:\dist) {5};
        \node[main,grnn] (3) at (45:\dist) {6};
        \node[main,yelw] (4) at (15:\dist) {5};
        \node[main,redd] (5) at (-15:\dist) {3};
        \node[main,bluu] (6) at (-45:\dist) {5};
        \node[main,grnn] (7) at (-75:\dist) {6};
        \node[main,yelw] (8) at (-105:\dist) {5};
        \node[main,redd] (9) at (-135:\dist) {3};
        \node[main,bluu] (10) at (-165:\dist) {5};
        \node[main,grnn] (11) at (165:\dist) {6};
        \node[main,yelw] (12) at (135:\dist) {5};
        \foreach \from/\to in {1/2, 5/6, 9/10}
            \path[->] (\from) edge[redd] (\to);
        \foreach \from/\to in {2/3, 6/7, 10/11}
            \path[->] (\from) edge[bluu] (\to);
        \foreach \from/\to in {3/4, 7/8, 11/12, 7/0}
            \path[->] (\from) edge[grnn] (\to);
        \foreach \from/\to in {4/5, 8/9, 12/1, 0/1}
            \path[->] (\from) edge[yelw] (\to);
        \path (10) -- node[pos=1] {} +(-\picgap, 0);
        \node at (0, -\dist-2em) {Time: $4t+3$};
    \end{tikzpicture}
\end{center}
From this we can gather that the period is 4. If the sequence were 5, 6, 5, 6 instead, then the period would be 2, since cycling this sequence by two gets us back to 5, 6, 5, 6. If the sequence were 5, 5, 5, 5 instead, then the period would be 1 and all values would become 5.

With these ideas in mind, suppose we are given a graph of nodes and edges, but no values for the nodes. If $g$ is the GCD of all simple cycle lengths, then the factors of $g$ are the possible values for the period. Suppose $p$ is a divisor of $g$, and we want to find an initial valuation that eventually ends up in an absorbing state with period $p$. All we need to do is pick any $u\in V$ and ensure that the sequence in the table is precisely $p$ long, which can be achieved as follows:
\begin{center}
    \begin{tabular}{ccc}
        \toprule
        Class && Maximum \\
        \midrule
        $[u]$ && $0$ \\
        $S[u]$ && $1$ \\
        $\vdots$ && $\vdots$ \\
        $S^{p-1}[u]$ && $p-1$ \\
        $S^p[u]$ && $0$ \\
        $S^{p+1}[u]$ && $1$ \\
        $\vdots$ && $\vdots$ \\
        $S^{2p-1}[u]$ && $p-1$ \\
        &$\vdots$& \\
        &$\vdots$& \\
        $S^{g-p}[u]$ && $0$ \\
        $S^{g-p+1}[u]$ && $1$ \\
        $\vdots$ && $\vdots$ \\
        $S^{g-1}[u]$ && $p-1$ \\
        \bottomrule
    \end{tabular}
\end{center}
To construct this, we could BFS from a node to find the distance $\text{dist}_u$ to all other nodes. Then we can set $f_0(v)=\text{dist}_v\mod p$ for all $v\in V$.
\begin{thm}\label{constructivePeriodP}
    Let $G=(V,E)$ be a strongly connected directed graph, and let $g$ be the GCD of all simple cycle lengths as defined in \Cref{pDividesG}. Then for all $p\geq 1$ that divide $g$, there is an initial valuation $f_0$ that is immediately in an absorbing state, and has period $p$.
\end{thm}
\begin{proof}
    Choose any $u\in V$. Then $[u],S[u],\cdots,S^{g-1}[u]$ partition $V$ by \Cref{sPartitions}, so for all $v\in V$ there is a unique $0\leq k<g$ such that $v\in S^k[g]$; let $f_0(v)=k\mod p$. By \Cref{inducto} we know $f_p(v)=\max_{w\in\Gamma^p(v)}f_0(w)$. Since there is a path from $u$ to $v$ of length $k$ and a path from $v$ to $w$ of length $p$, there is a path from $u$ to $w$ of length $k+p$, and thus $w\in S^{k+p}[u]$. Even if we have to subtract $g$ from $k+p$ to get it into the range $[0,g)$, it will still become $k\mod p$ when we take it modulo $p$, since $p$ divides $g$. Thus, $f_p(v)=k\mod g=f_0(v)$. Since this is true for all vertices, $f_p=f_0$ as valuations, meaning the period is at most $p$.
    
    Now, suppose $f_k=f_0$ for some $1\leq k<p\leq g$. Then $f_k(u)=f_0(u)=0$ (since $u\in[u]$), so there exists $v\in\Gamma^k(u)$ such that $f_0(v)=0$. Such a $v$ is in $S^k[u]$ for the partition, since $k<g$. But since $f_0(v)=0$, we know $k\mod p=0$, and thus $p$ divides $k$. This means $p$ is the smallest integer $k\geq 1$ that satisfies $f_k=f_0$, and thus $p$ is the period.
\end{proof}
Thus, we have found that the period must divide the GCD of all simple cycle lengths, and built examples of initial valuations that achieve all such periods. This was the main challenge for this chapter; in the next section we briefly touch upon the case of general directed graphs, then analyse this model through various experiments.

\section{General Directed Graphs}
Now that we have seen what happens for strongly connected graphs, consider the process of dividing a general graph into its SCCs and forming a topologically sorted meta-graph, described in \Cref{sec:SCCs}. 
\begin{center}
    \begin{tikzpicture}
        \node[main, scclabel] (A) {\huge\textbf{A}};
        \node[main, scclabel] (B) [right = of A] {\huge\textbf{B}};
        \node[main, scclabel] (C) [right = of B] {\huge\textbf{C}};
        \node[main, scclabel] (D) [right = of C] {\huge\textbf{D}};

        \path[->, scclabel] (A) edge (B)
                            (B) edge (C) edge[bend left = 35] (D)
                            (C) edge (D);
    \end{tikzpicture}
\end{center}
Suppose a suffix of the SCCs have each converged into predictable states with known periods. Then it may be possible to show that the rightmost SCC which has not yet converged will also converge, since the behaviour of all SCCs it points to are known. However, the period of this SCC may depend on the LCM of its children, and may thus be exponential. In fact, we can construct an example where the local period of a particular node is exponentially large in the number of nodes. To do this, we will have one node point to all other SCCs.
\begin{center}
    \begin{tikzpicture}
        \node[main, scclabel] (A) {\huge\textbf{A}};
        \node[main, scclabel] (B) [right = of A] {\huge\textbf{B}};
        \node[main, scclabel] (C) [right = of B] {\huge\textbf{C}};
        \node[main, scclabel] (D) [right = of C] {\huge\textbf{D}};

        \path[->, scclabel] (A) edge (B) edge[bend left=35] (C) edge[out=45] (D);
    \end{tikzpicture}
\end{center}
We can make the LCM of the periods of these SCCs as large as possible by having the periods be the first $k$ prime numbers. Thus, for $k\geq 2$, if we let $p_1,\cdots,p_k$ be the first $k$ prime numbers, so notably $p_1=2$ and $p_2=3$, then we get the following construction:
\begin{center}
    \begin{tikzpicture}
        \node[main] (u) {$p_k$};
        \node[main] (v11) at (12em, 5.5em) {$0$};
        \node[main] (v12) [right = of v11] {$2$};
        \node[main] (v21) at (12em, 2em) {$0$};
        \node[main] (v22) [right = of v21] {$0$};
        \node[main] (v23) [right = of v22] {$3$};
        \node       (v31) at (12em, -1.5em) {$\cdots$};
        \node[main] (vk1) at (12em, -5.5em) {$0$};
        \node[main] (vk2) [right = of vk1] {$0$};
        \node (vk3) [right = of vk2] {$\cdots$};
        \node[main] (vk4) [right = of vk3] {$p_k$};

        \def\bend{25}
        \path[->]   (u) edge (v11) edge (v21) edge (vk1) edge (v31)
                    (v11) edge (v12)
                    (v12) edge[bend right=\bend] (v11)
                    (v21) edge (v22)
                    (v22) edge (v23)
                    (v23) edge[bend right=\bend] (v21)
                    (vk1) edge (vk2) -- node(vk2p5) {} (vk4)
                    (vk2) edge (vk3)
                    (vk3) edge (vk4)
                    (vk4) edge[bend right=\bend] (vk1);
        
        \path (vk2p5)   -- node[pos=1] {$\vdots$} +(0, 4.5em)
                        -- node[pos=1] {$\underbrace{\phantom{wwwwwwwwwwwwwwwwwwwwww}}$} ++(0, -1.5em) -- node[pos=1] {$p_k$} +(0, -1.5em);
    \end{tikzpicture}
\end{center}
Let $u$ be the node with no incoming edges on the far left of this graph, and recall from \Cref{inducto} that $f_t(u)=\max_{v\in\Gamma^t(u)}f_0(v)$. The key feature of this graph reveals itself when considering $\Gamma^t(u)$ for various values of $t$. Any path from $u$ must choose an SCC to move to, and once this choice is made, there is only one possibility for how the path develops.
\begin{lem}\label{piEqualsPi}
    For $k\geq 2$, let $G$ be the above graph, and let $f_0$ be the above valuation. Then $f_{p_i}(u)=p_i$ for all $1\leq i\leq k$.
\end{lem}
\begin{proof}
    By \Cref{inducto}, we know $f_{p_i}(u)$ is the maximum label in the initial valuation that can be reached in $p_i$ steps. It is possible to reach $p_i$ in $p_i$ steps by construction of the graph. The only values greater than $p_i$ are $p_j$ for $j>i$, but these are unreachable in only $p_i$ steps.
\end{proof}
With this lemma, we can prove that the local period of $u$ is exponential in the number of nodes.
\begin{thm}\label{ahadExample}
    For $k\geq 2$, let $G$ be the above graph, and let $f_0$ be the above valuation. Then the local period of $u$ is $P\coloneqq p_1p_2\cdots p_k$.
\end{thm}
\begin{proof}
    First, we show that $f_{P+t}(u)=f_t(u)$ for all $t\geq 0$. When $t=0$, by going to $p_k$ in $p_k$ steps, then taking the loop of size $p_k$ another $p_1\cdots p_{k-1} - 1$ times, we arrive at $p_k$ with a total path length of $P$. By \Cref{inducto}, we get that $f_P(u)=p_k=f_0(u)$. To show that $f_{P+t}(u)=f_t(u)$ for $t>0$, we can write it using \Cref{inducto}:
    \[ \max_{v\in\Gamma^{P+t}(u)}f_0(v) = \max_{v\in\Gamma^t(u)}f_0(v) \]
    Thus, it suffices to show that $\Gamma^{P+t}(u) = \Gamma^t(u)$.

    $(\subseteq)$ \\
    Let $v\in\Gamma^t(u)$, so there is a path of length $t$ from $u$ to $v$. Since $t>0$, this path enters an SCC, which is a cycle of size $p_i$. Since $P$ is a multiple of $p_i$, we can loop around this cycle more to get a path to $v$ of length $P+t$, and thus $v\in\Gamma^{P+t}(u)$.

    $(\supseteq)$ \\
    Let $v\in\Gamma^{P+t}(u)$, so there is a path of length $P+t$ from $u$ to $v$. Since $t>0$, such a path reaches $v$ after only $t$ steps, then loops around some number of times. By not looping around, we get a path of length $t$ from $u$ to $v$, and thus $v\in\Gamma^t(u)$.

    This shows that the local period is at least $P$; we will show that the local period is at most $P$. Suppose the local period is $\ell$. Then by \Cref{piEqualsPi}, we know $f_{\ell+p_i}(u)=f_{p_i}(u)=p_i$ for all $1\leq i\leq k$. Then by \Cref{inducto} there is a path of length $\ell+p_i$ from $u$ to a node $v$ such that $f_0(v)=p_i$. But since $\ell+p_i>0$, we cannot have $v\neq u$, and thus $v$ is uniquely determined. The only way for such a path to exist is if $\ell+p_i$ is a multiple of $p_i$, meaning $p_i$ divides $\ell$. Since this is true for all $1\leq i\leq k$, and $p_1,\cdots,p_k$ are pairwise coprime, we conclude that $P=p_1\cdots p_k$ divides $\ell$ as desired.
\end{proof}
Since the local period of a particular node is $P$, the global period of the entire graph must be at least $P$, as the global period is the least value $\ell$ for which $f_\ell(v)=f_0(v)$ for all $v\in V$, whereas local periods need only hold for a fixed $v$. In fact, a similar argument can show that the global period is also $P$.

The number of nodes in the above example is upper bounded by $k^3$, since the $k$th prime is at most $k^2-1$ for $k\geq 2$ by the prime number theorem, which says that the bound grows asymptotically to $k\log k$. Then since the smallest prime is $2$, the period $P\geq 2^k$, so the result is exponential. Thus, the dominating factor in determining the period of a graph is the relationship between the periods of subsequent components, which can be exponential in the worst case.

For this graph, the convergence time is actually zero -- the initial state presented above is eventually reached again. It remains an open problem whether it is possible to construct an example with exponential convergence time, but one possibility might be to have two such exponential processes, one of which eventually overtakes the other, permanently altering the state of the graph. It's not immediately clear how this would be possible, so it may also be possible to show that the convergence time cannot be exponential, even for graphs which are not strongly connected; this would be an interesting topic for future research.

\section{Experimental Analysis}
Prior to this section, we found the convergent states of the graph in both the directed and undirected cases. This led to interesting theoretical results involving equivalence classes and colouring, but in practice, almost all social networks or randomly generated graphs will converge with period 1. This is because in social networks, there are almost certainly friends mutually following each other, or the graph is undirected; either way, there is almost certainly a cycle of length 2. Then it is highly unlikely that the graph is bipartite, since there is almost certainly a group of friends all following each other (creating a 3-clique). In randomly generated graphs with a large number of nodes, these structures almost always come up at least once, since the probability that they do not is exponential in the number of nodes. 

Thus, there is almost always an even and an odd cycle. Since the GCD of simple cycles is almost certainly 1, assuming the graph is strongly connected, \Cref{unintuitiveMaximumSpreadsStrong} tells us that the process takes the form of the maximum value spreading to the rest of the graph, and eventually, the entire graph will be equal to this value. Rather than test that the randomly generated graphs are strongly connected, it is easier to simply check if $f_t$ is a valuation that has already been seen before; the first occurrence of this is the convergence time, the only thing we need to measure in this section. Since the process is deterministic, the convergence time is fixed, no matter how many times the process is run. This means we do not have to worry about the variance within a particular graph and valuation, only the random generation of graphs and valuations itself.

\subsection{\ER graphs and social networks}
The results discussed above mean that we can consider a `worst case initial valuation', where there is a unique maximum and it must spread to all other nodes. For simplicity, we can assume all other nodes have initial value 0 and there is a unique node with maximum 1. We also test random valuations where the values are uniformly randomly chosen from $-b$ to $b$ for some $b$.

The below experiments were run for 100 trials each.
\begin{mytable}
    \centering
    \begin{tabular}{clcccccc}
        \toprule
        \multicolumn{2}{c}{\multirow{2}{*}{Graph}} & \multirow{2}{*}{\shortstack{Average \\ Diameter}\vspace{-0.5em}} & \multirow{2}{*}{\shortstack{Unique \\ Maximum}\vspace{-0.5em}} & \multicolumn{4}{c}{Range ($b$)} \\
        \cmidrule{5-8}
        &&&& 10 & 100 & 1\,000 & 10\,000 \\
        \midrule
        \multirow{4}{*}{\rotatebox[origin=c]{90}{\Erdos}} & 10 & 2.92 & 2.56 & 2.45 & 2.46 & 2.61 & 2.52 \\
        & 100 & 2.00 & 2.00 & 1.75 & 2.00 & 2.00 & 2.00 \\
        & 1\,000 & 2.00 & 2.00 & 1.00 & 1.97 & 2.00 & 2.00 \\
        & 10\,000 & 2.00 & 2.00 & 1.00 & 1.00 & 2.00 & 2.00 \\
        \multirow{4}{*}{\rotatebox[origin=c]{90}{\;\Barabasi}} & 10 & 2.00 & 1.00 & 1.00 & 1.00 & 1.00 & 1.00 \\
        & 100 & 3.00 & 2.86 & 2.12 & 2.69 & 2.79 & 2.80 \\
        & 1\,000 & 4.00 & 3.52 & 2.02 & 3.01 & 3.37 & 3.47 \\
        & 10\,000 & 5.00 & 4.04 & 2.14 & 3.00 & 3.86 & 3.99 \\
        \multicolumn{2}{c}{Twitter*} & 7 & 12.08 & 8.69 & 9.69 & 10.38 & 11.45 \\
        \multicolumn{2}{c}{Wikipedia*} & 6 & 7.74 & 6.01 & 6.70 & 7.51 & 7.71 \\
        \multicolumn{2}{c}{Facebook} & 8 & 6.41 & 2.15 & 4.03 & 5.74 & 6.37 \\
        \multicolumn{2}{c}{Twitch} & 10 & 6.98 & 5.00 & 5.64 & 6.41 & 6.93 \\
        \bottomrule
    \end{tabular}
    \caption{Convergence time of various graphs in the synchronous maximum model}
\end{mytable}
Initially, when running experiments, the average convergence time of Wikipedia* with the unique maximum valuation came out to be 2.35, far lower than expected. The vast majority of cases had a convergence time of 1, and sometimes the convergence time was much higher, around 8 or 9. The reason for this lies in the topological structure of Wikipedia*; the largest SCC has 1\,300 nodes, which only accounts for a quarter of the nodes, and this SCC is an absorbing component, i.e.\ it has no exit edges. Furthermore, 3\,849 of the 3\,858 nodes not in the largest SCC had no incoming edges, meaning that if they were selected to be the unique maximum, it would fizzle out in one step and produce a convergence time of 1. To fix this, the unique maximum was set to only occur in the largest SCC. With this change, the unique maximum does appear to be worse than all randomly generated valuations, aside from the small case of the \ER graph of 10 nodes, but this is likely due to randomness.

The \BA graphs had a slightly higher convergence time than the \ER graphs, but this is expected since they have a higher diameter. Since \BA graphs are undirected, we actually know the convergence time is strongly related to the diameter (\Cref{undPeriod1or2}). We did not show a similar result for the directed case; the challenging proof of \Cref{hardLem} was constructive and provided an upper bound on the convergence time, and with care, this upper bound can be reduced to be in the order of the LCM of all simple cycle sizes. However, this bound is usually enormous compared to the diameter of the graph, and the results of the experiments are closer to the diameter than the LCM, so this is omitted in the theoretical analysis. This means it may be possible to extend the proof techniques that showed the $O(D)$ bound in the undirected case to show the same bound for the directed case.

Diameter is not the sole factor affecting the convergence time however, since the Twitch graph has a higher diameter than Twitter*, but a lower convergence time. The Twitch graph has a smaller number of nodes and edges, but still a high diameter between the two that are furthest apart, whereas the Twitter graph is much larger in both nodes and edges, but more interconnected. Counter-intuitively, it appears that this more interconnected graph has a larger convergence time.

Higher values of $b$ corresponded to larger ranges, which increase the probability of the maximum being unique (in the case of Wikipedia*, we are talking about the maximum of the largest SCC, rather than the overall maximum, as this is what all values converge to). As this probability increases, the average convergence time gets closer and closer to the worst case where there is a unique maximum.

For \ER graphs, the probability that the diameter is not 2 decreases exponentially to zero, meaning with just 100 nodes, the diameter is almost certainly 2. Not only this, but it is also almost always the case that any node $u$ is a part of some diameter, meaning from any node $u$ there exists a node $v$ which is a distance of 2 apart. If this were not the case, then $u$ would be directly connected to all its neighbours, which is exponentially unlikely, and even if such a node somehow existed, there is even a chance that it does not have the unique maximum. Thus, the convergence time is almost always 2 for such graphs. The only mitigating factor would be if $b$ is too small, so there are a large enough number of values already equal to the maximum that they can spread within just one step, in which case it is almost certain that the convergence time will be 1.

\subsection{Asynchronous model}
This section analysed the synchronous maximum model through a theoretical and experimental lens, but without much change to the code, we can also analyse the asynchronous maximum model. This is a model with its own complex set of behaviours, which could prove to be a topic of interest for future research. Each iteration, a random node is selected, and its value is set to the maximum of its neighbours, overriding its own value. There are two key behaviours of this model:
\begin{itemize}
    \item The maximum may disappear. For example, if there is a unique maximum and it is chosen to update, then it will be overwritten.
    \item For strongly connected graphs, the entire graph eventually assumes the same value. In general, each SCC eventually assumes the same value. In either case, the period is 1.
\end{itemize}
Below is a quick proof sketch as to why the second behaviour holds, following the structure of \Cref{chap:load-balancing} and \Cref{chap:max-model}.

The argument in \Cref{markovFiniteSM} also applies in the asynchronous case, so the Markov chain is finite. Then clearly the valuations where all values are identical are absorbing states of size 1, since no further updates can be made from these states. We can show that no matter our valuation and its SCC, there is always a way to exit this SCC and make all values equal.

First, suppose $G$ is strongly connected and has vertices $v_1,\cdots,v_n$. WLOG, suppose $v_n$ currently has the maximum value. Then there is a path from $v_1$ to $v_2$, then to $v_3$ and so on up to $v_n$ (which may repeat vertices); call this path $w_1,\cdots,w_k$. If we update in the sequence $w_{k-1},w_{k-2},\cdots,w_1$, then we can show inductively that all values become maximum. This is because $w_k=v_n$ currently has the maximum value, and on each step, we make $w_i$ take the maximum from $w_{i+1}$ (even if it is already the maximum). Then, if $G$ is a general graph, we can perform this process for all SCCs in topological order from the absorbing states to the topmost ones, but for each SCC, we first update all nodes with exit edges to other SCCs to introduce those values.

Since Twitter* and Wikipedia* are not strongly connected, we keep track of which nodes are equal to the maximum of their neighbours with an array. When no nodes can update, the state has converged. For the dense \ER graphs, this process is slow, but since they are almost certainly strongly connected (especially for large $n$), we can instead check if all values are the same, since this occurs at the same moment. We also keep track of the final value achieved as a percentage of the original maximum. This is unique for all graphs except Twitter* (including Wikipedia*, since only one SCC is non-trivial, so all others will take its value). For Twitter*, we take the final value of its largest SCC. The experiment ran with the valuation being a random shuffling of the vertex labels $1,\cdots,n$. This is the worst case (without factoring in where to place the worst case values), since the vertex labels are all different, so no two can combine and spread throughout the graph; in other words, when a value spreads to the entire graph, we know it must have originated from one vertex only, and there was no way to stop this from happening.

The below experiments were run for 100 trials.
\begin{mytable}
    \centering
    \begin{tabular}{clcccc}
        \toprule
        \multicolumn{2}{c}{\multirow{2}{*}{Graph}} & \multirow{2}{*}{\shortstack{Average\\Diameter}\vspace{-0.3em}} & \multirow{2}{*}{\shortstack{Conv\vspace{0.2em}\\Time}\vspace{-0.3em}} & \multirow{2}{*}{\shortstack{Final\vspace{0.2em}\\Value (\%)}\vspace{-0.3em}} \\
        \\
        \midrule
        \multirow{4}{*}{\rotatebox[origin=c]{90}{\Erdos}} & 10 & 2.92 & 30.38 & 97.00 \\
        & 100 & 2.00 & 530.32 & 99.97 \\
        & 1\,000 & 2.00 & 7\,396.43 & 100.00 \\
        & 10\,000 & 2.00 & 99\,389.17 & 100.00 \\
        \multirow{4}{*}{\rotatebox[origin=c]{90}{\;\Barabasi}} & 10 & 2.00 & 31.95 & 87.40 \\
        & 100 & 3.00 & 545.84 & 99.91 \\
        & 1\,000 & 4.00 & 7\,905.44 & 99.99 \\
        & 10\,000 & 5.00 & 102\,217.08 & 100.00 \\
        \multicolumn{2}{c}{Twitter*} & 7 & 970\,405.99 & 100.00 \\
        \multicolumn{2}{c}{Wikipedia*} & 6 & 59\,238.65 & 99.89 \\
        \multicolumn{2}{c}{Facebook} & 8 & 42\,224.53 & 100.00 \\
        \multicolumn{2}{c}{Twitch} & 10 & 78\,911.45 & 99.99 \\
        \bottomrule
    \end{tabular}
    \caption{Convergence time of various graphs in the asynchronous maximum model}
\end{mytable}
As one might expect, the likelihood that the final value is either the maximum or is close to the maximum increases greatly as the number of nodes increases, since it is much less likely for the maximum to be erased, and much more likely that its value spreads to more nodes initially, making it harder to be erased later on. Also note that the \ER and \BA models seem to yield similar values, despite the fact that the \ER graphs are much more dense.

The convergence times are much larger for this model, mostly due to the fact that only one node can update on each iteration, whereas before they would all update. However, even if we account for this by dividing the convergence times by $n$, notice that the result increases beyond $2D$ for the \ER graphs. For example, we get $7.39643$ for $n=1\,000$ and $9.938917$ for $n=10\,000$, both of which exceed $2D$, since the average diameter is almost always $D=2$. In practice, for social networks, it appears that the convergence time divided by $n$ appears to be $\Theta(D)$, so it is possible that there is a connection.

Since the diameters of the graphs are similar, we can plot points to try to estimate the dependency of the convergence time on the number of nodes $n$.
\begin{myfigure}
    \centering
    \includegraphics[width=0.8\linewidth]{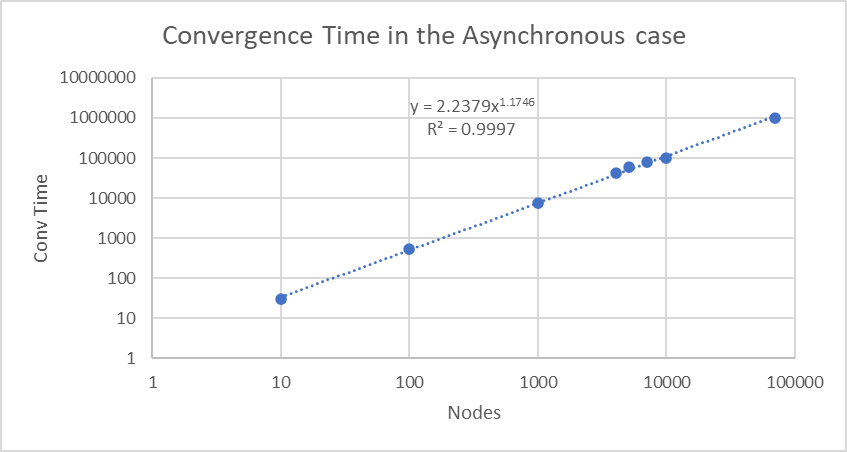}
\end{myfigure}
The relation appears to be linear (the R squared value for a straight line is $R^2=0.9984$), suggesting a theoretical bound that could be looked into in future.                               
\chapter{Concluding Remarks}\label{chap:conclusion}
In this chapter, we conclude our work with a brief recap of the key theorems shown in \cref{chap:load-balancing,chap:max-model}, and a summary of the experimental results. We then outline avenues for future research into this topic, other related models, and related topics.

\section{Conclusion}
\subsection{Load Balancing Model}
In \Cref{chap:load-balancing}, we showed that the sum of the nodes of a graph is invariant under the load balancing model, and that the square sum and maximum are monotonically decreasing. After showing that the Markov chain is finite, this allowed us to prove \Cref{divisionConnected} and \Cref{divisionAny}, which showed that each connected component always converges to a state where all values are one of two adjacent integers; the integers themselves, and how many of each there are, can be determined using the sum invariant and division algorithm. This gives a complete understanding of the unique final state of the graph, which can be checked in code by maintaining the minimum and maximum, and checking if they differ by 1.

We broke the convergence time into a height and width to be analysed separately. We showed that the height depends on the \emph{square sum gap}, which is half the difference between the initial and final square sums, then showed the tightness of this bound by presenting an initial valuation for the complete graph on $2^n$ nodes. This valuation has an initial square sum of $n2^n$, a final square sum of $0$, and an update sequence which takes the maximum amount of $n2^{n-1}$ shrink updates to converge. For the width, we defined the \emph{gambler's ruin valuations} which we believe to be the worst case for the average width, and gave a proof sketch as to why this width is $\Theta(n^3)$. We then gave some informal arguments that could hopefully be extended to show that the expected convergence time of any graph is $\bigO(n^3q)$.

The experiments showed that the maximum and minimum are comparatively quickly brought within 2 of each other, and the vast majority of the convergence time is spent swapping the minimum and maximum values randomly until they happen to collide with each other. We saw that for each $n$, the convergence time of the \ER graph of $n$ nodes counter-intuitively decreases with respect to the square sum gap $q$ up until some optimal value, then increases afterwards as expected. The height (i.e.\ the number of shrink updates) appears to increase linearly with respect to the square sum gap $q$, suggesting that the average number of shrink updates is a constant ratio of the maximum possible number of shrink updates.

Experiments on binomial valuations suggested that the width depended linearly on $n$ and the height depended linearly on $q$, suggesting that the average convergence time was $\Theta(nq)=\Theta(n^2\log n)$. Experiments also reinforced the idea that gambler's ruin valuations have an average convergence time of $\Theta(n^3)$.

\begin{mytable}
    \centering
    \begin{tabular}{ccccc}
        \toprule
        \multirow{2}{*}{Graph/Valuation} & \multicolumn{3}{c}{Average case} & Worst case \\
        \cmidrule(lr){2-4} \cmidrule(lr){5-5}
        & Conv Time & Width & Height & Height \\
        \midrule
        Binomial & $\Theta(nq)$ & $\Theta(n)$ & $\Theta(q)$ & $q$ \\
        Gambler's ruin & $\Theta(n^3)$ & $\Theta(n^3)$ & $1$ & $1$ \\
        \bottomrule
    \end{tabular}

    \caption{Experimental results of \Cref{chap:load-balancing}}
\end{mytable}

\subsection{Synchronous Maximum Model}
In \Cref{chap:max-model}, for undirected graphs, we showed that bipartite graphs spread the maxima of each coloured component to the entire component in at most $D$ steps, and therefore converge with period 2. Also, non-bipartite graphs converge with period 1 in at most $2D$ steps, and spread the global maximum everywhere. For arbitrary undirected graphs, which can be disconnected, we showed that we can split into cases based on the individual periods of each connected component.

For directed graphs, we showed that the local periods and global period are all equal for strongly connected graphs, then used this to show that the global period divides the GCD of all simple cycles $g$. Then, we split the graph into $g$ equivalence classes and showed that for large enough $t$, the value $f_{gt+r}(u)$ (i.e.\ the value of $u$ at time $gt+r$) is the maximum value $f_0(v)$ for all $v\in S^r[u]$, the equivalence class $r$ steps after $u$. We used this to explicitly construct a graph of period $p$ for any $p$ which divides $g$. Finally, for weakly connected graphs, we discussed SCCs and topological sorting, and provided a construction of a graph with an exponential period.

Experiments on undirected graphs confirmed that the convergence time was at most $2D$, but also indicated that directed graphs share the same convergence time. It remains to show via proof that the convergence time of directed graphs is $\Theta(D)$.

We also ran some experiments on the asynchronous model to see if it behaved like the synchronous model, and found that it did not, even if we divide by $n$ to account for the fact that only one node updates on each iteration.

\section{Future Work}
Future research questions can be divided into three categories:
\begin{itemize}
    \item \emph{Unproven results from this thesis.} Experiments indicated that the convergence time of the gambler's ruin valuation was indeed $\Theta(n^3)$ in the worst case, but the proof of $\Omega(n^3)$ is incomplete, and it has not yet been shown that the gambler's ruin valuation is the worst case; finishing off these proofs, or finding simpler ways to show the $\bigO(n^3)$ bound, would solidify these results and potentially yield techniques applicable to future problems.
    
    \item \emph{Research questions stemming from this thesis.} For example, the links between $n$ and $q$ discussed in \Cref{sec:expAnalysis1} could be investigated further; each $n$ appears to have an optimal $q$ for the convergence time on random graphs, and the convergence time for fixed $q$ converges to $q$ for high enough $n$, but this $n$ appears to depend on $q$ also; the reasons for this are currently unknown and could be explored further.

    Another question could be to find examples of graphs that meet certain criteria, or prove that they cannot exist. For example, one could search for a graph and initial valuation such that the synchronous maximum model has exponential convergence time, and possibly even find an example which is strongly connected (although this is unlikely to be possible). It may also be possible to find an example of a graph and valuation that achieves a worst case width of $\Theta(n^3)$ but is more scalable than the gambler's ruin valuation, being dense instead of sparse, or having a higher square sum gap than 2.

    Research could also be put into the asynchronous maximum model, which has so far been unexplored. The expected convergence time appears to be linear, which could be the subject of future investigation. We suspect that the convergence of this process has two stages; first, an entire cycle must obtain the maximum value, and second, this value must propagate to the rest of the graph.
    
    \item \emph{Related questions of interest.} While investigating related questions, a problem came up that is simple to state yet difficult to solve; consider the asynchronous majority model, where a node updates to a value if and only if that value is strictly more prevalent among its neighbours than all other values. For example, if a node has three red neighbours, three blue neighbours and two green neighbours, then it would not update (even if it were green), but if it had four red neighbours, it could update to red. For $k=2$, we can show that the period is 1, since such a state can be reached by performing as many white updates as possible, followed by as many black updates as possible \citep{brill2016pairwise}. This seems to be true for general $k$, but is challenging to prove, since this argument does not appear to generalise.
\end{itemize}                              


\bibliographystyle{anuthesis} 
\cleardoublepage\phantomsection
\bibliography{bib}
\end{document}